\documentclass[11pt]{report}

\usepackage{microtype}
\usepackage{amsfonts,amssymb,latexsym}

\usepackage{amsmath,amsthm}

\usepackage{color}
\usepackage{eucal}

\usepackage{fullpage}

\usepackage{dsfont}

\usepackage{xspace}

\usepackage{titlesec}

\usepackage[numbers]{natbib}

\definecolor{weborange}{rgb}{.8,.3,.3}
\definecolor{webblue}{rgb}{0,0,.8}
\definecolor{internallinkcolor}{rgb}{0,.5,0}
\definecolor{externallinkcolor}{rgb}{0,0,.5}

\usepackage[%
  pagebackref,
  hyperfootnotes=false,
  colorlinks=true,
  urlcolor=externallinkcolor,
  linkcolor=internallinkcolor,
  filecolor=externallinkcolor,
  citecolor=internallinkcolor,
  breaklinks=true,
  pdfstartview=FitH,
  pdfpagelayout=OneColumn]{hyperref}

\usepackage{graphicx}
\usepackage{bbm}
\usepackage{doi}

\theoremstyle{plain}
\newtheorem{theorem}{Theorem}[section]
\newtheorem{lemma}[theorem]{Lemma}
\newtheorem{corollary}[theorem]{Corollary}
\newtheorem{claim}[theorem]{Claim}
\newtheorem{proposition}[theorem]{Proposition}

\theoremstyle{definition}

\newtheorem{definition}[theorem]{Definition}

\newtheorem{remark}[theorem]{Remark}
\newtheorem{question}[theorem]{Question}

\newtheorem{observation}[theorem]{Observation}

\newcommand{\C}{{\mathbb C}}
\newcommand{\R}{{\mathbb R}}

\newcommand{\Z}{{\mathbb Z}}
\newcommand{\N}{{\mathbb N}}

\newcommand{\poly}{\operatorname{poly}}

\newcommand{\eps}{\varepsilon}

\input{Qcircuit}

\DeclareMathOperator{\Exp}{\mathbf{E}}
\renewcommand{\Pr}{\mathbf{P}}

\renewcommand{\bra}[1]{\langle#1|}
\renewcommand{\ket}[1]{|#1\rangle}
\newcommand{\braket}[2]{\langle#1|#2\rangle}
\newcommand{\ketbra}[2]{|#1\rangle\!\langle#2|}
\newcommand{\proj}[1]{\ketbra{#1}{#1}}

\newcommand{\tensor}{\otimes}
\newcommand{\Tr}{\operatorname{Tr}}
\newcommand{\polylog}{\operatorname{polylog}}

\newcommand{\sizee}{\ensuremath{\mathcal{C}_\epsilon}}

\usepackage{tikz}
\usepackage{circuitikz}
\usetikzlibrary{arrows,automata,chains,positioning,matrix,fit,calc}
\newcommand{\NOT}{\operatorname{NOT}}
\newcommand{\NOTNOT}{\operatorname{NOTNOT}}

\newcommand{\CNOT}{\operatorname{CNOT}}
\newcommand{\CNOTNOT}{\operatorname{CNOTNOT}}

\newcommand{\Fredkin}{\operatorname{Fredkin}}
\newcommand{\Toffoli}{\operatorname{Toffoli}}

\titleformat{\chapter}[display]
{\normalfont\Large\filcenter\rmfamily\scshape}
{\titlerule[1pt]%
 \vspace{1pt}%
 \titlerule
 \vspace{1pc}%
 \LARGE{Lecture} \thechapter}
{1ex}
{\large\slshape}

\newcommand{\lecture}[3]{%
  \chapter{#3}%
  \vspace{-5ex}%
  \textit{Lecturer: #1 \hfill Scribe: #2}\par%
  \vspace{1ex}\titlerule\vspace{2ex}}

\begin{document}

\begin{titlepage}
\begin{center}
\begin{huge}
\textsc{Lecture notes for the 28th\\ McGill Invitational Workshop on\\
Computational Complexity\\}
\end{huge}
\vspace{1.5cm}

\begin{Large}
Bellairs Institute\\
Holetown, Barbados

\vspace{2.5cm}
{\bf The Complexity of Quantum States and Transformations: \\ From Quantum Money to Black Holes}

\vspace{2cm}
\emph{Primary Lecturer}:\\
Scott Aaronson

\vspace{1cm}
\emph{Guest Lecturers}:\\
Adam Bouland\\
Luke Schaeffer
\end{Large}
\end{center}
\end{titlepage}

\newpage

\pagenumbering{roman}
\tableofcontents

\newpage

\chapter*{Foreword}
These notes reflect a series of lectures given by Scott Aaronson and
a lecture given by Adam Bouland and Luke Schaeffer at the 28th McGill Invitational Workshop on
Computational Complexity. The workshop was held at the Bellairs
Research Institute in Holetown, Barbados in February, 2016.

\newpage

\chapter*{Abstract}
This mini-course will introduce participants to an exciting frontier for quantum computing theory: namely, questions involving the computational complexity of preparing a certain quantum state or applying a certain unitary transformation. Traditionally, such questions were considered in the context of the Nonabelian Hidden Subgroup Problem and quantum interactive proof systems, but they are much broader than that. One important application is the problem of ``public-key quantum money'' -- that is, quantum states that can be authenticated by anyone, but only created or copied by a central bank -- as well as related problems such as copy-protected quantum software. A second, very recent application involves the black-hole information paradox, where physicists realized that for certain conceptual puzzles in quantum gravity, they needed to know whether certain states and operations had exponential quantum circuit complexity. These two applications (quantum money and quantum gravity) even turn out to have connections to each other! A recurring theme of the course will be the quest to relate these novel problems to more traditional computational problems, so that one can say, for example, ``this quantum money is hard to counterfeit if that cryptosystem is secure,'' or ``this state is hard to prepare if PSPACE is not in PP/poly.'' Numerous open problems and research directions will be suggested, many requiring only minimal quantum background. Some previous exposure to quantum computing and information will be assumed, but a brief review will be provided.

\lecture{Scott Aaronson}{Anil Ada and Omar Fawzi}{Quantum Information Basics}
\pagenumbering{arabic}

\section{Introduction}

This is a mini-course on quantum complexity theory, focusing in particular on the complexity of quantum states and unitary transformations.
To be a bit more specific, we'll be interested in the following sorts of questions:
\begin{itemize}
    \item {\bf Complexity of quantum states:} Given a quantum state, how many operations do we need to prepare it?
    \item {\bf Complexity of unitary transformations:} Given a unitary transformation, how many operations do we need to apply it?
\end{itemize}

One way to present quantum computing is to talk about which languages or decision problems we can solve in quantum polynomial time. Instead of that, we'll just talk directly about how hard it is to solve {\em quantum} problems, e.g., how hard is it to create a given quantum state or to apply a given transformation of states?

This is not a new idea; people have looked at such questions since the beginning of quantum computing theory in the 1990s. For example, even if you want classical information out of your computation at the end, the subroutines will still involve transforming quantum states.  What's new, in this course, is that
\begin{enumerate}
\item[(1)] we're going to be a little more systematic about it, pausing to ask whichever questions need asking, and
\item[(2)] we're going to use the complexity of states and unitaries as a connecting thread tying together a huge number of interesting topics in quantum computing theory, both ``ancient'' (i.e., more than $10$ years old) and ``modern.''
\end{enumerate}

The impetus for the course is that the same underlying questions, e.g.\ about the difficulty of applying a given unitary transformation, have recently been showing up in contexts as varied as quantum proofs and advice, the security of quantum money schemes, the black-hole information problem, and the AdS/CFT correspondence (!). One example of such a question is the so-called ``Unitary Synthesis Problem,'' which we'll state in Lecture 3.  So, that provides a good excuse to discuss all these topics (each of which is fascinating in its own right) under a single umbrella!

Here's a rough outline of the topics we'll cover in this course:
\begin{itemize}
    \item \textbf{Lecture 1:} Crash course on quantum mechanics itself.
    \item \textbf{Lecture 2:} Crash course on quantum computation.
    \item \textbf{Lecture 3:} Complexity of states and unitaries, using the famous Hidden Subgroup Problem (HSP) as a case study.
    \item \textbf{Lecture 4:} Quantum sampling and quantum witness states.
    \item \textbf{Lecture 5:} Quantum versus classical proofs and advice (featuring Grover's algorithm).
    \item \textbf{Lecture 6:} The black-hole information problem, the firewall paradox, and the complexity of unitaries.
    \item \textbf{Lecture 7:} Wormholes, AdS/CFT, and the complexity of quantum states.
    \item \textbf{Lecture 8:} Building a secure private-key quantum money scheme.
    \item \textbf{Lecture 9:} {\em Public}-key quantum money.
    \item \textbf{Lecture 10:} Special guest lecture on the classification of quantum gate sets.
\end{itemize}

Since the audience for these lectures includes both classical computational complexity theorists with no background in physics, {\em and} string theorists with no background in computer science (!), we'll try hard to remain accessible to both audiences, while still getting as quickly as we can into current research topics.  Readers who find some particular lecture too difficult or esoteric (or conversely, who already know the material) are invited to skip around.  Roughly speaking, Lectures 1 and 2 have essential quantum-information prerequisites that the rest of the course builds on, while Lecture 3 has less familiar material that nevertheless plays a role in several other lectures.  After that, Lecture 5 builds on Lecture 4, and Lecture 9 builds on Lecture 8, and we'll highlight other connecting links---but a reader who only cared (for example) about quantum gravity could safely skip many of the other lectures, and vice versa.

\section{Basics of Quantum Mechanics}

\subsection{Quantum States, Unitary Transformations, and Measurements}

In this section we'll introduce the basics of quantum mechanics that we'll be using for the rest of the course. Quantum mechanics is simply a mathematical generalization of the rules of probability theory. In this generalization, probabilities are replaced with \emph{amplitudes}, which are complex numbers. The central claim of quantum physics is that the state of a perfectly isolated physical system is described by a unit vector of complex numbers.

An $N$-dimensional \emph{quantum} state is defined to be a unit vector over $\C^N$. The notation that the physicists use to denote quantum states is called the \emph{ket notation}, which consists of an angle bracket and a vertical bar: $\ket{\cdot}$. Suppose there are $N$ possible states. We'll denote these states by $\ket{1}, \ket{2}, \ldots, \ket{N}$. A quantum state $\ket{\psi}$ is a \emph{superposition} of these $N$ states:
\[
\ket{\psi} = \alpha_1 \ket{1} + \alpha_2 \ket{2} + \cdots + \alpha_N \ket{N}.
\]
Here, for all $i$, $\alpha_i \in \C$ is the amplitude corresponding to state $i$, and these amplitudes must satisfy $\sum_i |\alpha_i|^2 = 1$. We think of $\ket{1}, \ket{2}, \ldots, \ket{N}$ as an orthonormal basis spanning $\C^N$ and we can think of $\ket{\psi}$ as the unit vector
\[
\begin{bmatrix}
\alpha_1 \\
\alpha_2 \\
\vdots \\
\alpha_N
\end{bmatrix}.
\]
The conjugate transpose of $\ket{\psi}$ is denoted using the \emph{bra notation}: $\bra{\cdot}$. So $\bra{\psi}$ represents the row vector
\[
\begin{bmatrix}
\alpha_1^* & \alpha_2^* & \cdots & \alpha_N^*
\end{bmatrix},
\]
where $\alpha_i^*$ denotes the complex conjugate of $\alpha_i$. Given two states $\ket{\psi} = \alpha_1 \ket{1} + \cdots + \alpha_N \ket{N}$ and $\ket{\varphi} = \beta_1 \ket{1} + \cdots + \beta_N \ket{N}$, their \emph{inner product} is
\[
\braket{\psi}{\varphi} = \alpha_1^* \beta_1 + \cdots + \alpha_N^* \beta_N.
\]
So the bra and the ket come together to form the \emph{bra-ket notation}.\footnote{This notation was introduced by Paul Dirac.} The \emph{outer product} $\ketbra{\psi}{\varphi}$ corresponds to the product
\[
\begin{bmatrix}
\alpha_1 \\
\alpha_2 \\
\vdots \\
\alpha_N
\end{bmatrix}
\begin{bmatrix}
\beta_1^* & \beta_2^* & \cdots & \beta_N^*
\end{bmatrix},
\]
which results in the $N \times N$ matrix of rank $1$ in which the $(i, j)$'th entry is $\alpha_i \beta_j^*$. The outer product $\ketbra{\psi}{\psi}$ is called a \emph{projection} to $\ket{\psi}$.

Quantum mechanics has two kinds of operations that you can apply to quantum states. The first is \emph{unitary transformations} and the second is \emph{measurements}.
\\\\
\noindent
{\bf 1. Unitary transformations.} A unitary transformation $U$ is simply a linear map $U: \C^N \to \C^N$ that preserves inner products between pairs of vectors. That is, if $\ket{\psi'} = U \ket{\psi}$ and $\ket{\varphi'} = U\ket{\varphi}$, then $\braket{\psi'}{\varphi'} = \bra{\psi} U^\dagger U \ket{\varphi} = \braket{\psi}{\varphi}$, where $U^\dagger$ denotes the conjugate transpose of $U$. This implies that a unitary transformation preserves the norm of a state. So as you'd expect, a unitary transformation maps a quantum state into another quantum state. There are various equivalent ways of defining a unitary transformation. We could've defined it as a linear transformation $U$ that satisfies $U^\dagger U = I$, where $I$ is the identity (in other words, the inverse of $U$ equal its conjugate transpose). A third definition of a unitary transformation is that the rows of $U$ (when $U$ is viewed as a matrix) form an orthonormal set of vectors, or equivalently that the columns do so.

Let's now consider an example. Suppose $N = 2$ (which is the case for a single qubit). Let's denote by $\ket{0}$ and $\ket{1}$ the two orthogonal states. Define $\ket{+} = (\ket{0} + \ket{1})/\sqrt{2}$ and $\ket{-} = (\ket{0} - \ket{1})/\sqrt{2}$. Observe that $\ket{+}$ and $\ket{-}$ form another orthonormal basis.

\begin{center}
	\includegraphics[scale=0.35]{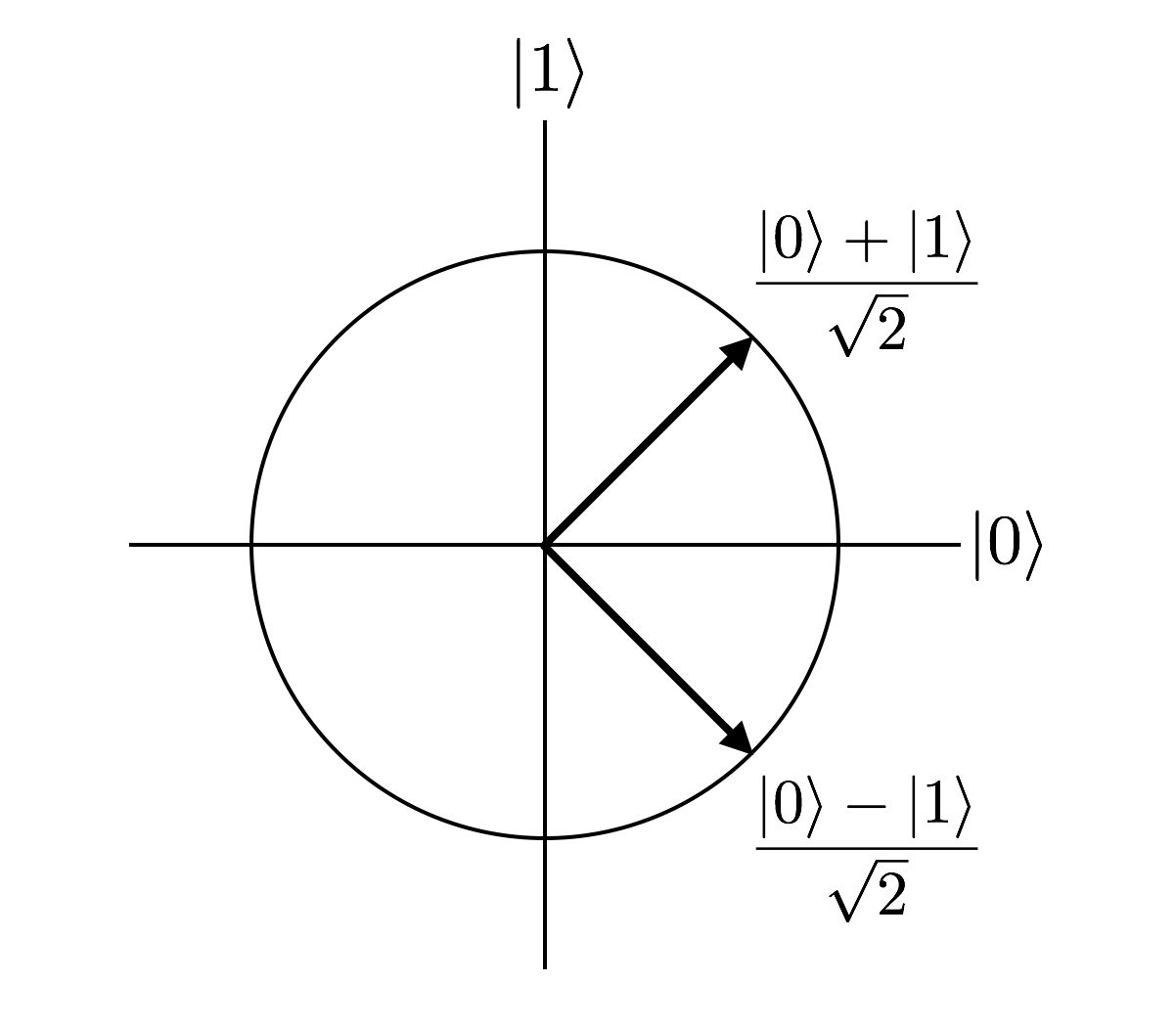}
\end{center}

We can apply unitary transformations to rotate states. For example, consider the following unitary transformation:
\[
U_\Theta =
\begin{bmatrix}
\cos \Theta & - \sin \Theta \\
\sin \Theta & \cos \Theta
\end{bmatrix}
\]
This transformation takes any state and rotates it by $\Theta$ counterclockwise. Fix $\Theta = \pi/4$. Then it's clear from the picture that if we were to apply $U_\Theta$ to $\ket{0}$, we'd get $\ket{+}$. If we were to apply $U_\Theta$ one more time, then we'd get $\ket{1}$. This small example illustrates what's called \emph{interference}: the central phenomenon of quantum mechanics that doesn't appear in classical probability theory.  Interference refers to the fact that amplitudes can cancel each other out (whereas probabilities can't)

To see this, let's write out explicitly how the state is changing. First, $U_\Theta$ changes $\ket{0}$ and $\ket{1}$ as follows:
\begin{align*}
\ket{0} \; &  \overset{U_\Theta}{\longrightarrow} \; \ket{+}  \; = \; \frac{1}{\sqrt{2}} \ket{0} + \frac{1}{\sqrt{2}} \ket{1}, \\
\ket{1} \; &  \overset{U_\Theta}{\longrightarrow} \; -\ket{-}  \; = \; -\frac{1}{\sqrt{2}} \ket{0} + \frac{1}{\sqrt{2}} \ket{1}.
\end{align*}
If we then start with $\ket{0}$ and apply $U_\Theta$ twice, we get
\begin{align*}
\ket{0}  & \; \overset{U_\Theta}{\longrightarrow} \; \frac{1}{\sqrt{2}} \ket{0} + \frac{1}{\sqrt{2}} \ket{1} \\
            & \; \overset{U_\Theta}{\longrightarrow} \; \frac{1}{\sqrt{2}} \left( \frac{1}{\sqrt{2}} \ket{0} + \frac{1}{\sqrt{2}} \ket{1} \right ) + \frac{1}{\sqrt{2}} \left(  -\frac{1}{\sqrt{2}} \ket{0} + \frac{1}{\sqrt{2}} \ket{1} \right ) \\
            & \quad = \quad \frac{1}{2}\ket{0} + \frac{1}{2} \ket{1} - \frac{1}{2}\ket{0} + \frac{1}{2}\ket{1} \\
            & \quad = \quad \ket{1}.
\end{align*}
Note that we arrive at the $\ket{1}$ state because the positive $1/2$ and the negative $1/2$ contributions to the amplitude of the $\ket{0}$ state cancel each other out, i.e., they interfere with each other destructively. One can describe the double-slit experiment in this manner. A photon can be in a superposition of going through two slits, and the two paths that it can take can interfere destructively and cancel each other out. This is why you don't see the photon appear at a certain location on the screen even though if you were to close one of the slits, the photon would have some nonzero probability to appear at that location.
\\\\
\noindent
{\bf 2. Measurements.} Quantum states can't exist forever in the abstract realm. At some point we have to \emph{measure} them. The most basic type of measurement we can do is with respect to the orthonormal basis $\{\ket{1}, \ket{2}, \ldots, \ket{N}\}$. If $\ket{\psi} = \alpha_1 \ket{1} + \cdots + \alpha_N \ket{N}$, and we measure $\ket{\psi}$, we'll get the outcome $i$ with probability $|\alpha_i|^2$. Measurement is a destructive operation. This means that once $\ket{\psi}$ is measured and outcome $i$ is observed, the state \emph{collapses} to $\ket{i}$, and all other information about the original state is vanished (or exists in parallel universes depending on your interpretation of quantum mechanics).

A more general version of the measurement rule allows us to measure a given state $\ket{\psi}$ in any orthonormal basis $\{\ket{v_1}, \ket{v_2}, \ldots, \ket{v_N}\}$. In this case, the probability that we get the outcome $i$ is $|\braket{\psi}{v_i}|^2$. And once the measurement is done, the state collapses to $\ket{v_i}$.

Note that the \emph{only} way one can access information about the amplitudes of a quantum state is through measurements. This fact is very important in quantum mechanics. Many of the misconceptions people have about quantum mechanics come from imagining that you have some kind of access to a quantum state other than measurements.

Once consequence of our not having such access is that the \emph{global phase} of a quantum state is unobservable. For example, $-\ket{0}$ and $\ket{0}$ are physically the same state, since the $-1$ multiplier in $-\ket{0}$ disappears whenever we calculate the probability of a measurement outcome by taking the absolute squares of amplitudes.  In contrast, \emph{relative} phase is physically observable: for example, $\alpha \ket{0}+\beta \ket{1}$ and $\alpha \ket{0}-\beta \ket{1}$ are different whenever $\alpha$ and $\beta$ are both nonzero, even though $\ket{1}$ and $-\ket{1}$ are the same when considered in isolation.

Two states are more distinguishable the closer they are to being orthogonal (in fact, two states are perfectly distinguishable if and only if they're orthogonal). For example, $\ket{+}$ and $\ket{-}$ are orthogonal, and indeed they can be perfectly distinguished by measuring them with respect to the basis $\{\ket{+}, \ket{-}\}$.

\subsection{Multipartite Systems, Entanglement, and Density Matrices}

Next, we'll talk about how to represent the combination of more than one quantum state. In general, to combine states, we use the tensor product. Suppose we have a qubit $\alpha \ket{0} + \beta \ket{1}$ and another qubit $\gamma \ket{0} + \delta \ket{1}$. The joint state of these two qubits is
\[
(\alpha \ket{0} + \beta \ket{1}) \tensor (\gamma \ket{0} + \delta \ket{1}) = \alpha \gamma (\ket{0} \tensor \ket{0}) + \alpha \delta (\ket{0} \tensor \ket{1}) + \beta \gamma (\ket{1} \tensor \ket{0}) + \beta \delta (\ket{1} \tensor \ket{1}).
\]
We often omit the tensor product sign and write the state as:
\[
\alpha \gamma \ket{0}  \ket{0} + \alpha \delta \ket{0}  \ket{1} + \beta \gamma \ket{1}  \ket{0} + \beta \delta \ket{1}  \ket{1}.
\]
Going even further, we use the shorthand $\ket{00}$ to denote $\ket{0}\ket{0}$, so we can rewrite the above state as:
\[
\alpha \gamma \ket{00} + \alpha \delta \ket{01} + \beta \gamma \ket{10} + \beta \delta \ket{11}.
\]
This is a vector in a 4-dimensional space spanned by $\{\ket{00}, \ket{01}, \ket{10}, \ket{11}\}$. A bipartite state is called \emph{separable} if it can be written as the tensor product of two states. Clearly, by construction, the above state is separable. However, there are states that aren't separable. The most famous example is called the {\em Bell pair} (or singlet, or EPR pair):
\[
\frac{\ket{00} + \ket{11}}{\sqrt{2}}.
\]
If a state is not separable, we call it \emph{entangled}, which can be thought of as the quantum analogue of correlation between two random variables.

Two qubits being entangled with each other doesn't prevent them from being separated by arbitrary distances (they could even be in two different galaxies).

Now suppose Alice and Bob are physically separated. Alice has one qubit, Bob has another qubit, and the joint state of the two qubits is
\[
a \ket{00} + b \ket{01} + c \ket{10} + d \ket{11}.
\]
What happens if Alice measures her qubit?

Here we need a rule for \emph{partial measurements}. The rule simply says that Alice observes the outcome $0$ with probability $|a|^2 + |b|^2$, and observes the outcome $1$ with probability $|c|^2 + |d|^2$. Also, after the measurement is done, the state undergoes a partial collapse. If Alice observes a $0$, we'd expect the collapsed state to be described by $a \ket{00} + b \ket{01}$. However, we need to normalize this vector so that it's a unit vector. Therefore the state actually collapses to $(a \ket{00} + b \ket{01})/\sqrt{|a|^2 + |b|^2}$, which means Bob's qubit after Alice's measurement would be $(a \ket{0} + b \ket{1})/\sqrt{|a|^2 + |b|^2}$. Similarly, if Alice observes a $1$, then Bob's state becomes $(c \ket{0} + d \ket{1})/\sqrt{|c|^2 + |d|^2}$.

Now suppose Alice and Bob share a Bell pair. We'd like a description of the qubit that Alice has. We know that the joint state is entangled, but can we get a description of Alice's qubit that's local to Alice?

One attempt would be to describe Alice's state as $\ket{+}$. But this would be flat-out wrong: you can check for yourself that, if we were to measure Alice's qubit in the $\{\ket{+}, \ket{-}\}$ basis, then we'd observe $\ket{+}$ and $\ket{-}$ with equal probability. More generally, if you measure Alice's qubit in \emph{any} orthonormal basis, you'll get the two outcomes with equal probability. So Alice's qubit is not behaving like the quantum states that we've talked about so far. In fact, it's behaving much more like a classical random bit.

We need a formalism to represent such states. A \emph{mixed state} is defined as a probability distribution over quantum states: $\{(p_i, \ket{\psi_i})\}$, where $p_i$ is the probability of state $\ket{\psi_i}$. The kind of quantum states we talked about before, the individual vectors $\ket{\psi_i}$, are often called \emph{pure states} in order to distinguish them from mixed states. With this formalism, Alice's qubit in the above scenario can be represented as the mixed state $\{(1/2, \ket{0}), (1/2, \ket{1})\}$.

Now, an important and counterintuitive property of mixed states is that different decompositions can physically represent the same state. For example, the mixed states $\{(1/2, \ket{0}), (1/2, \ket{1})\}$ and $\{(1/2, \ket{+}), (1/2, \ket{-})\}$ are actually physically indistinguishable (just like $\ket{0}$ and $-\ket{0}$ are physically indistinguishable).  That is, we can measure both mixtures in any basis we want, and the probabilities of all measurement outcomes will the same in both cases (as you can check for yourself).

This means that this encoding of a mixed state as a probability distribution over pure states is a redundant encoding.  Conveniently, though, there's a different way to represent a mixed state, called a \emph{density matrix}, that has no redundancy.

The density matrix, usually denoted by $\rho$, corresponding to the mixed state $\{(p_i, \ket{\psi_i})\}$ is the $N\times N$ matrix
\[
\rho = \sum_i p_i \ketbra{\psi_i}{\psi_i}.
\]
We leave it as an exercise to verify that two mixed states have the same density matrix if and only if they're indistinguishable by any measurement.

Suppose Alice and Bob jointly have a bipartite state $\sum \alpha_{ji} \ket{j} \ket{i}$. This state can be rewritten as $\sum \beta_i \ket{\psi_i} \ket{i}$, for some $\beta_i$ and $\ket{\psi_i}$. Then the density matrix that describes Alice's state is $\sum |\beta_i|^2 \ketbra{\psi_i}{\psi_i}$. The example below nicely illustrates this.
\\
\noindent
{\bf Example:} Let the joint state of Alice and Bob be
\[
\frac{\ket{00} + \ket{01} - \ket{10}}{\sqrt{3}},
\]
which we can rewrite as
\[
\sqrt{\frac{2}{3}} \left( \frac{\ket{0} - \ket{1}}{\sqrt{2}} \right) \ket{0} + \sqrt{\frac{1}{3}} \ket{0}\ket{1}.
\]
Then the density matrix describing Alice's state is
\[
   \frac{2}{3} \begin{bmatrix}
		     1/2 & -1/2 \\
		     -1/2 & 1/2
		    \end{bmatrix}        +  \frac{1}{3} \begin{bmatrix}
									1 & 0 \\
									0 & 0
									\end{bmatrix}
= \begin{bmatrix}
2/3 & -1/3 \\
-1/3 & 1/3
\end{bmatrix}.
\]

In general, any matrix that is Hermitian, positive semidefinite, and whose eigenvalues sum to 1 is a density matrix. We can state the rules of quantum mechanics directly in terms of density matrices. A quantum state is a mixed state represented by a density matrix $\rho$.  When we apply a unitary matrix $U$ to $\rho$, we get the state $U \rho U^\dagger$. It's also easy to verify that if we do a measurement in the basis $\{\ket{1}, \ket{2}, \ldots, \ket{N}\}$, we'll observe outcome $i$ with probability $\rho_{ii}$. In other words, the probabilities of the different outcomes are just the diagonal entries of $\rho$.

\bigskip

Let's now revisit the idea that in a bipartite system, one party's local measurement can affect the other party's state even though the two states can be arbitrarily far apart. A popular version of entanglement is that the particles in the universe are spookily connected; you measure one qubit and it must instantaneously send a signal to another qubit that it's entangled with.

But this is a misunderstanding of entanglement! Quantum mechanics actually upholds Einstein's principle that you can't send a signal instantaneously (and in particular, faster than the speed of light). The quantum-mechanical version of that principle can be proven and is known as the No Communication Theorem.

\begin{theorem}[No Communication Theorem]
Suppose Alice and Bob share a bipartite state. Then nothing that Alice chooses to do (i.e.\ any combination of measurements and unitary transformations) can change Bob's local state, i.e., the density matrix describing Bob's state.
\end{theorem}

The proof is left as an exercise for the reader.

As an example, consider the Bell state $\frac{1}{\sqrt{2}} (\ket{00} + \ket{11})$. Then the density matrix describing Bob's state is the following \emph{maximally mixed state}:
\[
\rho_B = \begin{bmatrix}
1/2 & 0 \\
0 & 1/2
\end{bmatrix}.
\]
If Alice measures her qubit, with $1/2$ probability she'll observe a $0$, in which case Bob's state will be $\ket{0}$, and with $1/2$ probability she'll observe a $1$, in which case Bob's state will be $\ket{1}$. So Bob's density matrix does not change; he still has a maximally mixed state. (Of course, \emph{if} we condition on Alice's outcome, that \emph{does} change Bob's density matrix. And this would also be true in the classical probability world.) As an exercise, one can also check that whatever unitary transformation Alice applies to her qubit won't change Bob's density matrix.

In 1935, Einstein, Podolsky and Rosen (EPR) published a famous paper which asked, in modern terms, whether quantum entanglement can be understood purely in terms of classical correlation.  They thought this was necessary to uphold the `no faster than light communication' principle. Indeed, the experiment described above, in which we measure the Bell pair in the $\{\ket{0}, \ket{1}\}$ basis, doesn't distinguish quantum entanglement from classical correlation.

John Bell, in the 1960s, clarified the situation significantly. He described experiments (in which Alice and Bob both randomly choose a basis to do their measurements in rather than just selecting the $\{\ket{0}, \ket{1}\}$ basis) that provably {\em distinguish} quantum entanglement from classical correlation.  More concretely, he proved that, if such experiments have the outcomes that quantum mechanics predicts (and later experimental work would show that they do), then the correlations obtained by quantum entanglement can't be achieved by any theory of shared classical randomness (i.e., classical correlations).

Entanglement is thus ``intermediate'' between classical correlation and faster-than-light communication: it doesn't allow faster-than-light communication, and yet any simulation of entanglement in a classical universe \emph{would} require faster-than-light communication.\footnote{From a computational complexity theorist's point of view, Bell's result can be restated as saying that there's a 2-prover game that can be won with higher probability in nature than it can be if the provers were separated and had only classical shared randomness.}

\subsection{More General Operations and Measurements}

The type of measurement we've discussed so far is called a \emph{projective measurement}: you pick an orthonormal basis to measure in and you ask your state to collapse to one of the basis vectors. But we can also make measurements that involve an additional ancilla system.

For example, suppose we have a state $\ket{\psi}$ that we want to measure, but we introduce an additional qubit in the $\ket{0}$ state. So the joint state is $\ket{\psi} \tensor \ket{0}$, and we can apply a projective measurement to this state. This is {\em not} the same as a projective measurement on just $\ket{\psi}$, since it lives in a higher-dimensional vector space. This more general kind of measurement goes by the catchy name {\em positive-operator valued measurement (POVM)}.

More precisely, a POVM can be described as follows. Suppose $\rho$ is an $N$-dimensional density matrix. A POVM is simply a set of Hermitian positive semidefinite $N$ by $N$ matrices $\{E_1, E_2, \ldots, E_k\}$ (where $k$ can be chosen to be any number), with the property that $\sum_i E_i$ equals the identity matrix $I$. When such a POVM is applied to $\rho$, we observe the outcome $i$ with probability $\Tr(E_i \rho)$. Let's underline the fact that if $k > N$, then the POVM allows for more than $N$ outcomes. Also note that the POVM formalism doesn't specify what happens to the post-measurement state. It's left as an exercise to show that any POVM (as just defined) can be described as a projective measurement applied to the original state together with an ancilla state. Conversely, any projective measurement done on the original state together with an ancilla state can be described as a POVM.

There's yet another formalism for quantum operations that in some sense, is more general than both unitaries and POVMs.  It's appropriately called the \emph{superoperator} formalism. This formalism captures any combination of unitaries, measurements, and ancillary systems in a single mathematical package.

Formally, a superoperator is specified by a set of matrices $\{E_1, E_2, \ldots, E_k\}$ satisfying $\sum_i E_i^\dagger E_i = I$. Given a density matrix $\rho$, the superoperator maps it to another density matrix
\[
S(\rho) = \sum_i^k E_i \rho E_i^\dagger.
\]
Note that $S(\rho)$ can have a different dimension than $\rho$. (Exercise: show that $S(\rho)$ is indeed a density matrix.) Superoperators can also be used to do a measurement: outcome $i \in \{1,2,\ldots, k\}$ is observed with probability $\Tr(E_i \rho E_i^\dagger)$ and $\rho$ collapses to $E_i \rho E_i^\dagger / \Tr(E_i \rho E_i^\dagger)$.

Another concept that we will use is  \emph{purification} of a mixed state. It's easy to prove that any mixed state $\rho$ in $N$ dimensions can be obtained by tracing out a higher dimensional pure state $\ket{\psi}$ (in fact, $N^2$ dimensions suffices for this). That is, given $\rho$ which acts on a Hilbert space $\mathcal{H}_A$, there is a pure state $\ket{\psi}$ which lives in a Hilbert space $\mathcal{H}_A \tensor \mathcal{H}_B$ such that $\Tr_B(\ketbra{\psi}{\psi}) = \rho$. This $\ket{\psi}$ is called a purification of $\rho$.

\subsection{Measuring Distances Between Quantum States}

There are several ways to measure distances between quantum states. For example, if we had two pure states, we could try to measure the distance between them using the Euclidean distance. However, note that even though the Euclidean distance between $\ket{0}$ and $-\ket{0}$ is $2$, they're really the same state. Thus, a better choice would be to look at the inner product of the states. In particular, if $\ket{\psi}$ and $\ket{\varphi}$ are two pure states, their distance can be measured using $|\braket{\psi}{\varphi}|$, which is always between $0$ and $1$. We can generalize the inner product to allow mixed states, which gives us a measure known as \emph{fidelity}. If $\rho$ is a mixed state and $\ket{\psi}$ is a pure state, then the fidelity between $\rho$ and $\ket{\psi}$ is $\sqrt{ \bra{\psi} \rho \ket{\psi} }$. If $\rho$ and $\sigma$ are two mixed states, then the fidelity between them is defined as
\[
F(\rho, \sigma) = \max_{\ket{\psi}, \ket{\varphi}} | \braket{\psi}{\varphi} |,
\]
where the maximum is over all purifications $\ket{\psi}$ of $\rho$ and all purifications $\ket{\varphi}$ of $\sigma$. The fidelity between $\rho$ and $\sigma$ can also be computed as
\[
F(\rho, \sigma) = \Tr\left(\sqrt{ \sqrt{\rho} \sigma \sqrt{\rho}}\right) = \| \sqrt{\rho} \sqrt{\sigma} \|_{tr},
\]
where the trace norm $\| \cdot \|_{tr}$ is equal to the sum of the singular values.

The second distance measure we will look at is the \emph{trace distance}, which is a generalization of variation distance between probability distributions. If $D_1$ and $D_2$ are two probability distributions, the variation distance between them is
\[
\| D_1 - D_2 \| = \frac{1}{2} \sum_x |D_1(x) - D_2(x)|.
\]
Informally, this is a measure of how well we can distinguish a sample from the first distribution from a sample from the second distribution. If the distance is $1$, we can tell them apart perfectly, and if it's $0$, we can't distinguish them at all. Any experiment that accepts a sample from the first distribution with probability $p$ will accept a sample from the second distribution with probability in $\big[p - \|D_1 - D_2\|, p + \|D_1 - D_2\|\big]$. The generalization of this to quantum states is the trace distance. Given two mixed states $\rho$ and $\sigma$, their trace distance is defined as
\[
\| \rho - \sigma \|_{tr} = \frac12 \Tr\left( \sqrt{(\rho - \sigma)^2} \right) = \frac12 \sum_{i} |\lambda_i| \ ,
\]
where $\lambda_i$ are the eigenvalues of $\rho - \sigma$. It's easy to check that the trace distance is invariant under unitary transformations (same is true for fidelity):
\[
\| \rho - \sigma \|_{tr} = \| U \rho U^\dagger - U \sigma U^\dagger \|_{tr}.
\]
Being a distance measure, trace distance indeed satisfies the triangle inequality:
\[
\| A - C\|_{tr} \leq \|A - B\|_{tr} + \|B - C\|_{tr}.
\]
Furthermore, if you trace out parts of $\rho$ and $\sigma$, that can only decrease the trace distance $\| \rho - \sigma \|_{tr}$ (exercise).

Similar to the variation distance, if we have two mixed states $\rho$ and $\sigma$ with $\| \rho - \sigma \|_{tr} \leq \eps$, and we have a 2-outcome POVM that accepts $\rho$ with probability $p$, then the same POVM accepts $\sigma$ with probability in $[p-\eps, p+\eps]$. This is a very important property that we'll use many times throughout the course.

The trace distance and fidelity are related as follows:
\[
1 - F(\rho, \sigma) \leq \| \rho - \sigma \|_{tr} \leq \sqrt{1 - F(\rho, \sigma)^2}.
\]
In the upper bound, equality is achieved if $\rho$ and $\sigma$ are both pure states. This means that if $\rho$ and $\sigma$ are pure states, then there is a simple expression for the trace distance in terms of the inner product.

\section{``Almost As Good As New'' Lemma and Quantum Union Bound }

We'll now see an application of trace distance. We know that measurement in quantum mechanics is an inherently destructive process. Yet, we should not interpret this as saying that every measurement is always maximally destructive. In reality, how destructive a measurement is depends on whether there are any basis vectors (in the basis in which we are measuring) that are close to the state being measured. If the state being measured is close to a basis vector, then with high probability it will collapse to that basis vector---and since those vectors were close to begin with, the loss of information about our original vector will be minimal.


\begin{lemma}[``Almost As Good As New Lemma''~\cite{aar:adv}, closely related to ``Gentle Measurement Lemma'' \cite{gentle}]
\label{lem:aagan}
Let $\rho$ be a mixed state acting on $\C^d$. Let $U$ be a unitary and $(\Pi_0, \Pi_1 = \mathsf{1} - \Pi_0)$ be projectors all acting on $\C^{d} \otimes \C^{d'}$. We interpret $(U, \Pi_0, \Pi_1)$ as a measurement performed by appending an ancillary system of dimension $d'$ in the state $\proj{0}$, applying $U$ and then performing the projective measurement $\{ \Pi_0, \Pi_1\}$ on the large system. Assuming that the outcome corresponding to $\Pi_0$ has probability $1-\eps$, i.e., $\Tr[ \Pi_0 (U \rho \otimes \proj{0} U^{\dagger}) ] = 1 - \eps$, we have
\begin{align*}
\| \rho - \tilde{\rho} \|_{tr} \leq \sqrt{\eps} \ ,
\end{align*}
where $\tilde{\rho}$ is state after performing the measurement and then undoing the unitary $U$ and tracing out the ancillary system: $\tilde{\rho} =  \Tr_{d'}\left(U^{\dagger} \left( \Pi_0 U (\rho \otimes \proj{0}) U^{\dagger} \Pi_0  +  \Pi_1 U (\rho \otimes \proj{0}) U^{\dagger} \Pi_1 \right) U \right)$.
\end{lemma}
\begin{proof}
We start with the case where $\rho$ is a pure state that we write $\proj{\psi}$.
Defining the two orthogonal unit vectors $\ket{\phi_0}$ and $\ket{\phi_1}$ as
\begin{align*}
\ket{\phi_0} = \frac{\Pi_{0} U (\ket{\psi} \otimes \ket{0})}{\sqrt{\Tr[\Pi_0 (U \rho \otimes \proj{0} U^{\dagger})]}} \quad \text{and} \quad \ket{\phi_1} = \frac{\Pi_{1} U (\ket{\psi} \otimes \ket{0})}{\sqrt{\Tr[\Pi_1 (U \rho \otimes \proj{0} U^{\dagger})]}} \ ,
\end{align*}
and using the defining expression for $\eps$, we have
\begin{align*}
U (\ket{\psi} \otimes \ket{0}) = \sqrt{1-\eps} \ket{\phi_0} + \sqrt{\eps} \ket{\phi_1} \ .
\end{align*}
With this notation, we have
\begin{align*}
\tilde{\rho} = \Tr_{d'}\left[ U^{\dagger} \left( (1-\eps) \proj{\phi_0} + \eps \proj{\phi_1} \right) U^{\dagger} \right] \ .
\end{align*}
We can now compute
\begin{align*}
\| U (\proj{\psi} \otimes \proj{0}) U^{\dagger} - \left((1-\eps) \proj{\phi_0} + \eps \proj{\phi_1} \right) \|_{tr} &= \| \sqrt{\eps (1-\eps)} \ketbra{\phi_0}{\phi_1} + \sqrt{\eps(1-\eps)} \ketbra{\phi_1}{\phi_0} \|_{tr} \\
&= \sqrt{\eps(1-\eps)} \ .
\end{align*}
To see the last equality, we observe that $\ket{\phi_0} \pm \ket{\phi_1}$ are the two eigenvectors with eigenvalues $\pm \sqrt{\eps(1-\eps)}$.
Using the fact that the trace distance cannot increase if we apply $U^{\dagger}$ followed by the partial trace on the system $\C^{d'}$,  we obtain
\begin{align}
\label{eq:aagan_pure}
\| \rho - \tilde{\rho} \|_{tr} \leq \sqrt{\eps(1-\eps)} \leq \sqrt{\eps} \ .
\end{align}
For the general case, we take an eigendecomposition of $\rho = \sum_i p_i \proj{\psi_i}$. We define
\begin{align*}
\eps_i &= \Tr[\Pi_1 (U \proj{\psi_i} \otimes \proj{0} U^{\dagger})] \ , \\
 \tilde{\psi}_i &= \Tr_{d'}\left(U^{\dagger} \left( \Pi_0 U (\proj{\psi_i} \otimes \proj{0}) U^{\dagger} \Pi_0  +  \Pi_1 U (\proj{\psi_i} \otimes \proj{0}) U^{\dagger} \Pi_1 \right) U \right) \ .
\end{align*}
 Note that $\sum_i p_i \eps_i = \eps$ and $\sum_{i} p_i \tilde{\psi}_i = \tilde{\rho}$. Using the triangle inequality and the pure state case~\eqref{eq:aagan_pure}, we have
\begin{align*}
\|\rho - \tilde{\rho} \|_{tr} &\leq \sum_i p_i \| \proj{\psi_i} - \tilde{\psi}_i \|_{tr} \\
&\leq \sum_{i} p_i \sqrt{\eps_i} \\
&\leq \sqrt{\eps} \ ,
\end{align*}
where we used the concavity of the square-root function in the last inequality.
\end{proof}

(Note: the proof above fixes an error in the proof given in \cite{aar:adv}.)

The above is also known as the gentle measurement lemma in the information theory literature (see e.g.,~\cite[Section 9.4]{wildebook} for more discussion). Its standard formulation is as follows.
\begin{lemma}[``Gentle Measurement''~\cite{gentle}]
Let $\rho$ be a mixed state and $\{\Lambda, \mathsf{1}-\Lambda\}$ be a two-outcome POVM with $\Tr[\Lambda \rho] = 1-\eps$, then
\begin{align*}
\| \tilde{\rho} - \rho \|_{tr} \leq \eps \ ,
\end{align*}
where $\tilde{\rho} = \frac{\sqrt{\Lambda} \rho \sqrt{\Lambda}}{\Tr[\Lambda \rho]}$.
\end{lemma}
To compare to Lemma~\ref{lem:aagan}, we can set $\Lambda = \bra{0} U^{\dagger} \Pi_0 U \ket{0}$. Note that here, the state conditioned on the measurement outcome is used, whereas Lemma~\ref{lem:aagan} uses the unconditioned state.

It's useful to generalize these results to multiple measurements performed in sequence.

\begin{lemma}[``Quantum Union Bound'' (e.g. \cite{aar:qmaqpoly})]
\label{lem:qunionbound}
Let $\rho$ be a mixed state acting on $\C^d$. Suppose we have $k$ measurements described by $(U_i, \Pi_0^i, \Pi_1^i)$ for every $i \in \{1, \dots, k\}$. Assume that for all $i \in \{1, \dots, k\}$, $\Tr[\Pi_0^i (U_i \rho \otimes \proj{0} U_i^{\dagger})] \geq 1-\eps$. Define the superoperator $\mathcal{M}_i$ taking operators on $\C^d$ to operators on $\C^d$ by
\begin{align*}
\mathcal{M}_i(\rho) = \Tr_{d'_i}\left(U_i^{\dagger} \left( \Pi^i_0 U_i (\rho \otimes \proj{0}) U_i^{\dagger} \Pi^i_0  +  \Pi_1^i U_i (\rho \otimes \proj{0}) U_i^{\dagger} \Pi_1^i \right) U_i \right) \ .
\end{align*}
Then,
\begin{align*}
\| \rho - (\mathcal{M}_k \circ \cdots \circ \mathcal{M}_1)(\rho) \|_{tr} \leq k \sqrt{\eps} \ .
\end{align*}
\end{lemma}
\begin{proof}
Using the triangle inequality, we have
\begin{align*}
\| \rho - (\mathcal{M}_k \circ \cdots \circ \mathcal{M}_1)(\rho) \|_{tr}
&\leq \| \rho - \mathcal{M}_k(\rho) \|_{tr} + \cdots + \| (\mathcal{M}_{k} \circ \cdots \circ \mathcal{M}_2)(\rho) - (\mathcal{M}_k \circ \cdots \circ \mathcal{M}_1)(\rho) \|_{tr} \\
&\leq \| \rho - \mathcal{M}_k(\rho) \|_{tr} + \cdots + \| \rho - \mathcal{M}_1(\rho) \|_{tr} \ ,
\end{align*}
where we used in the last inequality that applying a superoperator cannot increase the trace distance. To conclude, we observe that Lemma~\ref{lem:aagan} precisely states that $\| \rho - \mathcal{M}_i(\rho) \|_{tr} \leq \sqrt{\eps}$ for all $i$.
\end{proof}
We note that, in the special case of projective measurements, the bound of Lemma \ref{lem:qunionbound} can be improved from $O(k \sqrt{\eps})$ to $O(\sqrt{k \eps})$ \cite{sen}---and indeed, even better than that, to $O(k \eps)$ \cite{gao}.  Meanwhile, for arbitrary POVMs, Wilde \cite{wilde} showed that the bound can be improved to $O(\sqrt{k \eps})$, but whether it can be improved further to $O(k\eps)$ remains open.


\lecture{Scott Aaronson}{Eric Allender and Michael Saks}{Quantum Circuits, Gates, and Complexity}

\section{The No-Cloning Theorem}

Let's now prove an extremely basic fact about quantum mechanics that will be invoked throughout the course.

\begin{theorem}
There is no procedure
$\ket{\psi} \rightarrow \ket{\psi}^{\otimes 2}.$
(Here $\ket{\psi}^{\otimes 2}$ is an abbreviation for $\ket{\psi}\otimes \ket{\psi}$.)
\end{theorem}
\begin{proof}
We give two proofs of this theorem (called the {\bf No-Cloning Theorem}).

For the first proof, consider a transformation that maps
$(\alpha \ket{0} + \beta\ket{1})\ket{0}$
to
$(\alpha \ket{0} + \beta\ket{1}) (\alpha \ket{0} + \beta\ket{1})$.
But the latter expression is equal to
$\alpha^2\ket{00} + \alpha\beta \ket{01} + \alpha\beta\ket{10} + \beta^2 \ket{11}$,
and thus this is a manifestly {\em nonlinear} transformation, and thus cannot
be unitary. Technically, this assumes that the cloning procedure is unitary; one needs to say slightly
more to extend the theorem to superoperators.

Thus, for a more general proof, consider a superoperator $S$ such that
$S(\ket{\psi}) = \ket{\psi}^{\otimes 2}.$  Consider two state vectors
$\ket{\phi}$ and $\ket{\psi}$, and the associated density matrices
$\rho = \ketbra{\phi}{\phi}$ and
$\sigma = \ketbra{\psi}{\psi}$.
Let $\rho^{\otimes 2}$ and $\sigma^{\otimes 2}$ likewise be the density matrices
of the cloned states $S(\rho)$ and $S(\sigma)$, respectively.

Let $C = \Tr(\rho\cdot\sigma)$ and note that $0 \leq C \leq 1$.
The trace distance of
$\rho$ and $\sigma, ||\rho - \sigma||_{\rm tr}$, for density matrices
of this form, is equal to $2\sqrt{1-C}$.  Similarly,
$||\rho^{\otimes 2} - \sigma^{\otimes 2}||_{\rm tr}$ is equal to
$2\sqrt{1-C^2}$.
(For these equalities, consult equations (5) and (6) in
\cite{Maz15} and the references cited
therein.)
Note that $2\sqrt{1-C}< 2\sqrt{1-C^2}$.  Thus the trace distance
{\em increases} after applying the cloning superoperator $S$.
This contradicts the fact that applying a superoperator can never
increase the trace distance between two states.
\end{proof}

\section{Monogamy of Entanglement}

We now discuss another important property of quantum states. Consider a 3-qubit system with qubits held by Alice, Bob, and Charlie.
\emph{Monogamy of entanglement} is the principle that, if Alice's and Bob's qubits are fully entangled,
then there can be no entanglement---or for that matter, even classical correlation---between
Alice's and Charlie's qubits.  To illustrate, consider the pure state
$$\frac{1}{\sqrt{2}}(\ket{0}_A\ket{0}_B\ket{0}_C + \ket{1}_A\ket{1}_B\ket{1}_C).$$
The probability that Alice sees 0 is 1/2, and thus the density matrix
for the reduced system $BC$ is $\ket{00}_{BC}$ with probability 1/2,
and $\ket{11}_{BC}$ with probability 1/2.
Thus Bob and Charlie's density matrix is
$$\rho_{BC} = \begin{bmatrix}
1/2 & 0 & 0 & 0 \\
0 & 0 & 0 & 0 \\
0 & 0 & 0 & 0 \\
0 & 0 & 0 & 1/2
\end{bmatrix}.$$
This is an unentangled mixed state.

This should be contrasted with the Bell pair between
Alice and Bob, which is
$$\frac{1}{\sqrt{2}}(\ket{0}_A\ket{0}_B + \ket{1}_A\ket{1}_B),$$
or
$$\rho_{AB} = \begin{bmatrix}
1/2 & 0 & 0 & 1/2 \\
0 & 0 & 0 & 0 \\
0 & 0 & 0 & 0 \\
1/2 & 0 & 0 & 1/2
\end{bmatrix}.$$
This is an entangled pure state.

{\em Monogamy of Entanglement} says that, if $\rho_{AB}$ is the density matrix
of the maximally entangled state, then $\rho_{ABC} = \rho_{AB} \otimes \rho_C$.
I.e., if Alice and Bob are totally entangled, then there can be no entanglement
at all between Alice and Charlie---or even classical correlation.  There are quantitative versions of this theorem, as well as
generalizations to larger numbers of qubits.

Before moving on to quantum circuits, we should also say something about the different
ways to measure entanglement.  {\em Distillable Entanglement} is defined as the number of EPR pairs that one can obtain from a state, using only
local operations and classical communication.
{\em Entanglement of Formation} is defined as the number of EPR
pairs that are needed in order to create the state, again using local operations and classical communication.  These notions
coincide for pure states, but for mixed states, distillable entanglement can be strictly less than entanglement of formation.

\section{Quantum Circuits and Universal Gate Sets}

In this section, we introduce the conventions that we will follow, in
representing quantum circuits that act on $n$ qubits, implementing a
transformation
$$\sum_{x \in \{0,1\}^n} \alpha_x \ket{x} = \ket{\psi} \mapsto U(\ket{\psi})$$
where $U$ is a $2^n\times 2^n$ unitary matrix.  Below, we see an example
on $n=4$ qubits, where the gate $G$ is applied to the two low-order bits.
This circuit implements the $2^n\times 2^n$ matrix
$I\otimes I \cdots \otimes I \otimes G$, which has copies of the
$4$-by-$4$ unitary matrix $G$ along the main diagonal.
\[
\Qcircuit @C=1em @R=1.5em
{
    \lstick{\ket{x_1}} & \qw  & \qw \\
    \lstick{\ket{x_2}} & \qw   & \qw \\
    \lstick{\ket{x_3}}& \multigate{1}{G} & \qw \\
    \lstick{\ket{x_4}}& \ghost{G} & \qw
}
\]

Define $C(U)$ to be the minimum number of gates required in order to implement
$U$ (exactly) using {\em any} sequence of 1- and 2-qubit gates.  For almost all $U$, we have
$C(U) \geq 4^n$ by a dimension-counting argument, and this is approximately tight by an explicit construction.  However, our focus is
primarily what can be accomplished using a fixed collection of gates.  Here
are some important examples.

\begin{itemize}
\item The Hadamard Gate
$$H = \frac{1}{\sqrt{2}}
\left[
\begin{array}{cc}
1 & 1 \\
1 & -1
\end{array}
\right]
$$
mapping $\ket{0}$ to $\ket{+}$ and $\ket{1}$ to $\ket{-}$ (and vice versa).

\item The {\em Controlled Not} gate (CNOT), mapping $\ket{x,y}$ to $\ket{x,x\oplus y}$ (where $\oplus$ represents XOR).
\[
{\rm CNOT} =
\left[
\begin{array}{cccc}
1 & 0 & 0 & 0 \\
0 & 1 & 0 & 0 \\
0 & 0 & 0 & 1 \\
0 & 0 & 1 & 0
\end{array}
\right]
\]

As an example, the following circuit diagram represents a Hadamard gate
on the first qubit providing the ``control'' for a CNOT applied to the
second qubit.

\[
\Qcircuit @C=1em @R=1.5em
{
    \lstick{\ket{x_1}} & \gate{H} & \ctrl{1} & \qw \\
    \lstick{\ket{x_2}} & \qw   & \targ & \qw
}
\]

Composition of gates corresponds to multiplying the relevant unitary matrices.
Continuing with this example, applying this circuit to the input $\ket{x}=\ket{00}$
yields
\[
\left[
\begin{array}{cccc}
1 & 0 & 0 & 0 \\
0 & 1 & 0 & 0 \\
0 & 0 & 0 & 1 \\
0 & 0 & 1 & 0
\end{array}
\right]
\left[
\begin{array}{cccc}
1/\sqrt{2} & 0 & 1/\sqrt{2} & 0 \\
0 & 1/\sqrt{2} & 0 & -1/\sqrt{2} \\
1/\sqrt{2} & 0 & 1/\sqrt{2} & 0 \\
0 & 1/\sqrt{2} & 0 & -1/\sqrt{2}
\end{array}
\right]
\left[
\begin{array}{c}
1\\
0\\
0\\
0\\
\end{array}
\right]
=
\left[
\begin{array}{c}
1/\sqrt{2}\\
0\\
0\\
1/\sqrt{2}\\
\end{array}
\right]
\]
\item The {\em Toffoli} gate, mapping $\ket{x,y,z}$ to $\ket{x,y,z\oplus x y}$.

\item The {\em Phase} gate
$$P =
\left[
\begin{array}{cc}
1 & 0 \\
0 & i
\end{array}
\right]
$$
\end{itemize}

Let $\mathcal{G}$ be a finite set of gates.
$\mathcal{G}$ is {\em universal} if we can use gates from $\mathcal{G}$ to
approximate {\em any} unitary $U$ on any number of qubits to any
desired precision.  To give some examples:

\begin{itemize}
\item \{Toffoli,$H$,Phase\} is universal.
\item \{CNOT, $G$\} is universal for almost any 1-qubit gate $G$---i.e., with
probability 1 over the set of 2-by-2 unitary matrices $G$.
\item In particular, \{CNOT, $G$\} is universal, for
$$G =
\left[
\begin{array}{cc}
3/5 & 4i/5 \\
4/5 & -3i/5
\end{array}
\right]
$$
\item \{Toffoli,$H$\} is {\em not} universal, since these are both matrices
over $\R$.  However, this set is nonetheless sometimes called
``universal for quantum computing'' since it densely generates the real
orthogonal group.  This is reasonable, since one can ``simulate''
the state
$\sum_x \alpha_x \ket{x}$ by
$\sum_x \mathrm{Re}(\alpha_x)\ket{0x} + \mathrm{Im}(\alpha_x)\ket{1x}$.

\item \{CNOT,$H$,Phase\} is {\em not} universal, since it turns out to generate only a
discrete subgroup.  Furthermore, the celebrated \emph{Gottesman-Knill Theorem} states
that circuits over this set of gates can be simulated in classical polynomial time
(and in fact they can be simulated in the complexity class $\mathsf{\oplus L}$ \cite{ag}).
\end{itemize}

See Lecture 10 for more about universality of gate sets.

\subsection{The Solovay-Kitaev Theorem}

We now state one of the most important results about quantum gates.

\begin{theorem}[Solovay-Kitaev (see \cite{solovaykitaev})]
Given any physically universal gate set $\mathcal{G}$ and a target unitary $U$, the
number of $\mathcal{G}$ gates that are needed in order to approximate $U$ to
accuracy $\epsilon$ (entrywise) is $\log^{O(1)}(1/\epsilon)$, omitting dependence on the
number of qubits $n$.
\end{theorem}

It is known that if $\mathcal{G}$ has certain special properties, then the number of
gates drops down to $O(\log 1/\epsilon)$, which is optimal (see
\cite{HRC02}).

\section{Relations to Complexity Classes}

We assume that the reader is familiar with the complexity class
$\mathsf{P}$ (deterministic polynomial time).  The analogous probabilistic
class $\mathsf{BPP}$ consists of all languages $L\subseteq \{0,1\}^*$ for which there's a
Turing machine running in polynomial time, which takes as input a pair
$(x,y)$, where the length of $y$ is bounded by a polynomial in the length
of $x$, such that
\begin{itemize}
\item if $x \in L$ then $\Pr_y[M(x,y)=1] > 2/3$, and
\item if $x \not\in L$ then $\Pr_y[M(x,y)=1] < 1/3$
\end{itemize}
(So in particular, for every $x$, the probability that $M$ accepts $x$ is bounded away from 1/2.)
Languages in $\mathsf{BPP}$ are considered easy to compute; indeed, it's
widely conjectured today that $\mathsf{P} = \mathsf{BPP}$.  (See \cite{IW97}.)

By analogy, Bernstein and Vazirani \cite{bv} defined $\mathsf{BQP}$, or Bounded-Error Quantum Polynomial-Time, as the class of all languages efficiently decidable by quantum computers.

\begin{definition}
$\mathsf{BQP}$ is the class of languages $L \subseteq \{0,1\}^*$ for which there
exists a $\mathsf{P}$-uniform family of polynomial-size quantum circuits $\{C_n\}_{n \geq 1}$
acting on $p(n)$ qubits (for some polynomial $p$) over some finite universal
gate set $\mathcal{G}$ such that, for all $n$ and for all $x \in \{0,1\}^n,$
\begin{eqnarray*}
x \in A \Longrightarrow \Pr[C_n\ \mbox{accepts}\  \ket{x}\otimes  \ket{0}^{\otimes p(n)-n}] \geq 2/3.\\
x \not\in A \Longrightarrow \Pr[C_n\ \mbox{accepts}\  \ket{x}\otimes  \ket{0}^{\otimes p(n)-n}] \leq 1/3.
\end{eqnarray*}
\end{definition}

(The original definition given by Bernstein and Vazirani was in terms of
``quantum Turing machines.''  But this definition, in terms of quantum circuits, is easier to work with and known to be equivalent.)

As a first observation, the Toffoli gate can be used to implement universal classical computation (for example, by simulating a NAND gate).
Hence, we can do error amplification like usual, by repeating an algorithm several times and then outputting the majority vote.
 For this reason, the constants $2/3$ and $1/3$ in the definition of $\mathsf{BQP}$ are arbitrary, and can be replaced
by other constants (or even by, say, $1-2^{-n}$ and $2^{-n}$), exactly like for the classical probabilistic complexity class $\mathsf{BPP}$.

\begin{proposition}
$\mathsf{P} \subseteq \mathsf{BQP}$
\end{proposition}
\begin{proof}
Again, we use \{Toffoli\} as our gate set.  For every polynomial-time
computable function $f$, there is a circuit using only Toffoli gates
implementing the transformation
$\ket{x,0,0^m} \mapsto \ket{x,f(x),0^m}$, for some $m$ bounded by a
polynomial in the length of $x$.\end{proof}

\begin{proposition}
$\mathsf{BPP} \subseteq \mathsf{BQP}$.
\end{proposition}
\begin{proof}
On input $x$, using Hadamard gates, a quantum circuit can generate a sequence
of
random bits $y$, which can then be fed into a circuit that simulates the
(deterministic) operation of a $\mathsf{BPP}$ machine $M$ on input $(x,y)$.
\end{proof}

\subsection{Classical Simulation of Quantum Computation}

Recall that $\mathsf{GapP}$ is the class of functions $f:\{0,1\}^* \longrightarrow \mathbb{Z}$
for which there is a non-deterministic Turing machine $M$ such that $f(x)=A_M(x)-R_M(x)$,
where $A_M(x)$ (resp. $R_M(x)$) is the number of accepting (resp. rejecting) computations of $M$ on
input $x$.  A somewhat more convenient characterization of $\mathsf{GapP}$
is:  $\mathsf{GapP} = \{f-g : f,g \in \mathsf{\#P}$\}.  $\mathsf{GapP}$ has
several useful closure properties.  For instance, if $g(x,y,i)$
is a $\mathsf{GapP}$ function
and $F(x,y,i,a_1,\ldots,a_{q(n)})$ can be computed by using $(x,y,i)$ to
build (in polynomial time) an unbounded fan-in arithmetic formula of depth
$O(1)$ and then applying the formula to the integers
$(a_1,\ldots,a_{q(n)})$ (for some polynomial $q$), the function
\[
h(x) = \sum_{y \in \{0,1\}^{q(n)}} F(x,y,g(x,y,1),g(x,y,2),\ldots , g(x,y,q(n)))
\]
is in $\mathsf{GapP}$.  (In the application below, the function
$g(x,y,1),g(x,y,2),\ldots, g(x,y,q(n))$ will represent the real and
imaginary parts of certain complex numbers, and $F(\ldots)$ will compute the
real (or imaginary) part of their product.  $\mathsf{GapP}$ and its closure
properties are discussed in \cite{FFK94}.)
The class $\mathsf{PP}$ is the class of languages $L$ that are definable
by $\{x:f(x) >0\}$ for some function $f \in \mathsf{GapP}$.

\begin{theorem}
\label{bqppp}
$\mathsf{BQP} \subseteq \mathsf{PP}$
\end{theorem}

\begin{proof}
To simulate a quantum circuit on a given input $x$, it suffices to estimate the probability $p(x)$ that
the final measurement of the output bit gives value 1.
The state of the computation of a quantum circuit after each gate
is described by a vector $\ket{z}$ in the Hilbert space
$\mathbb{C}^{2^{p(n)}}$ where $p(n)$ is the number of qubits that the circuit
acts on (including the input qubits, output qubit, and ancilla qubits).
As usual we write such a state vector in the form $\ket{\alpha} = \sum_{z} \alpha_z|z\rangle$
where $z$ ranges over $\{0,1\}^{p(n)}$.
If $\ket{\phi}=\sum_z \phi_z \ket{z}$ is the state at the conclusion of the computation then
$p(x)$ is the sum of $|\phi_{z}|^2$ over all
$z$ whose output bit is set to 1.

Each gate of the circuit acts on the quantum state via a unitary matrix.
Thus $\ket{\phi} = U_tU_{t-1}\ldots U_1 \ket{\sigma}$ where $\ket{\sigma}=\ket{\sigma(x)}$ is the
initial state of the qubits and $U_j$ is the unitary matrix corresponding to the $j$th gate.  The entries of the matrices $U_j$ are complex, but we will
represent them by pairs of dyadic rational numbers (where these dyadic
rationals are close approximations to the real and imaginary parts of the
actual entries, and we assume without loss of generality that the denominators
of all entries are the same power of two).

If $P$ is the diagonal matrix such that $P_{z,z}$ is equal to the value of the output bit corresponding to basis
state $z$ then the acceptance probability $p(x)$ satisfies:

\[
p(x) = \langle \sigma |U_1^{-1} \cdots U_t^{-1}PU_t \cdots U_1\ket{\sigma},
\]

Using the standard connection between matrix multiplication and counting
the number of paths in a graph, we can observe that $p(x)$ is the sum
(over all paths $p$ from the start vertex to the terminal vertex in the graph)
of the product of the edge weights encountered along path $p$ (where the
edge weights are given by the product of the edge weights.
That is, $p(x) = \sum_p \prod_i (\mbox{weight of the $i$-th edge along path $p$})$.

Each of the matrices appearing in this product is exponentially large, but succinctly described: for any
basis states $y,z$ the $y,z$ entry of any of the matrices can be computed in
time polynomial in $n$.
Furthermore, by the closure properties of $\mathsf{GapP}$ discussed above and
the fact that the entries in each matrix have the same denominator, it follows
that $p(x) = f(x)/2^{q(n)}$ for some polynomial
$q$ and $\mathsf{GapP}$ function $f$.  Thus $p(x) > 1/2$ if and only if
$f(x) - 2^{q(n)+1} > 0$.
Since $f(x) - 2^{q(n)+1}$ is a $\mathsf{GapP}$ function, this completes
the proof.
\end{proof}

It's known that $\mathsf{PP} \subseteq \mathsf{P^{\#P}} \subseteq \mathsf{PSPACE} \subseteq \mathsf{EXP}$,
where $\mathsf{\#P}$ is the class of all combinatorial counting problems, $\mathsf{PSPACE}$ is the class
of problems solvable by a classical computer using polynomial memory, and $\mathsf{EXP}$ is the class of
problems solvable by a classical computer using $2^{n^{O(1)}}$ time. \ So to summarize, we have

$$ \mathsf{P} \subseteq \mathsf{BPP} \subseteq \mathsf{BQP} \subseteq \mathsf{PP} \subseteq \mathsf{P^{\#P}} \subseteq \mathsf{PSPACE} \subseteq \mathsf{EXP}.$$

Or in words: ``quantum computers are at least as powerful as classical computers, and at most exponentially more powerful.''  This means, in particular, that
quantum computers can't solve uncomputable problems like the halting problem; they can ``merely'' solve computable problems up to exponentially faster. It also
means that, in our present state of knowledge, we have no hope of proving unconditionally that $\mathsf{BPP} \ne \mathsf{BQP}$ (i.e., that quantum computers are strictly more powerful than classical probabilistic ones), since any such proof would imply $\mathsf{P} \ne \mathsf{PSPACE}$, which would be essentially as big a deal as a proof of $\mathsf{P} \ne \mathsf{NP}$. The best we can do is get \emph{evidence} for $\mathsf{BPP} \ne \mathsf{BQP}$, e.g.\ by finding examples of $\mathsf{BQP}$ problems that are considered unlikely to be in $\mathsf{BPP}$.

\begin{remark}
It can be shown that one can choose the gates of a quantum circuit so that that resulting
unitaries have all of their entries in $\{0,\pm 1/2,\pm 1 \}$.  In this case the $\mathsf{GapP}$ computation
can be carried out without introducing any approximation errors.
This implies that
the above simulation of $\mathsf{BQP}$ extends to the class $\mathsf{PQP}$ which
is the quantum analog of $\mathsf{PP}$, and shows that $\mathsf{PQP}=\mathsf{PP}$.
\end{remark}

\begin{remark}
It's also known that $\mathsf{BQPSPACE}=\mathsf{PSPACE}$: that is, quantum computers can provide at most a polynomial savings in space complexity.  This is because there's a generalization of Savitch's theorem, by Papadimitriou from 1985 \cite{Pap85}, which gives $\mathsf{PPSPACE}=\mathsf{PSPACE}$.  That immediately implies the desired equivalence. More precisely, we get that any quantum algorithm using space $s$ can be simulated classically using space $s^2$.
Ta-Shma \cite{TS13} provided a candidate example where using a quantum algorithm might indeed provide this quadratic space savings.
\end{remark}

\subsection{Quantum Oracle Access}

As with classical computation, it is illuminating to consider quantum computation relative to an oracle.
As usual, an oracle is an arbitrary language $A \subseteq \{0,1\}^*$ and for any complexity class
$C$, we write $C^A$ for the relativized class in which $C$ has the power to make oracle
calls to $A$ at unit cost. (More generally, we may also consider oracles that are functions
$A:\{0,1\}^* \longrightarrow \{0,1\}^*$.)

In the setting of quantum circuits, an oracle $A$ is accessed using ``oracle gates.''
There are two different definitions of oracle gates, which turn out to be equivalent to each other.
In the first definition, each basis state of the form $\ket{x,z,w}$ gets mapped to $\ket{x,z\oplus A(x),w}$.
Here $x$ is an input, $A(x)$ is the oracle's response to that input, $z$ is an ``answer qubit'' where $A(x)$ gets
written, $\oplus$ denotes XOR (the use of which makes the oracle gate unitary), and $w$ is a ``workspace register,''
consisting of qubits that don't participate in this particular query.
In the second definition, each basis state of the form $\ket{x,z,w}$ gets mapped to $(-1)^{z\cdot A(x)} \ket{x,z,w}$.
In other words, the amplitude gets multiplied by $A(x)$, if and only if $z=1$.
To see the equivalence of the two types of oracle gate, we simply need to observe that either one becomes the other
if we conjugate the $\ket{z}$ register by a Hadamard gate before and after the query.

As usual, oracles provide a lens for exploring the relationships between different complexity classes.
Given two classes $C_1$ and $C_2$ such that $C_1 \subseteq C_2$, one can ask
(a) Does the containment hold relative to every oracle? (b) Are there
oracles with respect to which equality holds? and (c) Are there oracles with
respect to which the containment  is strict?

It is routine to show that the containments $\mathsf{BPP} \subseteq \mathsf{BQP} \subseteq \mathsf{PSPACE}$
hold with respect to any oracle.  Since there are oracles $A$ such that $\mathsf{P}^A=\mathsf{PSPACE}^A$
(e.g. any $\mathsf{PSPACE}$-complete oracle), it also follows that $\mathsf{BPP}^A=\mathsf{BQP}^A$ for those oracles.
This leaves the third question: are there oracles with respect to which $\mathsf{BPP}$ is
strictly contained in $\mathsf{BQP}$?

The answer is yes.  Bernstein and Vazirani \cite{bv} constructed an oracle $A$ and a language $L \in \mathsf{BQP}^A$,
such that any classical randomized algorithm to decide $L$ must make $\Omega(n^{\log n})$ queries to $A$.  This was later
improved by Simon \cite{simon}, who found an exponential separation between quantum and classical query complexities.

In the remainder of this lecture, we'll state Simon's problem and see why a quantum algorithm achieves an exponential
advantage for it. Then, in Lecture 3, we'll put Simon's problem into the more general framework of the {\em Hidden Subgroup Problem}
(HSP).

In Simon's problem, we're given oracle access to a function $f:\{0,1\}^n\longrightarrow \{0,1\}^n$, mapping $n$-bit
inputs to $n$-bit outputs. We're promised that either

\begin{enumerate}
\item[(i)] $f$ is 1-to-1, or
\item[(ii)] there exists a ``secret string'' $s\ne 0^n$, such that for all $x\ne y$, we have $f(x)=f(y)$ if and only if $x\oplus s = y$
(where $\oplus$ denotes bitwise XOR).
\end{enumerate}

\noindent The problem is to decide which, by making as few queries to $f$ as possible.

Using a classical randomized algorithm, it's not hard to see that $\Theta(2^{n/2})$ queries are necessary and sufficient to solve this problem
with high probability.

To see the sufficiency: we just keep querying $f(x)$ for randomly chosen inputs $x$, until we happen to have found a pair $x\ne y$ such that $f(x)=f(y)$, at which
point we know that $f$ is 2-to-1 and can stop. If we {\em don't} find such a collision pair after $2^{n/2}$ queries, then we guess that $f$ is 1-to-1. This works
because of the famous {\em Birthday Paradox}: if $f$ is indeed 2-to-1, then we have constant probability of finding a collision pair after $\sqrt{2^n} = 2^{n/2}$
queries, for exactly the same reason why, in a room with $\sqrt{365}$ people, there's a decent chance that at least two of them share a birthday.

To see the necessity: suppose that with $1/2$ probability, $f$ is a random 1-to-1 function
from $\{0,1\}^n$ to itself, and with $1/2$ probability, it's a random function satisfying condition (ii)
above (with $s\ne 0^n$ chosen uniformly at random).
It's not hard to show that for any choice of queries, the decision rule that minimizes the
overall probability of error is to guess we're in case (i) if and only if the answers to all queries are
different. But if the algorithm makes $q$ queries to $f$, then in case (ii), the probability of answering correctly
is at most the expected number of collision pairs (i.e.,
$x \neq y$ such that $f(x)=f(y)$).  This is $q(q-1)2^{n-1}$, which is $o(1)$ if $q=o(2^{n/2})$.

On the other hand, there's a quantum algorithm for Simon's problem that uses only $O(n)$ queries to $f$, and $n^{O(1)}$ computational steps in total.
That algorithm is as follows:

\begin{enumerate}
\item[(1)] Prepare an equal superposition over all $2^n$ possible inputs to $f$:
$$\frac{1}{\sqrt{2^n}}\sum_{x \in \{0,1\}^n} \ket{x} \ket{0^n}.$$
\item[(2)] Query $f$ in superposition to get
$$\frac{1}{\sqrt{2^n}}\sum_{x \in \{0,1\}^n} \ket{x} \ket{f(x)}.$$
We don't actually care about the value of $f(x)$; all we care about is what the computation of $f(x)$ does to the original $\ket{x}$ register.
\item[(3)] Measure the $\ket{f(x)}$ register (this step is not actually necessary, but is useful for pedagogical purposes).
If $f$ is 1-to-1, then what's left in the $\ket{x}$ register will be a single computational basis state, $\ket{x}$. If $f$ is 2-to-1,
then what's left will be an equal superposition over two inputs $x$ and $y$ with $f(x)=f(y)$ and hence $y=x\oplus s$:
$$\frac{\ket{x}+\ket{y}}{\sqrt{2}}.$$
The question now is: in case (ii), what measurement of the above superposition will tell us anything useful about $s$? Clearly measuring
in the computational basis is pointless, since such a measurement will just return $x$ or $y$ with equal probability, revealing no information about $s$.
If we could measure the state twice, and see $x$ {\em and} $y$, then we could compute $s=x\oplus y$. But alas, we can't measure the state twice. We could repeat the entire computation from the beginning, but if we did that, then with overwhelming probability we'd get a superposition over a different $(x,y)$ pair, so
just like in the classical case, it would take exponentially many repetitions before we'd learned anything about $s$.
\item[(4)] So here's what we do instead. We apply Hadamard gates to each of the $n$ qubits in the $\ket{x}$ register, and only {\em then} measure that register in the computational basis. What does Hadamarding each qubit do? Well, it maps each basis state $\ket{x}$ to
    $$ \frac{1}{\sqrt{2^n}} \sum_{z\in \{0,1\}^n} (-1)^{x\cdot z}\ket{z}. $$
    Therefore it maps $\frac{1}{\sqrt{2}}(\ket{x}+\ket{y})$ to
    $$ \frac{1}{\sqrt{2^{n+1}}} \sum_{z\in \{0,1\}^n} (-1)^{x\cdot z + y\cdot z}\ket{z}. $$
    When we measure the above state, which basis states $z$ have a nonzero probability of being observed? Well, any $z$ for which the two contributions to its amplitude interfere constructively (i.e., are both positive or both negative), rather than interfering destructively (i.e., one is positive and the other is negative).  That means any $z$ such that $$x\cdot z \equiv y\cdot z \pmod 2.$$  That, in turn, means any $z$ such that $$(x\oplus y)\cdot z \equiv 0 \pmod 2$$ --- or equivalently, any $z$ such that $s\cdot z \equiv 0 \pmod 2$.  In conclusion, we don't learn $s$ when we measure, but we do learn a random $z$ such that $s\cdot z \equiv 0$---that is, a random mod-2 linear equation that $s$ satisfies.
\item[(5)] Now all that remains to do is to repeat steps (1) through (4) until we've learned enough about $s$ to determine it uniquely (assuming we're in case (ii)). \ Concretely, every time we repeat those steps, we learn another random linear equation that $s$ satisfies:
    $$s\cdot z_1 \equiv 0 \pmod 2, \quad \quad s\cdot z_2 \equiv 0 \pmod 2, \quad \ldots$$
    We just need to continue until the only two solutions to the linear system are $0^n$ and $s$ itself. A simple probabilistic analysis shows that, with overwhelming probability, this will happen after only $O(n)$ repetitions. Furthermore, we can easily check whether it's happened, using Gaussian elimination, in classical polynomial time.  Thus, our overall algorithm uses $O(n)$ quantum queries to $f$ and polynomial computation time (with the latter dominated by {\em classical} computation).
\end{enumerate}

\lecture{Scott Aaronson}{Michal Kouck\'{y} and Pavel Pudl\'ak}{Hidden Subgroup and Complexity of States and Unitaries}

\section{Hidden Subgroup Problem}

We are going to look at the so-called {\em Hidden Subgroup Problem} $HSP(G)$ over a finite group $G$, which is a generalization of Simon's problem from Lecture 2.
In the HSP, we're given black-box
access to some function $f:G\rightarrow \{0,1\}^*$, which is constant on cosets of some hidden subgroup $H \le G$.
On different cosets the function must take different values. The problem is to determine the subgroup $H$, for example by a list of its generators.
Or, in a common variant of the problem, we just want to determine whether $H$ is trivial or nontrivial.

Simon's algorithm, which we covered in Lecture 2, solves this problem in quantum polynomial time in the special case $G=\Z^n_2$:

\begin{theorem}[\cite{simon}]
$HSP(\Z^n_2) \in \mathsf{BQP}^f$.
\end{theorem}

(Exercise for the reader: check that $HSP(\Z^n_2)$ is just Simon's problem.)

Shor's famous factoring algorithm is directly inspired by Simon's algorithm. Indeed, Shor's algorithm
gives a classical reduction from factoring an integer $N=pq$ to a different problem, {\em Period-Finding} over the integers.

In Period-Finding, we're given black-box access to a function $f:\Z \rightarrow \Z$, and we need to determine the smallest $k>0$ such that $f(x)=f(x+k)$ for all integers $x$. In the case of Shor's algorithm, the function $f$ happens to be $f(x)=a^x \mod N$, for some randomly chosen $a$. Shor makes three crucial observations:

\begin{enumerate}
\item[(1)] The function $f(x)=a^x \mod N$ can be computed in polynomial time, using the well-known trick of {\em repeated squaring}. Thus, in this case,
there's no need for an oracle; the ``oracle function'' can just be computed explicitly.
\item[(2)] The period of $f$ always divides $(p-1)(q-1)$, the order of the multiplicative group modulo $N$. For this reason, if we can solve Period-Finding for various choices of $a$, that gives us enough information to determine $p$ and $q$. (This step is just classical number theory; it has nothing to do with quantum computing.)
\item[(3)] We can solve Period-Finding in quantum polynomial time, indeed with only $O(1)$ queries to $f$. This is the ``hard'' step, and the only step that actually uses quantum mechanics. Much of the technical difficulty comes from the fact that $\Z$ is infinite, so we have to work only with a small
finite part of $\Z$ by cutting it off somewhere. However, while the details are a bit complicated, conceptually the algorithm is extremely similar to Simon's algorithm from Lecture 2---just with a secret period $r$ in place of the ``secret string'' $s$, and with a discrete Fourier transform in place of Hadamarding each qubit.
\end{enumerate}

This leads to:

\begin{theorem}[\cite{shor}]
Integer Factoring (when suitably phrased as a decision problem) is in $\mathsf{BQP}$.
\end{theorem}


Shortly afterward, Kitaev generalized both Simon's and Shor's algorithms, to conclude that $HSP(G)$ can be solved in quantum polynomial time for {\em any} abelian group $G$.

\begin{theorem}[\cite{kitaev:meas}]
For any finite abelian group $G$, $HSP(G) \in \mathsf{BQP}^f$.
\end{theorem}

Inspired by the success with abelian groups, many researchers tried to design quantum algorithms for the HSP over various non-abelian groups. Unfortunately,
after twenty years of effort, this program has had very limited success, and almost all of it is for groups that are ``close'' in some way to abelian groups (for example, the Heisenberg group).

Why do people care about non-abelian HSP?  One major reason is that Graph Isomorphism can be reduced to HSP over the symmetric group $S_n$.
 A second reason is that, as shown by Regev \cite{Regev04}, Approximate Shortest Lattice Vector
 can ``almost'' be reduced to HSP over the dihedral group $D_N$. (Here the ``almost'' is because Regev's reduction assumes that the HSP algorithm works via ``coset sampling'': the approach followed by Simon's algorithm, Shor's algorithm, and essentially all other known HSP algorithms.)

Now, Graph Isomorphism and Approximate Lattice Vector are the two most famous ``structured'' $\mathsf{NP}$ problems: that is, problems that are not known to be in $\mathsf{P}$, but also have strong theoretical reasons not to be $\mathsf{NP}$-complete, and that therefore plausibly inhabit an ``intermediate zone'' between $\mathsf{P}$ and $\mathsf{NP}$-complete. And assuming tentatively that $\mathsf{NP} \not\subset \mathsf{BQP}$, which many people believe, these ``structured'' $\mathsf{NP}$ problems constitute a large fraction of all the interesting problems for which one could try to design an efficient quantum algorithm.

The reduction from Graph Isomorphism ($GI$) to $HSP(S_n)$ is  simple and works as follows. Fix graphs $A$ and $B$
on $n$ vertices. We define a function $f$ from  $S_{2n}$ into graphs on $2n$ vertices as follows: $f(\sigma)$ is the graph obtained by permuting by $\sigma$ the disjoint union of
$A$ and $B$. Determining whether the hidden subgroup of $f$ can interchange the vertices of $A$ and $B$ answers whether $A$ and $B$
are isomorphic. Since $f$ can be computed easily given $A$ and $B$, a polynomial-time quantum algorithm for $HSP(S_n)$ would imply $GI\in \mathsf{BQP}$.

(Of course, now we have Babai's {\em classical} quasi-polynomial algorithm for $GI$ \cite{Babai16}, which implies a quantum quasi-polynomial algorithm for the same problem!  But besides reducing quasipolynomial to polynomial, a quantum algorithm for $HSP(S_n)$ would also yield solutions to numerous {\em other} isomorphism problems, like isomorphism of rings, multivariate polynomials, and linear codes, which Babai's algorithm doesn't address.)

Meanwhile, a {\em lattice} is a set of vectors in $\R^n$ that's closed under integer linear combinations.  In the {\em Approximate Shortest Vector Problem}, the goal is to find a nonzero vector in a lattice $L$ that's at most a factor (say) $\sqrt{n}$ longer than the shortest nonzero vector in $L$.
This problem is closely related to solving $HSP$ over the dihedral group $D_n$.  A fast quantum algorithm for Approximate Shortest Vector would
mean that many of the public-key cryptographic schemes that weren't yet broken by Shor's factoring algorithm would be broken by quantum computers as well.
So for example, the security of the Regev and Ajtai-Dwork cryptosystems, as well as Gentry's fully homomorphic encryption \cite{gentry}, are based on the hardness of Approximate Shortest Vector or variants thereof.

Why might one even hope for a fast quantum algorithm to solve $HSP$ over nonabelian groups? In 1997, Ettinger, H\o yer and Knill \cite{ehk} showed that
at least from the perspective of query complexity, quantum computing is powerful enough. Indeed, using only polynomially
(in $\log \left|G\right|$) many queries to $f$, a quantum algorithm can extract enough information to solve $HSP(G)$ for {\em any} finite group $G$, abelian or not.
In Ettinger-H\o yer-Knill's algorithm, the post-processing after the queries to $f$ takes exponential time, but doesn't require any extra queries to $f$.

\begin{theorem}[Ettinger, H\o yer, Knill 1997 \cite{ehk}]
\label{ehkthm}
For every group $G$, $HSP(G)$ is solvable with $O(\log^2|G|)$ quantum queries to $f$. Furthermore, the queries can be efficiently constructed.
\end{theorem}
\begin{proof}
Take
\[
\frac 1{\sqrt{|G|}} \sum_{x\in G}\ket{x}\ket{f(x)}.
\]
Measure $\ket{f(x)}$. We get, for some $g\in G$,
\[
\frac 1{\sqrt{|H|}} \sum_{h\in H}\ket{gh} .
\]
Call this state $\ket{gH}$. More precisely, we'll have a mixed state given by the density matrix
\[
\rho_H:=\Exp_g{[\ketbra{gH}{gH}]}.
\]
For $k=O(\log^2|G|)$, repeat this procedure $k$ times to get $\rho_H^{\otimes k}$. This is the only part that queries $f$.

\begin{claim}
$\rho_H^{\otimes k}$ information-theoretically determines $H$.
\end{claim}
This follows from the fact that $G$ has $\leq |G|^{\log|G|}$ subgroups (because every subgroup has at most $\log|G|$ generators).

\bigskip
Let  $H_1,\dots,H_r$ be a listing of all subgroups in decreasing order. Define measurements $M_i$ such that $M_i$ accepts $\rho_{H_i}$ and rejects $\rho_{H_j}$ for $j>i$ with probability
\[
1-\frac{1}{|G|^{\log|G|}}.
\]
Define $M_i$ using the projection operator on the space spanned by $\{\ket{gH_i}\ ;\ g\in G\}$. Then if $H\geq H_i$ it accepts with certainty. Otherwise it rejects with probability at least $1/2$. (This is because $|\braket{H_i}{H_j}|\le \frac{1}{\sqrt{2}}$ as can be easily verified.) If we use $\rho_H^{\otimes k}$, then we reject with probability
$$1-\frac 1{2^k}\geq 1-\frac{1}{|G|^{\log|G|}}$$
in the second case. Note that if $i$ is the first one such that $H\geq H_i$, then $H=H_i$.
By the Quantum Union Bound (Lemma \ref{lem:qunionbound}), we can do all measurements $M_i$ and almost preserve the state until some $M_i$ accepts.
\end{proof}

\section{Circuit Complexity}

Can the EHK measurement be performed by a polynomial size quantum circuit? Could it be done under some wild complexity assumption? For example, if $\mathsf{P}=\mathsf{PSPACE}$, then does it follow that $HSP(G) \in \mathsf{BQP}^f$ for all finite groups $G$?

(This question is interesting only for black-box access to $G$, as $HSP(G)$ is clearly in $\mathsf{NP}$ for explicit presentations of $G$.)

To investigate such questions, let's now formally define a measure of quantum circuit complexity.

\begin{definition}
Given an $n$-qubit unitary transformation $U$, we let ${\cal C}_{\epsilon}(U)$ be minimum size of a quantum circuit (say, over $\{H,P,\mbox{Toffoli}\}$) needed to implement $U$ to entry-wise precision $\epsilon$ (say, $\epsilon=2^{-n}$).
\end{definition}

What can we say about this measure?  Given a Boolean function $f:\{0,1\}^N\to\{0,1\}$, let $U_f:\ \ket{x,a}\mapsto \ket{x,a\oplus f(x))}$.  Then the following can be established by a standard counting argument.

\begin{observation}
There are $2^{2^n}$ Boolean functions $f$, but only $\leq T^{O(T)}$ quantum circuits with $\leq T$ gates.  Hence, for almost every $f$,
we must have ${\cal C}_{\epsilon}(U_f) = \Omega\left(\frac{2^n}n\right)$.
\end{observation}

The above observation also holds if we relax the error $\epsilon$ in approximating $U$ to, say, $1/3$. The counting argument works because we have a discrete set of gates; for a continuous set of gates, we would instead use a dimension argument (as mentioned in Lecture 2).

\begin{observation}
If $f\in \mathsf{BQP}$, then ${\cal C}_{\epsilon}(U_f)\leq n^{O(1)}$. Thus, strong enough lower bounds on ${\cal C}_{\epsilon}(U_f)$ would imply ordinary complexity
class separations.
\end{observation}

\begin{question}
Are there unitaries $U$ for which ${\cal C}_{\epsilon}(U) > n^{\omega(1)}$ would {\em not} imply anything about standard complexity classes (e.g., $\mathsf{P} \neq \mathsf{PSPACE}$)---that is, for which we have a hope of proving superpolynomial lower bounds on ${\cal C}_{\epsilon}(U)$ {\em unconditionally}?
\end{question}

The above question is really asking: how hard is it to prove exponential lower bounds on quantum circuit complexity for explicitly-given unitary matrices?  To be more concrete, we could for example ask: is there a ``natural proof'' barrier \cite{rr} for this lower bound problem?  By this we mean: are there polynomial-size quantum circuits that give rise to $2^n \times 2^n$ unitary matrices
that can't be distinguished from Haar-random $U$, say by any $2^{O(n)}$-time classical algorithm with access to the entries of the matrix?  A natural guess is that such circuits exist, as a randomly-chosen quantum circuit seems to produce such a matrix. But it remains open whether this could be shown under some standard cryptographic assumption.

We can also study the quantum circuit complexity of unitary transformations relative to an oracle $A$:

\begin{definition}
We let ${\cal C}^A_{\epsilon}(U)$ be the minimum size of a quantum circuit with $A$-oracle gates that $\epsilon$-approximates $U$.
\end{definition}

We can now state one of Scott's favorite open problems about the complexity of unitary transformations, a problem that will show up over and over in this course:

\begin{question}
{\bf The Unitary Synthesis Problem.} Is it true that for every $n$-qubit unitary transformation $U$, there exists an oracle $A$ such that ${\cal C}^A_{\epsilon}(U)\leq n^{O(1)}$? (Here $A$ is a Boolean function, which can take inputs that are more than $n$ bits long---as it must, for counting reasons.)
\end{question}

What we're here calling the Unitary Synthesis Problem was first raised in a 2006 paper of Aaronson and Kuperberg \cite{ak}.  They conjectured that the answer is no: that is, that there exist unitaries (for example, Haar-random unitaries) that can't be implemented in $\mathsf{BQP}^A$ for any Boolean oracle $A$. However, the best they were able to prove in that direction was the following:

\begin{theorem}[Aaronson-Kuperberg \cite{ak}]
There exist $n$-qubit unitaries $U$ such that, for all oracles $A$, a polynomial-time quantum algorithm cannot implement $U$ with just one query to $A$, assuming that the algorithm must implement {\em some} unitary matrix on the $n$ qubits in question regardless of which $A$ it gets.
\end{theorem}

\section{Quantum State Complexity}

Now, let's look at the complexity of generating a prescribed quantum state, as opposed to implementing a prescribed unitary matrix.

\begin{definition}
Given an $n$-qubit state $\ket{\psi}$, we let $\cal C_{\epsilon}(\ket{\psi})$ be the minimum size of a quantum circuit (over the set $\{H,P,\mbox{Toffoli}\}$) that maps $\ket{0}^{\otimes m}$, for some $m>n$, to a state $\rho$ such that
\[
\|\rho-\ketbra{\psi}{\psi}\otimes \ketbra{0\cdots 0}{0 \cdots 0} \|_{tr}\leq\epsilon .
\]

If we allow garbage in the ancilla bits (and only consider the trace distance between the first $n$ qubits and $\ketbra{\psi}{\psi}$), then we call the measure $\cal C^*_{\epsilon}(\ket{\psi})$.
\end{definition}

For general circuits, it's not known whether there's a separation between $\cal C$ and $\cal C^*$, or whether it's always possible to remove garbage.  We know how to remove garbage in situations like this:
\[
 \ket{x}\ket{\mathrm{garbage}(f(x))}\ket{f(x)}
\]
by copying the classical bit $f(x)$
\[
 \ket{x}\ket{\mathrm{garbage}(f(x))}\ket{f(x)}\ket{f(x)}
\]
and then undoing the computation of $f$ to get $\ket{x}\ket{00 \cdots 0}\ket{0}\ket{f(x)}$.
However, in the case of computing a quantum state $\ket{\psi_x}$,
\[
 \ket{x}\ket{\mathrm{garbage}(\psi_x)}\ket{\psi_x},
\]
we don't know how to remove the garbage, since by the No-Cloning Theorem we can't copy the {\em quantum} state $\ket{\psi_x}$.

\begin{question}
Is there any plausible candidate for a separation between $\cal C_{\epsilon}(\ket{\psi})$ and $\cal C^*_{\epsilon}(\ket{\psi})$?
\end{question}

We now make some miscellaneous remarks about quantum state complexity.

First, just like with unitaries, it's easy to show, by a counting argument, that almost all $n$-qubit states $\ket{\psi}$ satisfy $\mathcal{C}_{\epsilon}(\ket{\psi}) = 2^{\Omega(n)}$.  The challenge is whether we can prove strong lower bounds on the circuit complexities of {\em specific} states.  (But we'll also see, in Lecture 3, how this challenge differs from the analogous challenge for unitaries, by being more tightly connected to ``standard'' complexity theory.)

Second, any state with large circuit complexity must clearly be highly entangled, since (for example) any $n$-qubit separable state $\ket{\psi}$ satisfies $\mathcal{C}_{\epsilon}(\ket{\psi}) = O(n)$.  However, it's important to understand that the converse of this statement is false; {\em complexity is not at all the same thing as entanglement.}  To illustrate, the state
\[
\left(\frac{\ket{00}+\ket{11}}{\sqrt 2}\right)^{\otimes n}
\]
has a maximal entanglement for any state with the same number of qubits, but its complexity is trivial (only $O(n)$).

Fourth, it's probably easy to prepare quantum states that have exponential circuit complexity, using only polynomial time!  The difficulty is merely to prepare the {\em same} complex state over and over again.  To illustrate this point, consider the following procedure: apply a random quantum circuit to $n$ qubits, initially set to $\ket{0}^{\otimes n}$. Then measure
the first $n/2$ qubits in the computational basis, obtaining some outcome $x\in \{0,1\}^{n/2}$.  Let $\ket{\psi_x}$ be the resulting state of the remaining $n/2$ qubits.  As far as anyone knows, this $\ket{\psi_x}$ will satisfy $\mathcal{C}_{\epsilon}(\ket{\psi_x}) = 2^{\Omega(n)}$ with overwhelming probability.

Fifth, here's an extremely interesting open question suggested by Daniel Roy.

\begin{question}
Do there exist $\ket{\psi},\ket{\phi}$ such that
$\braket{\psi}{\phi}=0$,
\[
{\cal
  C}\left(\frac{\ket{\psi}\ket{\psi}+\ket{\phi}\ket{\phi}}{\sqrt{2}}
\right)\leq n^{O(1)},
\]
but
\[
{\cal C}(\ket{\psi}),{\cal C}(\ket{\phi})> n^{O(1)}\ ?
\]
\end{question}

The copies of  $\ket{\psi}$ and $\ket{\phi}$ are important because otherwise the problem is trivial: let $\ket{\psi}$ satisfy ${\cal C}(\ket{\psi}) > n^{O(1)}$.
Then

$$\frac{\ket{0}^{\otimes n}+\ket{\psi}}{\sqrt{2}}, \frac{\ket{0}^{\otimes n}-\ket{\psi}}{\sqrt{2}}$$

\noindent both have exponential quantum circuit complexity, but a linear combination of the two (namely $\ket{0}^{\otimes n}$) is trivial.

Sixth, here's a useful, non-obvious proposition about the behavior of quantum circuit complexity under linear combinations (a version of this proposition is proven in Aaronson \cite{aar:mlin}, but probably it should be considered ``folklore'').

\begin{proposition}
\label{lincomb}
If $\braket{\psi}{\phi}=0$, then
\[
{\cal C}(\alpha\ket{\psi}+\beta\ket{\phi})\leq
O({\cal C}(\ket{\psi})+{\cal C}(\ket{\phi})+n).
\]
\end{proposition}
\begin{proof}
It's easy to prepare
\[
\alpha\ket{0}\ket{\psi}+\beta\ket{1}\ket{\phi}.
\]
(We're ignoring the ancilla bits that should be put to $0$.)
Let $U$ and $V$ denote the unitary operators used to
construct $\ket{\psi}$ and $\ket{\phi}$ respectively. Then the above
state is
\[
\alpha\ket{0} U \ket{0}^{\otimes n}+\beta \ket{1} V \ket{0}^{\otimes n}.
\]
Applying $U^{-1}$, we get
\begin{equation}\label{e1}
\alpha \ket{0} \ket{0}^{\otimes n}+\beta \ket{1} U^{-1}V \ket{0}^{\otimes n}.
\end{equation}
According to our assumption, we have
\[
0=\braket{\psi}{\phi}=\bra{0}^{\otimes
  n} U^\dagger V\ket{0}^{\otimes n} = \bra{0}^{\otimes n} U^{-1}V \ket{0}^{\otimes n}.
\]
Hence the amplitude of $\ket{0}^{\otimes n}$ in
$U^{-1}V\ket{0}^{\otimes n}$ is $0$. Let $W$ be the unitary
operator that implements the reversible Boolean function
\[
(a_0,a_1,\dots,a_n)\mapsto (a_0\oplus OR(a_1,\dots,a_n),a_1,\dots,a_n).
\]
Applying $W$ to (\ref{e1}) we switch $1$ to $0$:
\[
\alpha \ket{0} \ket{0}^{\otimes n}+\beta \ket{0} U^{-1}V \ket{0}^{\otimes n}.
\]
Applying $U$ we get
\[
\alpha \ket{0} U \ket{0}^{\otimes n}+\beta \ket{0} V \ket{0}^{\otimes n},
\]
which is the state we needed (with an extra ancilla bit $\ket{0}$).
\end{proof}

We remark that Proposition \ref{lincomb} can be generalized to any pair of states $\ket{\psi},\ket{\phi}$ satisfying $\left\|\left\langle \psi | \phi \right\rangle \right\| < 1-\eps$, although in that case we pick up a multiplicative $O(1/\eps)$ factor in the circuit complexity.

\subsection{State Versus Unitary Complexity}

So far, we discussed the problem of determining the circuit complexity of a given $n$-qubit unitary transformation, and we introduced the Unitary Synthesis Problem: for every unitary $U$, does there exist an oracle $A$ such that $U$ can be implemented in $n^{O(1)}$ time with oracle access to $A$?  We also discussed the closely-related question of whether, if (say) $\mathsf{P}=\mathsf{PSPACE}$, large classes of explicitly-describable $n$-qubit unitary transformations would then have $n^{O(1)}$ quantum circuit complexity.

We then moved on to discuss the complexity of quantum {\em states}.  Let's now show that, in the setting of quantum state complexity, the two problems mentioned in the previous paragraph both have positive answers.

\begin{proposition}
\label{canmake}
For every $n$-qubit state $\ket{\psi}$, there exists an oracle function $A:\{0,1\}^*\longrightarrow \{0,1\}$ such that $\sizee^A(\ket{\psi}) \leq n^{O(1)}$.
\end{proposition}
 \begin{proof}
 Suppose the state is
 $$\ket{\psi} = \Sigma_{x \in \{0,1\}^n} \alpha_x \ket{x}.$$
 For initial prefixes of $x$ of length $w$,  $|w| <n$, define $\beta_w:= \sqrt{\Sigma_{y} |\beta_{w y }|^2}$.
Suppose we knew $\beta_0$  and $\beta_1$, square roots of probabilities that the first qubit is either 0 or 1. Then we could prepare  $\beta_0 \ket{0} + \beta_1 \ket{1}$.

  Next, recurse: conditioned on the first qubit being $\ket{0}$, we can put the second qubit in an appropriate state. Continue in this way.  Map $\beta_0 \ket{0}+\beta_1\ket{1}$ to
 \[
 \beta_0 \ket{0} \left(\frac{\beta_{00}}{\beta_0} \ket{0} + \frac{\beta_{01}}{\beta_0} \ket{1}\right)+ \beta_1 \ket{1} \left(\frac{\beta_{01}}{\beta_0} \ket{0} + \frac{\beta_{11}}{\beta_1} \ket{1}\right),
 \]
and so on.  Then, as a last step, apply phases: letting $\gamma_x := \alpha_x / \| \alpha_x \|$, map each basis state $\ket{x}$ with $\alpha_x \neq 0$ to $\gamma_x \ket{x}$.

 All we need now is an oracle $A$ that encodes all of the $\beta_x$'s and $\gamma_x$'s.
  \end{proof}

Proposition \ref{canmake} illustrates the difference between the complexity of states and the complexity of unitaries.  Some people assert that the two must be equivalent to each other, because of the so-called {\em Choi-Jamiolkowski isomorphism}, according to which the maximally-entangled state

\[
\ket{\psi_U} := \frac{1}{\sqrt{N}} \sum_{i=1}^N \ket{i} U \ket{i}
\]
encodes {\em all} the information about the unitary transformation $U$.  However, the issue is that there might be easier ways to prepare $\ket{\psi_U}$, other than by first preparing $\ket{\psi_I}$ and then applying $U$ to the second register.  To illustrate this, let $\sigma:\{0,1\}^n\rightarrow \{0,1\}^n$ be a {\em one-way permutation}---that is, a permutation that's easy to compute but difficult to invert---and consider the state
$$ \ket{\phi} = \frac{1}{\sqrt{2^n}} \sum_{x\in \{0,1\}^n} \ket{x} \ket{\sigma(x)}. $$
This state is easy to make if we can first prepare $\ket{x}$ in the first register and then, conditioned on that, prepare $\ket{\sigma(x)}$ in the second register.  By contrast, suppose we needed to start with $\sum_x \ket{x}\ket{x}$, and then map that to $\ket{\phi}$ by applying a unitary transformation to the second register only (with no access to the first register).  To do that, we'd need the ability to pass reversibly between $\ket{x}$ and $\ket{\sigma(x)}$, and therefore to invert $\sigma$.  (See also Lecture 7, where this distinction shows up in the context of AdS/CFT and wormholes.)

As a further observation, given any $n$-qubit unitary transformation $U$, there's an oracle $A$, relative to which a polynomial-time quantum algorithm can apply an $n$-qubit unitary $V$ that agrees with $U$ on any $n^{O(1)}$ rows. (The proof of this is left as an exercise.)

An interesting open problem is whether Proposition \ref{canmake} is tight in terms of query complexity:

\begin{question}
For every $n$-qubit state $\ket{\psi}$, does there exist an oracle $A$ such that $\ket{\psi}$ can be prepared by a polynomial-time quantum algorithm that makes $o(n)$ queries to $A$?  What about $1$ query to $A$?
\end{question}

Proposition \ref{canmake} has the following corollary.

\begin{corollary}
\label{makecor}
If $\mathsf{BQP}=\mathsf{P^{\#P}}$, then $\sizee(\ket{\psi}) \leq n^{O(1)}$ for essentially all the explicit states $\ket{\psi}$ we'll talk about.
\end{corollary}
\begin{proof}
For any $\ket{\psi}$ whose amplitudes are computable in $\mathsf{\#P}$, we can use the recursive construction of Proposition \ref{canmake} to compute the $\beta_x$'s and $\gamma_x$'s ourselves.
\end{proof}

Corollary \ref{makecor} means that, if we want to prove $\sizee(\ket{\psi}) > n^{O(1)}$ for virtually any interesting state $\ket{\psi}$, then at a minimum, we'll also need to prove $\mathsf{BQP} \neq \mathsf{P^{\#P}}$.  For this reason, it might be easier to prove superpolynomial lower bounds on the complexities of unitaries than on the complexities of states.

\lecture{Scott Aaronson}{Valentine Kabanets and Antonina Kolokolova}{QSampling States and $\mathsf{QMA}$}

\section{QSampling States}

In theoretical computer science, an {\em efficiently samplable distribution} means a probability distribution over $n$-bit strings that we can sample, either exactly or approximately, by a classical algorithm running in $n^{O(1)}$ time.  (Here we'll assume we're in the setting of nonuniform complexity.)  To give two examples, the uniform distribution over perfect matchings of a bipartite graph $G$, and the uniform distribution over points in a polytope in $\R^n$ (truncated to some finite precision), are both known to be efficiently samplable.

Now let $D=\{p_x\}_x$ be an efficiently samplable distribution. Then we can associate with $D$ the quantum state
\[
\ket{\psi_D} = \sum_{x \in \{0,1\}^n} \sqrt{p_x} \ket{x},
\]
which following Aharonov and Ta-Shma \cite{at}, we call the {\em QSampling state} for $D$.  We can then ask about the quantum circuit complexity of preparing $\ket{\psi_D}$.

We do know that there's a classical polynomial-time algorithm that takes as input a random string $r$, and that outputs a sample $x_r$ from $D$.  Suppose $\|r\| = m = n^{O(1)}$.  Then by uncomputing garbage, it's not hard to give a polynomial-size quantum circuit that prepares a state of the form
\[
\frac{1}{\sqrt{2^m}} \sum_{r \in \{0,1\}^m} \ket{r} \ket{x_r}.
\]
Crucially, though, the above state is not the same as $\ket{\psi_D}$!  The trouble is the $\ket{r}$ register, which acts as garbage---and not only that, but necessarily decohering garbage, for which we might introduce the term {\em nuclear waste} (to distinguish it from garbage unentangled with the $\ket{x_r}$ register).

Aharonov and Ta-Shma \cite{at} gave some nontrivial examples where a QSampling state can be prepared by a polynomial-size quantum circuit, {\em without} the ``nuclear waste.''  For example, they showed that given as input a bipartite graph $G$, one can efficiently prepare (close to) an equal superposition over all the perfect matchings of $G$.  The proof of that
required examining the details of the Markov Chain Monte Carlo algorithm of Jerrum, Sinclair, and Vigoda \cite{jsv}, which is what showed how to {\em sample} a perfect matching.  In other cases, however, we might have distributions that are samplable but not QSamplable (though every QSamplable distribution is samplable).

But why do we even care about the distinction between the two?  Why is getting rid of ``nuclear waste'' important?  Well, consider the following example.

Let $G$ be a graph on $n$ vertices. Let $D_G$ be the uniform distribution over all permutations of vertices of $G$:
\[
\{\sigma(G)\}_{\sigma \in S_n}.
\]
The quantum version of this distribution is:
\[
\ket{\psi_{D_G}} = \frac{1}{\sqrt{n!/|Aut(G)|}} \cdot \sum_{\sigma} \ket{\sigma(G)},
\]
where $\sigma(G)$ represents the result of applying the permutation $\sigma$ to the vertices of $G$, and the sum includes only one $\sigma$ in each orbit.

(If $G$ is rigid---i.e., has no non-trivial automorphisms---then $\ket{\psi_{D_G}}$ will be a superposition over all $n!$ permutations of the $G$.  Otherwise, it will be a superposition over $n!/|Aut(G)|$ permutations, since there will be collisions.)

\begin{claim}
\label{gibqp}
If  $\sizee(\ket{\psi_{D_G}}) \leq n^{O(1)}$ via a uniform polynomial time algorithm, then Graph Isomorphism is in $\mathsf{BQP}$.
\end{claim}

\begin{proof}
Given $n$-vertex graphs $G$ and $H$, we need to decide whether $G \cong H$.  Suppose we can prepare $\ket{\psi_{D_G}}$ and $\ket{\psi_{D_H}}$.  Then we can also prepare
\[
\frac{\ket{0} \ket{\psi_{D_G}} + \ket{1} \ket{\psi_{D_H}}}{\sqrt{2}}.
\]
After doing so, apply a Hadamard gate to the first qubit to switch to the $\ket+, \ket-$ basis, and then measure the first qubit.  As happens over and over in quantum information, we can use this trick to test whether the states $\psi_{D_G}$ and $\psi_{D_H}$ are equal or orthogonal.

In more detail, if $\ket{\psi_{D_G}} = \ket{\psi_{D_H}}$, then we obtain $\ket+\tensor\ket{\psi_{D_G}}$. Therefore we always see $\ket+$. Otherwise, if $\ket{\psi_{D_G}}$ and $\ket{\psi_{D_H}}$ are orthogonal, then we obtain the maximally mixed state in the first qubit.  So in this case, we see a random bit.

If $G \cong H$, then $\ket{\psi_{D_G}}=\ket{\psi_{D_G}}$, and otherwise these states are orthogonal, because they are superpositions over disjoint sets.
\end{proof}

By contrast, note that if all we knew how to prepare was the ``nuclear waste state''

$$\frac{1}{\sqrt{2^n}} \sum_r \ket{r} \ket{x_r},$$

\noindent then the two states could be orthogonal even when the graphs were not isomorphic.  This is why we need QSampling.

Note that, even given Babai's breakthrough $n^{(\log n)^{O(1)}}$ algorithm for Graph Isomorphism \cite{Babai16}, it's still unclear how to prepare the QSampling states for GI in less than exponential time.

\begin{observation}
Suppose $\mathsf{NP} \subseteq \mathsf{BQP}$.  Then $\mathcal{C}_{1/n^{O(1)}} (\ket{\psi_D}\leq n^{O(1)}$ for every samplable distribution $D$.  In other words, there is a weaker assumption than the $\mathsf{BQP} = \mathsf{P^{\#P}}$ of Corollary \ref{makecor}, which would make QSampling states easy to generate to within $1/n^{O(1)}$ precision.
\end{observation}
\begin{proof}
The approximate counting result of Stockmeyer \cite{DBLP:conf/stoc/Stockmeyer83} says that, in $\mathsf{BPP^{NP}}$, one can approximate any $\#\mathsf{P}$ function to within a small multiplicative error. The $\beta_x$ numbers from the recursive sampling procedure of Proposition \ref{canmake} are in $\#\mathsf{P}$.
\end{proof}

Let's now give a substantial generalization of Claim \ref{gibqp}, due to Aharonov and Ta-Shma \cite{at}.

\begin{theorem}[Aharonov and Ta-Shma \cite{at}]\label{ATS}
Suppose QSampling states can be prepared in quantum polynomial time.  Then $\mathsf{SZK} \subseteq \mathsf{BQP}$.
\end{theorem}

Here $\mathsf{SZK}$ is Statistical Zero-Knowledge: roughly speaking, the class of all languages $L\subseteq \{0,1\}^*$ for which there exists a probabilistic protocol between Arthur and Merlin (a polynomial-time verifier and an all-powerful but untrustworthy prover), whereby Arthur can be convinced that the input is in $L$ without learning anything about the proof.  The standard example of an $\mathsf{SZK}$ protocol is the following protocol for Graph Non-Isomorphism:
\begin{quote}
Suppose Merlin wants to convince Arthur that two graphs, $G$ and $H$, are not isomorphic.  Then Merlin says to Arthur: pick one of the graphs uniformly at random, randomly permute its vertices, send it to me and I will tell you which graph you started with. If Merlin always returns the right answer, then the graphs are almost certainly not isomorphic---because if they were, then Merlin would give the right answer in each such test with probability only 1/2.
\end{quote}

The class of all problems that have a protocol like that one is $\mathsf{SZK}$.

\begin{theorem}[Sahai-Vadhan \cite{DBLP:journals/jacm/SahaiV03}]
\label{svthm}
The following problem, called Statistical Difference, is a complete promise problem for $\mathsf{SZK}$:
Given efficiently samplable distributions $D_1$, $D_2$ and promised that either $\|D_1 -D_2\| \leq a$ or $\|D_1 -D_2\| \geq b$ (where, say, $a=0.1$ and $b=0.9$), decide which.
\end{theorem}

For the proof of Theorem \ref{svthm}, we need $b^2 >a$.  So the ``default constants of theoretical computer science,'' namely $1/3$ and $2/3$, turn out to work, but this is just a lucky coincidence!

\begin{remark}
The Hidden Subgroup Problem, for any group $G$, can be shown to be in $\mathsf{SZK}$.  But $\mathsf{SZK}$ seems more general than $HSP$: an $\mathsf{SZK}$ problem need not have any group-theoretic structure.
\end{remark}

\begin{proof}[Proof of Theorem \ref{svthm}]
Because of \cite{DBLP:journals/jacm/SahaiV03}, it suffices for us to say:
Suppose we could prepare the states
\[
\ket{\psi_{D_1}} = \sum_{x \in \{0,1\}^n} \sqrt{p_x} \ket{x},
\]
and
\[
\ket{\psi_{D_2}} = \sum_{x \in \{0,1\}^n} \sqrt{q_x} \ket{x}.
\]
Then as before, make the state
\[
\frac{\ket0\ket{\psi_{D_1}}+\ket1\ket{\psi_{D_2}}}{\sqrt{2}}
\]
and measure the control qubit in the $\ket+,\ket-$ basis.  The probability to get the outcome $\ket+$ is then
\[
\Pr[\ket+] = \frac{1+\mathrm{Re}\braket{\psi_{D_1}}{\psi_{D_2}}}{2}.
\]

Interestingly, the absolute phases of $\ket{\psi_{D_1}}$ and $\ket{\psi_{D_2}}$ (as determined by the procedures that prepare them starting from $\ket{0}^{\otimes n}$) could actually matter in this test!  But in any case, we can still use the test to determine if the inner product is close to $1/2$ (as it is in the orthogonal case) or far from it.

There's no classical analogue for this interference test. Interference matters: applying a Hadamard gate to the control qubit gives
\[
\ket0\ket{\psi_{D_1}}+\ket1\ket{\psi_{D_2}}+\ket0\ket{\psi_{D_2}}-\ket1\ket{\psi_{D_1}}.
\]
Cancellation happens only if $\ket{\psi_{D_1}}=\ket{\psi_{D_2}}$.

We have
\[
\mathrm{Re}\braket{\psi_{D_1}}{\psi_{D_2}} = \sum_x \sqrt{p_x q_x} = F(D_1,D_2),
\]
where the fidelity $F(D_1,D_2)$ satisfies the inequalities:
\[
1-||D_1 - D_2|| \leq F(D_1,D_2) \leq \sqrt{1-||D_1-D_2||^2}.
\]
Thus, if $||D_1-D_2|| \approx 1$, then $\braket{\psi_{D_1}}{\psi_{D_2}} \approx 0$, while if $||D_1-D_2|| \approx 0$ then $\braket{\psi_{D_1}}{\psi_{D_2}} \approx 1$.
\end{proof}

In 2002, Scott proved that at least in the black-box model, there is no efficient quantum algorithm for $\mathsf{SZK}$ problems.

\begin{theorem}[\cite{aar:col}]\label{A03}
There exists an oracle $A$ such that $\mathsf{SZK}^A \not\subset \mathsf{BQP}^A$.
\end{theorem}

This strengthened the result of Bennett et al.:
\begin{theorem}[\cite{DBLP:journals/siamcomp/BennettBBV97}]
There exists an oracle $A$ such that $\mathsf{NP}^A \not\subset \mathsf{BQP}^A$.
\end{theorem}

Theorem~\ref{A03} was a corollary of a quantum lower bound for the following \textbf{Collision Problem}:
\begin{quote}
Given oracle access to a function  $f\colon \{1,\dots, N\} \to \{1,\dots, N\}$,  where $N$ is even and a promise that either $f$ is 1-to-1 or 2-to-1, decide which. (This is a generalized version of Simon's problem, but without any group structure.)
\end{quote}

It's easy to design a quantum algorithm that ``almost'' solves the collision problem with only $1$ query.  To do so, we prepare the state
\[
\frac{1}{\sqrt{N}} \sum_{x=1}^N \ket{x} \ket{f(x)}.
\]
Just like in Simon's algorithm, if we measure the $\ket{f(x)}$ register in the above state, then we're left with a state in the $\ket{x}$ register of the form

$$\frac{\ket{x} + \ket{y}}{\sqrt{2}},$$

\noindent for some $(x,y)$ such that $f(x)=f(y)$.  If we could just measure the above state twice, to get both $x$ and $y$, we'd be done!  The problem, of course, is that we can't measure twice.  We could repeat the algorithm from the beginning, but if we did, we'd almost certainly get a different $(x,y)$ pair.  In Simon's algorithm, we were able to do an interference experiment that revealed a little bit of collective information about $x$ and $y$---just enough so that, with sufficient repetitions, we could put together the partial information to get what we wanted.

The result of Aaronson \cite{aar:col} showed that the same is not possible for the Collision Problem:

\begin{theorem}[\cite{aar:col}]
\label{colthm}
Any quantum algorithm for the Collision Problem needs $\Omega(N^{1/5})$ queries.
\end{theorem}

This was improved by Yaoyun Shi \cite{shi} to $\Omega(N^{1/3})$, which is tight.  (Classically, $\Theta(\sqrt{N})$ queries are necessary and sufficient.)

Now, observe that the Collision Problem is in $\mathsf{SZK}$, via a protocol similar to that for Graph Non-Isomorphism:
\begin{quote}
Arthur chooses a uniformly random $x$, then sends $f(x)$ to Merlin and asks what $x$ was.  If $f$ is 1-to-1, then Merlin can always return the correct answer.
If, on the other hand, $f$ is 2-to-1, then Merlin can answer correctly with probability at most $1/2$.  Furthermore, in the honest case (the case that $f$ is 1-to-1), Arthur gains zero information from Merlin's reply.
\end{quote}

Another way to see that the Collision Problem is in $\mathsf{SZK}$ is to observe that it's easily reducible to the Statistical Difference problem (we leave this as an exercise for the reader).

Now, by Aharonov and Ta-Shma's result, the fact that there's an oracle separation between $\mathsf{SZK}$ and $\mathsf{BQP}$ implies that {\em there also exists an oracle $A$ relative to which there's a samplable distribution $D$ that's not QSamplable.}  Indeed, one can get that relative to such an oracle, $\sizee^A(\ket{\psi_D})>n^{O(1)}$.

\section{Quantum Witness States}

Our second example of an interesting class of states is {\em quantum witnesses}. These are states that, when given to you, you can use to verify some mathematical statement---e.g., that an input $x$ belongs to a language $L$---efficiently.

Recall that $\mathsf{MA}$ (Merlin-Arthur) is the probabilistic generalization of $\mathsf{NP}$: that is, it's the class of languages $L$ for which there exists
a probabilistic polynomial-time algorithm $A$ such that, for all inputs $x\in L$, there's a polynomial-size witness that causes $A$ to accept $x$ with probability at least $2/3$, while for all $x\not\in L$, no witness causes $A$ to accept $x$ with probability greater than $1/3$.  We imagine an omniscient but untrustworthy wizard, Merlin, trying to convince a skeptical probabilistic polynomial-time verifier, Arthur, of the statement $x\in L$. Unlike an oracle, Merlin can lie, so his answers
need to be checked.

We now define $\mathsf{QMA}$, the quantum generalization of $\mathsf{MA}$, which will lead us to the concept of quantum witnesses.

\begin{definition}
$\mathsf{QMA}$ (Quantum Merlin-Arthur) is the class of languages $L \subseteq \{0,1\}^n$ for which there exists a polynomial-time quantum verifier $V$ and polynomials $r,q$ such that for all $x$,
\begin{align*}
x \in L & \Rightarrow \quad \exists \ket{\phi_x} \quad  \Pr[V(x,\ket{\phi_x},\ket{00\dots 0}) \text{ accepts }]    \geq 2/3 \\
x \notin L & \Rightarrow \quad \forall \ket{\phi}  \quad \; \ \Pr[V(x,\ket{\phi},\ket{00\dots 0}) \text{ accepts }]    \leq  1/3
\end{align*}
Here, $\ket{\phi_x}$ is a state on $p(n)$ qubits, where $n=|x|$, and there are $q(n)$ ancilla qubits $\ket{00\dots 0}$.
\end{definition}

Note that, if some mixed state $\rho$ causes Arthur to accept, then by convexity, at least one pure state in the mixture makes him accept with at least as great a probability.  So we can assume without loss of generality that Merlin's witness is pure.

An obvious question is, can the completeness and soundness probabilities $2/3$ and $1/3$ in the definition of $\mathsf{QMA}$ be amplified (as they can for
$\mathsf{BPP}$ and other probabilistic complexity classes)?  The difficulty is that Arthur can't reuse his witness state $\ket{\phi_x}$, since it will quickly become corrupted.  However, what Arthur can do instead is to ask Merlin for $k$ witness states, check them independently, and accept if at least (say) $60\%$ cause him to accept.  If $x \in L$, then Merlin can simply give Arthur the same state $k$ times, which will cause Arthur to accept with probability at least $1-2^{-\Omega(k)}$.  If, on the other hand, $x \notin L$, any state of the form $\ket{\phi_1}\otimes \cdots \otimes \ket{\phi_k}$ will cause Arthur to accept with probability at most $2^{-\Omega(k)}$. One might worry that Merlin could entangle the $k$ witness states. If he does, however, that will simply make some of the witness states mixed. Merlin would do as well or better by putting his ``best'' state, the state that maximizes Arthur's acceptance probability, in all $k$ of the registers---and as we've seen, that state is pure without loss of generality.

In 2005, Marriott and Watrous \cite{DBLP:journals/cc/MarriottW05} proved a much more nontrivial result, with numerous applications to $\mathsf{QMA}$. We'll simply state their result without proof.

\begin{theorem}[Marriott and Watrous \cite{DBLP:journals/cc/MarriottW05}]
There exists a way to do {\em in-place amplification} for $\mathsf{QMA}$: that is, to replace the $(2/3,1/3)$ probability gap by $(1-2^{-k},2^{-k})$ for any desired $k$, by increasing Arthur's running time by an $O(k)$ factor, and crucially, without increasing the size of the witness state $\ket{\phi_x}$ at all.
\end{theorem}

\section{$\mathsf{QMA}$-Complete Promise Problems}

$\mathsf{QMA}$ can be seen as the quantum generalization of $\mathsf{NP}$ (or more precisely, as the quantum generalization of $\mathsf{MA}$). Just like thousands of practical problems, in optimization, constraint satisfaction, etc., were shown to be $\mathsf{NP}$-hard or $\mathsf{NP}$-complete, so today there's a growing list of {\em quantum} problems (albeit, so far merely dozens, not thousands!) that have been shown to be $\mathsf{QMA}$-hard or $\mathsf{QMA}$-complete.

As a technicality, because of the $(2/3,1/3)$ gap in defining $\mathsf{QMA}$, one needs to talk about $\mathsf{QMA}$-complete {\em promise} problems rather than $\mathsf{QMA}$-complete languages---but that's not a big deal in practice.  In fact, Goldrech \cite{goldreich:promise} has argued that complexity classes should have been defined in terms of promise problems in the first place, with languages just a special case.  This point of view is particularly natural in quantum complexity theory, where there are almost always natural complete promise problems, but rarely natural complete languages.

Here's the ``original'' $\mathsf{QMA}$-complete promise problem, discovered by Kitaev (see \cite{KSV02}), and later improved by Kempe, Kitaev, and Regev \cite{kkr} among others.

\begin{quote}
\textbf{Local Hamiltonians Problem} (quantum generalization of MAX-$2$-SAT).

We're given as input $m$ two-outcome POVMs $E_1,\dots, E_m$ acting on at most $2$ qubits each out of $n$ qubits total, as well as parameters $q$ and $\eps \geq 1/n^{O(1)}$.  let
$$A(\ket{\psi}) = \sum_{i=1}^n \Pr[\text{$E_i$ accepts $\ket{\psi}$}],$$
and then let
$$p = \max_{\ket{\psi}} A(\ket{\psi})$$
The problem is to decide whether $p \geq q$ or $p \leq q - \eps$, promised that one of these is the case.

Note that, if we restrict each $E_i$ to measure in the computational basis, and either accept or reject each of the four possibilities $00,01,10,11$ with certainty, then this problem reduces to the $\mathsf{NP}$-complete MAX-$2$-SAT problem.
\end{quote}

It's not hard to see that Local Hamiltonians is in $\mathsf{QMA}$.  The interesting part is the following:

\begin{theorem}[see \cite{KSV02}]
\label{localham}
Local Hamiltonians is $\mathsf{QMA}$-complete.  That is, {\em any} $\mathsf{QMA}$ problem can be efficiently reduced to Local Hamiltonians.
\end{theorem}

We won't prove Theorem \ref{localham} here. Intuitively, however, the result shows that {\em ground states of local Hamiltonians}---or equivalently, solutions to ``local'' quantum constraint satisfaction problems, with constraints involving only $2$ qubits each---can encode as much complexity as can arbitrary quantum witness states.

\section{Group Non-Membership}

We'll now give our first {\em surprising} example of a language in $\mathsf{QMA}$---that is, of a class of mathematical statements that can be efficiently verified using quantum witnesses, but possibly not using classical witnesses.

Our example will involve the concept of a {\em black-box group}, which was introduced (like so much else in this area) by Babai.  A black-box group is simply
a finite group $G$, typically of order $\exp(n^{O(1)})$, for which we know polynomial-size strings labeling the generators $\langle g_1, \dots, g_k \rangle$, and also have access to a black box that performs the group operations for us: that is, an oracle that returns the labels of group elements $gh$ and $g^{-1}$, given as input the labels of $g$ and $h$.  We also assume that this black box can tell us whether two labels are of the same $G$ element (or equivalently, that it can recognize a label of the identity elements).  In the quantum case, for reasons we'll see, we'll need one additional assumption, which is that every element of $G$ has a {\em unique} label.

The \textbf{Group Membership Problem} is now the following.  In addition to the black-box group $G$, we're also given a subgroup $H \leq G$, via a list of labels of its generators $\langle h_1, \dots, h_l\rangle$---as well as an element $x \in G$.  The problem is to decide whether $x \in H$.

We wish to decide this ``efficiently''---by which we mean, in time polynomial in $\log |G|$.  Recall that the group order $|G|$ is thought of as exponential, and also that every finite group $G$ has a generating set of size at most $\log_2 |G|$.

We can easily produce more concrete examples of this problem, by replacing the black-box groups $G$ and $H$ with explicit groups.  For example: given a list of permutations of an $n$-element set, decide whether they generate a target permutation.  (In the 1960s, Sims \cite{sims} showed how to solve this problem in polynomial time.)  Or: given a list of invertible matrices over a finite field, we can ask whether they generate a given target matrix. (For this problem, Babai, Beals, and Seress \cite{BBS09} gave a polynomial-time algorithm given an oracle for factoring---and therefore a $\mathsf{BQP}$ algorithm---assuming the field is of odd characteristic.)

Let's now consider the structural complexity of the Group Membership Problem.  If $x\in H$, then there's clearly a witness for that fact: namely, a sequence of operations to generate $x$ starting from the generators $\langle h_1, \dots, h_l\rangle$ of $H$.  To make this witness efficient, we might need to introduce a new element for each operation, and feed the new elements as inputs to later operations: for example, $h_3=h_1h_2$, then $h_4=h^2_3h_1^{-1}h_2$, etc.  By a result of \cite{DBLP:conf/focs/BabaiS84}, there always exists such a ``straight-line program,'' of size polynomial in $\log |H|$, to prove that $x\in H$.  Therefore Group Membership is in $\mathsf{NP}$.

OK, but what if $x$ is {\em not} in $H$?  Does Group {\em Non}-Membership (GNM) also have $\mathsf{NP}$ witnesses?  In the black-box setting, the answer is easily shown to be ``no'' (indeed, it's ``no'' even for cyclic groups).  For explicit groups, whether $GNM\in \mathsf{NP}$ is a longstanding open problem, related to the Classification of Finite Simple Groups and the so-called ``Short Presentation Conjecture.''

But even if $GNM\not\in \mathsf{NP}$, that's not the end of the story.  Indeed, Babai observed that $GNM$ has an $\mathsf{AM}$ (Arthur-Merlin) protocol, in which Arthur sends a random challenge to Merlin and then receives a response.

\begin{theorem}
\label{gnmam}$GNM \in \mathsf{AM}$.
\end{theorem}
\begin{proof}[Proof Sketch] The protocol involves yet another result of Babai \cite{DBLP:conf/stoc/Babai91}: namely, given any black-box group $G$, there is an algorithm, running in $\log^{O(1)} |G|$ time, that takes a nearly-random walk on $G$, and thus returns a nearly uniformly-random element $g\in G$.

Given this algorithm, Arthur simply needs to flip a coin, if it lands heads then sample a random element from $H$, and if it lands tails then sample a random element from $\langle H,x\rangle$ (that is, the least subgroup containing both $H$ and $x$).  Arthur then sends the element to Merlin and asks him how the coin landed.

If $x\in H$ then $H = \langle H,x\rangle$, so Merlin must guess incorrectly with probability $1/2$.  If, on the other hand, $x\not\in H$, then $\langle H,x\rangle$ has at least twice the order of $H$, so Merlin can guess correctly with probability bounded above $1/2$.
\end{proof}

The above is actually a statistical zero-knowledge protocol, so it yields the stronger result that $GNM\in \mathsf{SZK}$.

On the other hand, one can show (see for example Watrous \cite{DBLP:conf/focs/Watrous00}) that, at least in the black-box setting, $GNM\not\in \mathsf{MA}$: that is, verifying group non-membership requires interaction.

Or it at least, it requires interaction classically!  In 2000, Watrous discovered that Merlin can also use {\em quantum} witnesses to prove group non-membership, without any need for answering random challenges from Arthur.

\begin{theorem}[Watrous \cite{DBLP:conf/focs/Watrous00}]
\label{gnmqma}
$GNM \in \mathsf{QMA}$.
\end{theorem}
\begin{proof}[Proof Sketch]
The protocol is as follows.  In the honest case, Merlin sends Arthur the ``subgroup state''
\[
\ket{H} = \frac{1}{\sqrt{|H|}} \sum_{h \in H} \ket{h},
\]
i.e.\ a uniform superposition over all the elements of $H$.  Using $\ket{H}$ together with the group oracle, Arthur then prepares the state
\[
\frac{\ket0\ket{H} + \ket1\ket{Hx}}{\sqrt{2}},
\]
where $\ket{Hx}$ is a uniform superposition over the elements of the right-coset $Hx$.  (I.e., he puts a control qubit into the $\ket+$ state, then right-multiplies by $x$ conditioned on the control qubit being $\ket{1}$.)  Then Arthur measures the control qubit in the $\ket+,\ket-$ basis.

Just like we saw in Claim \ref{gibqp}, if $x\in H$ then $H=Hx$ and $\ket{H}=\ket{Hx}$, so this measurement will always return $\ket+$.  If, on the other hand, $x \notin H$, then $H$ and $Hx$ are disjoint sets, so $\braket{H}{Hx}=0$ and the measurement will return $\ket+$ and $\ket-$ with equal probability.

There is, however, one additional step: Arthur needs to verify that Merlin indeed gave him the subgroup state $\ket{H}$, rather than some other state.

There's no advantage for Merlin to send an equal superposition over a supergroup of $H$, since that would only decrease the probability of Arthur's being convinced that $x\not\in H$.  Nor is there any advantage in Merlin sending an equal superposition over a left-coset of such a group.  Thus, it suffices for Arthur to verify a weaker claim: that whatever state Merlin sent, it was (close to) an equal superposition over a left-coset of a supergroup of $H$.

To verify this claim, as in the proof of Theorem \ref{gnmam}, Arthur uses the result of Babai \cite{DBLP:conf/stoc/Babai91} that it's possible, in classical polynomial time, to generate a nearly uniformly-random element $h\in H$.  Given a purported witness $\ket{\psi}$, Arthur generates such an element $h$, and then prepares the state

\[
\frac{\ket0\ket\psi+\ket1\ket{\psi h}}{\sqrt{2}},
\]

\noindent where $\ket{\psi h}$ is the state obtained by taking all the elements in the $\ket{\psi}$ superposition and right-multiplying them by $h$.  Similar to the previous test, Arthur then measures his control qubit in the $\ket+,\ket-$ basis, and accepts the witness state $\ket{\psi}$ only if he observes the outcome $\ket+$---in which case, $\ket{\psi}$ is invariant under right-multiplication by random $H$ elements, and hence (one can show) close to a coset of a supergroup of $H$.  Note that in the honest case, this test doesn't even damage $\ket{\psi}$, so Arthur can reuse the same state to verify that $x\not\in H$.
\end{proof}

Let's make two remarks about the proof of Theorem \ref{gnmqma}.

First, we needed the assumption that every element of $G$ has a unique label, because if that assumption failed, then the states $\ket{H}$ and $\ket{Hx}$ could have been orthogonal even in the case $x\in H$.

Second, given that (by Babai's result) Arthur can efficiently sample a nearly-random $H$ element on his own, one might wonder why he needed Merlin at all!  Why couldn't he just prepare the witness state $\ket{H}$ himself?  The answer, once again, is the ``nuclear waste problem'': Arthur can efficiently prepare a state like

$$\sum_r \ket{r}\ket{h_r},$$

\noindent where $h_r$ is a nearly-random $H$ element.  But the above state is useless for Watrous's protocol, because of the garbage register $\ket{r}$.  And it's {\em not} known whether Arthur can remove $\ket{r}$ on his own; he might need Merlin to provide the garbage-free state $\ket{H}$.

(We should mention that, even in the black-box setting, it remains an open problem whether $GNM\in \mathsf{BQP}$, so it's conceivable that Arthur can always prepare $\ket{H}$ on his own.  But no one knows how.)

\lecture{Scott Aaronson}{Adam Brown and Fernando Pastawski}{$\mathsf{QCMA}$, $\mathsf{PostBQP}$, and $\mathsf{BQP/qpoly}$}

In Lecture 4, we discussed quantum superpositions of the form

$$\ket{H}  := \frac{1}{\sqrt{| H |}} \sum_{h \in H} \ket{h},$$

\noindent where $H$ is an exponentially-large finite group (say, a subgroup of a black-box group $G$).  We saw how such a superposition could be {\em extremely useful if someone handed it to you} (since you could use it to decide membership in $H$ in polynomial time), but also {\em intractable to prepare yourself} (because of the ``nuclear waste'' problem, and the distinction between sampling and QSampling).

In this lecture, we'll look more closely at the question of whether being given a complicated quantum state, one that you couldn't prepare by yourself, can really enhance your computational power.  In particular, could it enhance your power {\em beyond} what would be possible if you were given a classical string of similar size?

\section{$\mathsf{QCMA}$}

To study that question, we need to define one more complexity class: a cousin of $\mathsf{QMA}$ called $\mathsf{QCMA}$, which stands for the ironic ``Quantum Classical Merlin Arthur.'' $\mathsf{QCMA}$ is same thing as $\mathsf{QMA}$, except that now Merlin's witness must be classical (Arthur's verification procedure can still be quantum).  Thus, $\mathsf{QCMA}$ provides an alternative quantum generalization of $\mathsf{NP}$.

Clearly $\mathsf{QCMA} \subseteq \mathsf{QMA}$, since Arthur can simulate $\mathsf{QCMA}$ by immediately measuring a $\mathsf{QMA}$ witness to collapse it onto a classical state or probability distribution thereof.

Before going further, let's step back to explore how $\mathsf{QCMA}$ and $\mathsf{QMA}$ fit into the bestiary of complexity classes.

\begin{figure}[h!]\label{fig:ComplexityClassHierarchy}
  \begin{center}
\includegraphics[width=0.45\textwidth]{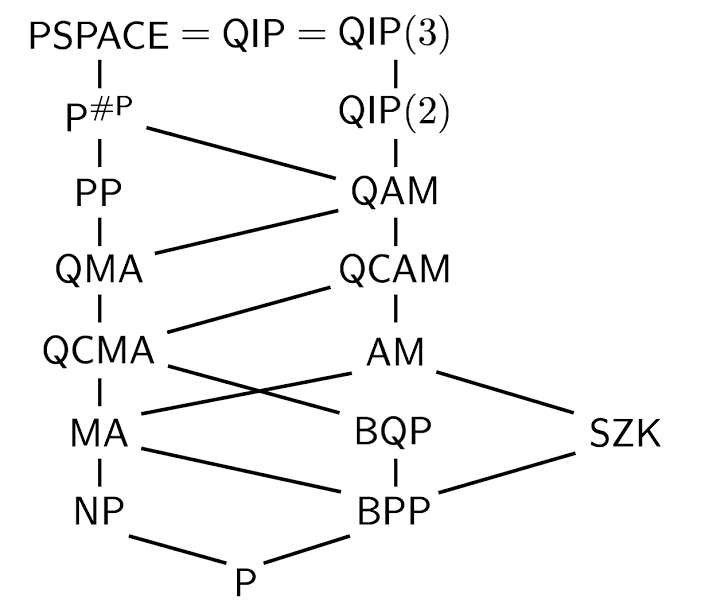}
\caption{Some relevant complexity classes, partially ordered by known inclusions.}
  \end{center}
\end{figure}

There's an interesting proof, which we won't present here, that $\mathsf{QMA} \subseteq \mathsf{PP}$.  There's also a whole plethora of classes such as $\mathsf{QAM}$, where Arthur sends random challenge string to Merlin and Merlin replies with a quantum state.

Then there's $\mathsf{QIP}$, the class of languages that admit a {\em quantum interactive proof protocol}, in which quantum data can go back and forth between Arthur and Merlin for a polynomial number of rounds.  In a 2009 breakthrough, Jain, Ji, Upadhyay, and Watrous \cite{JJUW11} showed that $\mathsf{QIP}=\mathsf{PSPACE}$: that is, once we're in the unbounded-round setting, quantum proofs give no advantage over classical proofs.  On the other hand, Kitaev and Watrous \cite{kitaevwatrous} also showed that $\mathsf{QIP(3)}=\mathsf{QIP}$: that is, 3-round quantum interactive proof protocols (an initial message from Merlin, a challenge from Arthur, and then a response from Merlin) have the full power of unbounded-round systems, a power that we now know to be simply $\mathsf{PSPACE}$.  Since in the classical case, constant-round interactive proof protocols give no more than $\mathsf{AM}$, the ability to do all of $\mathsf{QIP}=\mathsf{PSPACE}$ with only $3$ rounds shows a dramatic difference between the quantum and classical cases.

\section{$\mathsf{QMA}$ versus $\mathsf{QCMA}$}

But let's now return to the class $\mathsf{QCMA}$---and in particular, the question of whether it equals $\mathsf{QMA}$.

Since we saw in Lecture 4 that $GNM \in \mathsf{QMA}$, one way to investigate this question is to ask whether $GNM \in \mathsf{QCMA}$.  That is, in Watrous's $\mathsf{QMA}$ protocol for Group Non-Membership, did the advantage actually come from the quantum nature of the witness state $\ket{H}$?  Or could one also prove Group Non-Membership with a {\em classical} witness, verified using a quantum computer?

One might hope to resolve this problem in an oracular way, by showing that any $\mathsf{QCMA}$ protocol for Group Non-Membership would require Merlin to make a superpolynomial number of quantum queries to the black-box group $G$.  However, a result of Aaronson and Kuperberg \cite{ak} rules out such an oracle lower bound, by showing that we actually {\em can} solve $GNM$ using a $\mathsf{QCMA}$ protocol that makes only $\log^{O(1)} |G|$ queries to the group oracle.

\begin{theorem} [Aaronson-Kuperberg \cite{ak}]
\label{akthm}
There is a $\mathsf{QCMA}$ protocol for $GNM$ that makes
  $\mathsf{polylog}| G |$ queries to $G$, while using exponential post-processing time (in fact  $| G |^{\mathsf{polylog} | G |}$).
\end{theorem}

This is an ``obstruction to an obstruction,'' since it says that the $\mathsf{QCMA}$ query complexity of $GNM$ can only be lower-bounded by $\log^{O(1)} |G|$.
I.e., even if $\mathsf{GNM}$ does not belong to $\mathsf{QCMA}$, we
won't be able to prove that using only query complexity: computational complexity would necessarily also be involved.

To prove Theorem \ref{akthm}, we'll need the following corollary of the Classification of Finite Simple Groups:

\begin{theorem}
\label{finitesimple}
  Any group of order $\exp ( n)$ can be specified using $n^{O ( 1)}$ bits (in other words, there are most $k^{\log^{O(1)}k}$ non-isomorphic groups of order $k$).
\end{theorem}

We now sketch the proof of Theorem \ref{akthm}.

\begin{proof}[Proof Sketch for Theorem \ref{akthm}]
  The $\mathsf{QCMA}$ protocol works as follows.  Merlin tells Arthur an explicit model group $\Gamma$, which he claims is isomorphic to $G$.  Merlin then needs to describe an isomorphism $\Gamma \equiv G$ in detail, and Arthur needs to verify the purported isomorphism.  Once this is done, Arthur can then verify that $x\in H$ by simply checking that the corresponding statement is true in the model group $\Gamma$, which requires no further queries to the $G$ oracle.

  In more detail, Merlin sends Arthur a list of generators $e_1,
  \ldots, e_k \in \Gamma$, with the property that every $\gamma \in \Gamma$ can be written in the form $e_1^{\alpha_1} \ldots e_k^{\alpha_k}$ with $k = \mathsf{polylog} | \Gamma |$ and $\alpha_i \in \{0,1\}$.  Merlin also tells Arthur a mapping $e_i \rightarrow g_i$,
  which he claims can be extended to a full isomorphism between the groups, $\varphi : \Gamma
  \rightarrow G$.  To define $\varphi$: for each element $\gamma \in
  \Gamma$, we take its lexicographically first representation in the form $\gamma = e_1^{\alpha_1} \ldots
  e_k^{\alpha_k}$, and then define

  $$ \varphi ( \gamma) := \varphi (  e_1^{\alpha_1}) \ldots \varphi ( e_k^{\alpha_k}) = g_1^{\alpha_1} \ldots  g_k^{\alpha_k}. $$

  Merlin likewise gives Arthur a purported mapping from a model subgroup $\Delta \leq \Gamma$ to $H\leq G$, and from some element $z\in \Gamma \diagdown \Delta$ to $x\in G$.

  For simplicity, let's restrict ourselves to how Arthur checks the claimed isomorphism $\varphi : \Gamma
  \rightarrow G$ (checking the other parts is analogous).  There are two things that Arthur needs to check:
  \begin{enumerate}
    \item \textbf{That $\varphi$ is a homomorphism.}
    According to a result by Ben-Or, Coppersmith, Luby, and Rubinfeld \cite{bclr}, this can be done efficiently by checking the homomorphism $\varphi ( x \cdot y) =  \varphi ( x) \cdot \varphi ( y)$ on randomly chosen inputs $x$ and $y$.
    Even if it is not an exact homomorphism, and this goes undetected by the test (which may only happen if there is an epsilon fraction of incorrect entries, i.e.\ if it is close to a homomorphism), there is a further theorem which guarantees that the map $\varphi$ may actually be fixed into a homomorphism by taking $\Phi ( x) = \mathsf{maj}_y   \varphi ( x y) \varphi ( y^{- 1})$ (i.e.\ simply taking some form of majority vote).

    \item   \textbf{That the homomorphism is a 1-to-1 embedding}.  In other words, we must verify that the kernel of $\varphi$, or the subgroup of $\Gamma$ mapping to the identity, is trivial.  But this is just an instance of the Hidden Subgroup Problem!  (Indeed, the HSP with a normal subgroup, which is thought to be easier than the general problem.)  So
        we can use the algorithm of Ettinger, H\o yer, and Knill \cite{ehk} (Theorem \ref{ehkthm} from Lecture 3), which lets us solve the HSP with only polynomially many queries plus exponential post-processing.
  \end{enumerate}
  Note that these checks don't rule out that there could be additional elements of $G$ that aren't even in the image of $\varphi$: that is, they verify that $\varphi$ is an embedding, but not that it's an isomorphism.  Fortunately, though, this distinction is irrelevant for the specific goal of verifying that $x\not\in H$.
\end{proof}

So the bottom line is that, if we want to prove an oracle separation between $\mathsf{QMA}$ and $\mathsf{QCMA}$, we'll need a candidate problem that's less ``structured'' than Group Non-Membership.  In fact, finding an oracle $A$ such that $\mathsf{QCMA}^A \neq \mathsf{QMA}^A$ remains a notorious open problem to this day.  In their 2007 paper, Aaronson and Kuperberg \cite{ak} were able to make some progress toward this goal, by showing that there exists a {\em quantum} oracle $U$ such that $\mathsf{QCMA}^U \neq \mathsf{QMA}^U$.  Here, by a quantum oracle, they simply meant an infinite collection of unitary transformations, $\{U_n \}_n$, that can be applied in a black-box manner.  For convenience, we also assume the ability to apply controlled-$U_n$ and $U_n^{-1}$ (though in this case, the $U_n$'s that we construct will actually be their own inverses).

\begin{theorem}[Aaronson and Kuperberg \cite{ak}]
\label{akthm2}
  There exists a quantum oracle $U = \{ U_n \}_{n \geq 1}$ such that $\mathsf{QMA}^U
  \neq \mathsf{QCMA}^U$.
\end{theorem}

\begin{proof}[Proof Sketch]
  For each value of $n$, we choose an $n$-qubit pure state $\ket{\psi}=\ket{\psi_n}$ uniformly at random from the
  Haar measure.  We then design an $n$-qubit unitary transformation $U_n$ so that one of the following holds: either
\begin{enumerate}
\item[(i)] $U_n =I$ is the identity, or
\item[(ii)] $U_n =I- 2 \ket{\psi}\bra{\psi}$ implements a reflection about $\ket{\psi}$.
\end{enumerate}
Clearly, a witness for case (ii) is just $\ket{\psi}$ itself.
  We now use that ``the space of quantum states is big,'' in the sense that there exist
   $$ N = 2^{\Omega ( 2^n)}  \hspace{1em} \text{states }
    \ket{ \psi_1}, \ldots, \ket{ \psi_N }, $$
  of $n$ qubits each, such that $\braket{ \psi_j}{ \psi_i} \leq \varepsilon$ for all
  pairs $i \neq j \in \{ 1, \ldots, N \}$.  We can achieve that, for example, using states of the form
  \[
  \ket{ \psi_j } := \frac{1}{\sqrt{2^n}} \sum_x ( - 1)^{f_i ( x)} \ket{ x},
  \]
  where $f_1, \ldots, f_N$ are randomly-chosen Boolean functions.

  We now need an extended digression about the second most famous quantum algorithm after Shor's algorithm: namely, Grover's algorithm \cite{grover}!

  Grover's algorithm is usually presented as a quantum algorithm to search a database of $N$ records for a desired record in only $O(\sqrt{N})$ steps.  More abstractly, the algorithm
  is given oracle access to a Boolean function $f : \{ 1, \ldots, N \} \rightarrow \{ 0, 1 \}$.  Its goal is to decide whether there exists an $i\in \{ 1,ldots,N\}$ such that $f ( i) = 1$. Or, in a common variant, the algorithm needs to find $i$ such that $f(i)=1$, promised that there is a unique such $i$.  In 1996, Grover \cite{grover} gave a quantum algorithm to solve this problem using only $O ( \sqrt{N} )$ queries
  to $f$.  This is also known to be optimal for quantum algorithms, by a result of Bennett, Bernstein, Brassard, and Vazirani \cite{bbbv}.

  Just like the core of Shor's algorithm is period-finding, the core of Grover's algorithm is a more general procedure called {\em amplitude amplification}.  In amplitude amplification,
  we consider two nearly-orthogonal pure states $\ket{ v }$ and $\ket{ w }$, which satisfy
  $\braket{ v }{ w } | = \varepsilon$.
  For the moment, we'll assume that $\varepsilon$ is known, although this assumption can be lifted in ``fixed-point'' amplitude amplification.  We're given a copy of $\ket{v}$, and
  we want to convert it into $\ket{w}$, or something close to $\ket{w}$.  To do this, we can use two quantum oracles: $U_v = I- 2 \ket{v}\bra{v}$ and $U_w = I- 2 \ket{w}\bra{w}$.

  The central result of amplitude amplification is that we can achieve this conversion task using only $O( 1/ \varepsilon)$ queries to $U_v$ and $U_w$, as follows:
$$
    \ket{w} \leftarrow \overbrace{\ldots U_v U_w U_v U_w}^{1 / \epsilon \;
    \mathsf{times}} \ket{v} $$

  \begin{figure}[h!]\label{fig:AmplitudAmplification}
  \begin{center}
\includegraphics[width=0.6\textwidth]{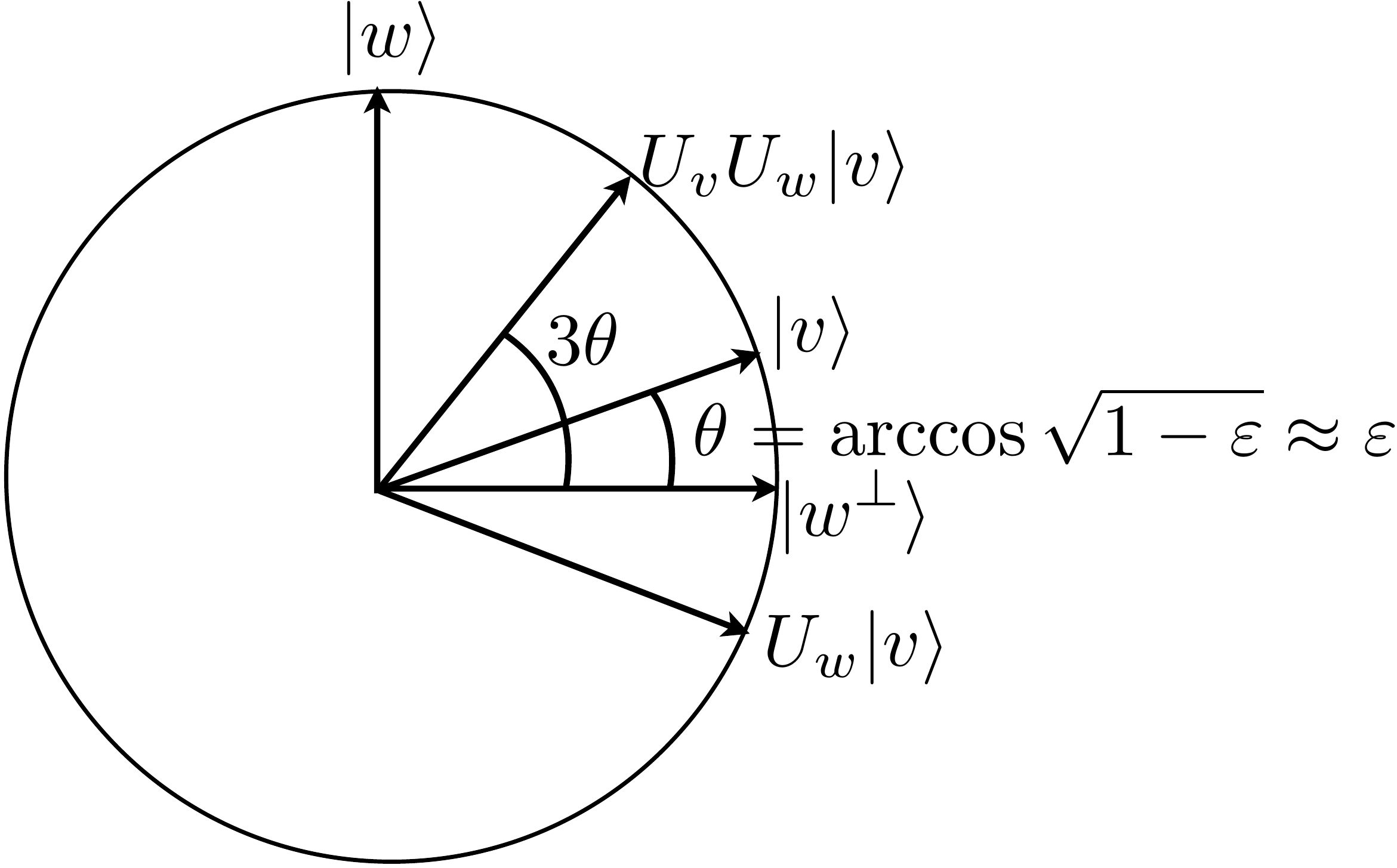}
\caption{The intermediate states in the relevant 2-dimensional subspace for amplitude amplification.
This provides a ``proof by picture'' for why amplitude amplification works.}
\end{center}
\end{figure}
Why does amplitude amplification work?  We can simplify the analysis by observing that the intermediate states are restricted to a purely 2-dimensional subspace:
  namely, the subspace spanned by the basis states $\ket{ v }$ and $\ket{ w }$ (see Figure \ref{fig:AmplitudAmplification}).  Note that $\ket{ w^{\perp} }$ and $\ket {v }$ have an angle of $\theta = \mathsf{arcos} (\sqrt{1-\varepsilon})\approx \varepsilon$.
  Each reflection step increases the angle to $\ket{w^\perp}$ by $\theta$ and hence, the number of steps to approach $\ket{w}$ is of the order $1/\varepsilon$.

Grover's algorithm works by taking the initial state to be the uniform positive superposition over all basis states $\ket{ v } = \frac{\ket{1} +  \ldots + \ket{N}}{\sqrt{N}}$ and $\ket{w} = \ket{i}$ be the basis state with index $i$ such that $f (i) = 1$, then we have $| \ketbra{ v}{w } | =  \frac{1}{\sqrt{N}} $.
Thus, amplitude amplification lets us find $\ket{i}$ in a time proportional to $\sqrt{N}$.\footnote{As a side remark, the state $\ket{ v}$ does not by itself let us implement $U_v$, and
  this fact is used explicitly in quantum money schemes (see Lectures 8 and 9).
  On the other hand, if we  have some unitary $U$ which prepares $\ket{ v }$, meaning
  for example that $U \ket{ 0 } = \ket{ v }$, and we also have $U^{-1}$, then we can also implement $U_v = U U_0 U^{- 1}$, where $U_0$ implements the reflection around the $\ket{ 0 }$ state.}

Now, to return to the proof of Theorem \ref{akthm2}: Merlin can give Arthur a $w$-bit witness string (there are $2^w$ possible such
strings).
Let's say that the Hilbert space is fragmented into $2^w$ regions, one for each string that Merlin could send.
Conditioned on being in a particular region, $\ket{\psi}$ will still be
uniformly random.
Now with high probability, the volume of that region is $\sim 1 / 2^w$ in terms of the probability mass.
One can show that the best possible case, from the algorithm's standpoint, is that we've localized $\ket{ \psi }$ to within a spherical cap (i.e., the set of all $\ket{v}$ such that $| \braket{ v }{ \psi } |
\geq h$, for some $h$).
Here we can calculate that the behavior of $h$ as a function of $w$ is given by $h \sim \sqrt{\frac{w + 1}{2^n}}$.
Since Grover search and amplitude amplification are optimal, this can be used to show that the number of queries is $\Omega\left(\frac{2^n}{w+1}\right)$.

\end{proof}

As mentioned earlier, it remains open whether there exists a classical oracle $A$ such that $\mathsf{QMA}^A \neq \mathsf{QCMA}^A$.
Very recently, Fefferman and Kimmel \cite{FK15} presented interesting progress on this problem.
Specifically, they separated $\mathsf{QMA}$ from $\mathsf{QCMA}$ relative to a classical oracle of an unusual kind.
A standard classical oracle receives an input of the form $\ket{x,y}$ and outputs $\ket{x, y\oplus f(x)}$, whereas theirs receives an input of the form $\ket{x}$ and outputs $\ket{\pi(x)}$, where $\pi$ is some permutation.  Furthermore, their oracle is probabilistic, meaning that the permutation $\pi$ is chosen from a probability distribution, and can be different at each invocation.

Let's now notice that there's a fundamental question about the complexity of quantum states, which is directly related to the $\mathsf{QMA}$ versus $\mathsf{QCMA}$ question:

\begin{question}
\label{qmaques}
For every language $L\in \mathsf{QMA}$, for every $\mathsf{QMA}$ verifier $V$ for $L$, and for every input $x\in L$, is there a polynomial-size quantum circuit to prepare a state $\ket{ \psi}$ such that $V(x,\ket{\psi})$ accepts with high probability?
\end{question}

{\em If} the answer to the above question is yes, then it follows immediately that $\mathsf{QMA}=\mathsf{QCMA}$, since Merlin could simply send Arthur a classical description of a circuit
to prepare an accepting $\mathsf{QMA}$ witness $\ket{ \psi}$.  However, the converse is not known: even if $\mathsf{QMA}=\mathsf{QCMA}$, there might still be $\mathsf{QMA}$ verifiers for
which accepting witnesses are exponentially hard to prepare.  (We {\em can} always prepare accepting witnesses efficiently if we make the stronger assumption that $\mathsf{BQP}=\mathsf{P^{\#P}}$.)

\subsection{Quantum Advice States}

We now turn from quantum witness states to another extremely interesting family of complicated states: namely, quantum {\em advice} states.  Unlike a witness state, an advice state can
be trusted: it doesn't need to be verified.  But there's another difference: we want the {\em same} advice state to be useful for deciding many different instances.  Indeed, we'd like
a family of advice states $\{ \ket{\psi_n} \}_{n \geq 1}$ that depend only on the input length $n$, and not on the input itself.  Think of a graduate adviser who's honest but very busy, and only gives advice to his or her students based on which year they are in.

This concept originates in classical complexity theory, and in particular, with the complexity class $\mathsf{P/poly}$.

\begin{definition}
  $\mathsf{P/poly}$ is the class of languages $L\subseteq \{0,1\}^*$ decidable by a polynomial-time Turing machine
  which also receives a polynomial-size advice string $a_n$ only dependent on the length $n$ of the input.
\end{definition}

We now define the quantum analogue of $\mathsf{P/poly}$.

\begin{definition}
  $\mathsf{BQP/qpoly}$ is the class of languages $L \subseteq \{ 0, 1 \}^*$ for which
  there exists a polynomial-time quantum algorithm $Q$, and a set of polynomial-size
  {\em quantum advice states} $\{ \ket{\psi_n} \}_{n\geq 1}$, such that
  $Q ( \ket{ x } \ket{ \psi_n} \ket{ 0 \cdots 0})$ accepts with probability $\geqslant 2 / 3$ if $x \in L$ and
  rejects with probability $\geqslant 2 / 3$ if $x \not{\in} L$ for every $x
  \in \{ 0, 1 \}^n$.
\end{definition}

Because any mixed state can be purified using twice as many qubits, we may assume the advice states to be pure or mixed without changing the power of this complexity class.

What kinds of quantum states can be useful as quantum advice?  We'll try to explore that question by studying the power of $\mathsf{BQP/qpoly}$.

Just like we can ask whether $\mathsf{QCMA}=\mathsf{QMA}$, so we can ask whether $\mathsf{BQP/poly}=\mathsf{BQP/qpoly}$, where $\mathsf{BQP/poly}$ represents
$\mathsf{BQP}$ with polynomial-size {\em classical} advice.  And just as Aaronson and Kuperberg \cite{ak} gave a quantum oracle separating $\mathsf{QCMA}$ from $\mathsf{QMA}$
(Theorem \ref{akthm2}), so an analogous construction gives a quantum oracle separating $\mathsf{BQP/poly}$ from $\mathsf{BQP/qpoly}$:

\begin{theorem}
There exists a quantum oracle $U$ such that $\mathsf{BQP}^U\mathsf{/qpoly} \neq \mathsf{BQP}^U\mathsf{/poly}$.
\end{theorem}

Just like with $\mathsf{QCMA}$ versus $\mathsf{QMA}$, it remains open whether there's a classical oracle separating $\mathsf{BQP/poly}$ from $\mathsf{BQP/qpoly}$.

On the other hand, a variation of the Group Non-Membership Problem gives us an example of something in $\mathsf{BQP}^U\mathsf{/qpoly}$ that's not {\em known} to
be in $\mathsf{BQP/poly}$.  Namely, what we can do is fix in advance the group $G_n$ and subgroup $H_n$ as functions of $n$, and then only take an element $x\in G_n$ as input.  The problem, as before, is to decide whether $x\in H_n$.

The adviser knows both $G_n$ and $H_n$, and can provide $\ket{H_n}$ as the advice state.
Since the adviser is trusted, we don't even need to check any longer that $\ket{H_n}$ is valid, yielding a simpler protocol than before.

\section{Upper-Bounding $\mathsf{BQP/qpoly}$}

For most complexity classes, such as $\mathsf{QMA}$, it's obvious that not all languages are contained in the class.
But in 2003, Harry Buhrman asked the following question: how can we even exclude the possibility that $\mathsf{BQP}/\mathsf{qpoly}$ equals $\mathsf{ALL}$, the class of all languages?

We can think of the problem of deciding a language $L$ on inputs of length $n$ as computing an $n$-input boolean function $L_n:\{0,1\}^n \rightarrow \{0,1\}$.
Since the description of the quantum state $\ket{\psi_n}$ provided by the adviser involves exponentially many amplitudes, the adviser could choose to provide a state that encodes the entire truth table of $L_n$---for example,
$$\ket{\psi_n} = \frac{1}{2^{n/2}} \sum_{x\in\{0,1\}^n} \ket{x}\ket{L_n(x)}.$$
However, it's not clear how one would use this advice: note that measuring in the standard basis gives only a $2^{-n}$ probability of getting the answer for the desired input $x$.

We can exclude the possibility that $\mathsf{BQP/qpoly} = \mathsf{ALL}$ by showing that
\begin{equation}
  \mathsf{BQP} / \mathsf{qpoly} \subseteq \text{Anything} / \mathsf{poly}.
\end{equation}
This follows from a simple counting argument.  Since one needs $2^n$ bits to specify a function $L_n : \{0,1\}^n\rightarrow \{0,1\}$, whereas the $\mathsf{/poly}$ adviser can provide only $n^{O(1)}$ bits, we find that for every uniform complexity class $\mathcal{C}$ (that is, every class defined by
a countable set of machines), almost all languages $L$ must lie outside of $\mathcal{C}\mathsf{/poly}$.

The reason why the na\"{i}ve counting argument doesn't work for quantum advice is that, as we've seen, there are {\em doubly}-exponentially many $n$-qubit states that are almost orthogonal to one another.

And in fact, there are complexity classes that really {\em do} become all-powerful when enhanced with quantum advice!  For example, let $\mathsf{PQP}$ be the quantum analogue of $\mathsf{PP}$: that is, the class of languages decidable by a polynomial-time quantum algorithm that only needs to return the correct answer with probability greater than $1/2$.  Then we have:

\begin{proposition}
  $\mathsf{PQP/qpoly} = \mathsf{ALL}.$
\end{proposition}
\begin{proof}
We can simply use the aforementioned advice state
$$ \frac{1}{\sqrt{2^n}} \sum_{x \in \{ 0, 1 \}^n}\ket{ x} \ket{ L_n( x) }. $$
The $\mathsf{PQP}$ algorithm that uses this state will measure the first register.  If $x$ is observed, then it returns $L_n( x)$, whereas if any input other than $x$ observed, it returns $0$ or $1$ both with probability $1/2$.
This exponentially small bias toward the right answer is sufficient for a $\mathsf{PQP}$ machine to succeed.
\end{proof}

Thus, if we want to upper-bound $\mathsf{BQP/qpoly}$, then we'll need to use further structure in the way quantum states can be accessed and measured.

In 2004, Aaronson \cite{aar:adv} proved that $\mathsf{BQP/qpoly} \subseteq \mathsf{PP/poly}$, which put the first upper bound on the power of quantum advice, and implied that $\mathsf{BQP/qpoly} \neq \mathsf{ALL}$.  To explain how this works, we need to introduce one more complexity class: $\mathsf{PostBQP}$.

\begin{definition} $\mathsf{PostBQP}$
is the same as $\mathsf{BQP}$, except augmented with the power of postselection.
That is, we assume the ability to measure a single qubit, and ``postselect'' on a specific measurement outcome (e.g., by projecting onto
the $\ket{1}$ state).  We also assume that the postselection measurement yields the desired outcome with probability at least $1/\exp(n^{O(1)})$ (if it doesn't, then our machine doesn't define a $\mathsf{PostBQP}$ language).  Conditioned on the postselection measurement succeeding, we then measure a second qubit.  If $x\in L$, then the measurement of the second qubit returns $\ket{1}$ with probability at least $2/3$, while if $x\in L$
then it returns $\ket{1}$ with probability at most $1/3$.
\end{definition}

As a side note, why do we need to assume that the postselected qubit returns
  $\ket{1}$ with probability at least $1/\exp(n^{O(1)})$?  Because otherwise, we might need more than $n^{O(1)}$ bits even to represent
  the probability, and many of the basic properties of quantum computation (such as the Solovay-Kitaev Theorem) would break down.
  If our gate set is reasonably ``tame'' (for example, if all the entries of the unitary matrices are algebraic numbers), then one can prove
  that every polynomial-size circuit that accepts with nonzero probability accepts with probability at least $1/\exp(n^{O(1)})$---so in that
  case, no separate assumption is needed.  On the other hand, there also exist exotic gate sets (for example, involving Liouville numbers) for which this is false: that is, for which polynomial-size circuits can generate nonzero probabilities that are doubly-exponentially small or even smaller. For more about this, see Aaronson's blog post \cite{Aar14}.

In terms of class inclusions, where does $\mathsf{PostBQP}$ fit in?  It's not hard to see that

$$ \mathsf{NP} \subseteq \mathsf{PostBQP} \subseteq \mathsf{PP}. $$

For the first inclusion, we construct a quantum circuit $C$ associated to our favorite $\mathsf{NP}$-complete problem (say, $C(x)=1$ if $x$ satisfies
a given $3SAT$ instance $\varphi$, and $C(x)=0$ otherwise), and then prepare the state
\begin{equation}
   \frac{\sqrt{1-\varepsilon^2}}{\sqrt{2^n}} \sum_{x \in \{0,1\}^n} \ket{x} \ket{ C ( x) } + \varepsilon
  \ket{ \mathsf{NULL} } \ket{1},
\end{equation}
where (say) $\varepsilon = 4^{-n}$.  We then measure the second register, postselect on getting $\ket{1}$, and then accept if and only if we find an $x$ in the first register such that $C(x)=1$.

For the second inclusion, we simply repeat the proof of $\mathsf{BQP} \subseteq \mathsf{PP}$ (Theorem \ref{bqppp} in Lecture 2), except that now we only sum over computational paths which give the desired outcome on the postselected qubit.

In fact, Aaronson \cite{aar:pp} showed that the converse statement, $\mathsf{PP}\subseteq \mathsf{PostBQP}$, holds as well:
\begin{theorem}[\cite{aar:pp}]
\label{postbqppp}
  $\mathsf{PostBQP}=\mathsf{PP}$.
\end{theorem}

We won't prove Theorem \ref{postbqppp} in this course, but will just make a few remarks about it.  First, Theorem \ref{postbqppp} gives a surprising equivalence between a quantum complexity class and a classical one.  Indeed, the theorem can even be used to give alternative, ``quantum'' proofs of classical properties of $\mathsf{PP}$, such as the fact that $\mathsf{PP}$ is closed under intersection.

Theorem \ref{postbqppp} can also be used to argue about the classical hardness of tasks such as BosonSampling \cite{aark}.  This is because of the dramatic difference that Theorem \ref{postbqppp} reveals between $\mathsf{PostBQP}$ and its classical counterpart, $\mathsf{PostBPP}$ (also known as $\mathsf{BPP_{path}}$).  In particular, it's known that

$$\mathsf{PostBPP} \subseteq \mathsf{BPP^{NP}} \subseteq \mathsf{PH}.$$

One striking statement---which is a consequence of the above fact together with Toda's Theorem \cite{toda} that $\mathsf{PH}\subseteq\mathsf{P^{PP}}$---is that if $\mathsf{PostBPP}=\mathsf{PostBQP}$, then the polynomial hierarchy $\mathsf{PH}$ collapses to the third level.  This is much stronger evidence that $\mathsf{PostBPP}\neq \mathsf{PostBQP}$ than we have that (say) $\mathsf{BPP}\neq \mathsf{BQP}$.

With $\mathsf{PostBQP}$ in hand, we're ready to prove the limitation on quantum advice.  In particular, we'll show how to
simulate a quantum advice state by using {\em classical} advice together with postselection.
\begin{theorem}[\cite{aar:adv}]
\label{bqpqpolyub}
   $\mathsf{BQP/qpoly} \subseteq \mathsf{PostBQP/poly} = \mathsf{PP/poly}$.
\end{theorem}
\begin{proof}
We'll consider an amplified advice state $\ket{ \psi }$, consisting of $n$ copies of the original advice state $\ket{ \psi_n }$:
$$ \ket{ \psi }  = \ket{ \psi_n }^{\otimes n}. $$
This lets us output the right answer with exponentially small error probability, $\frac{1}{\exp ( n)}$.
Let $m$ be the number of qubits in $\ket{\psi}$; then $m$ is still polynomial in $n$.

The idea will be to guess the quantum advice state $\ket{\psi}$ without having any quantum access to it.
We will start out with ``the most na\"{i}ve guess'' for $\ket{\psi}$, which we take to be the maximally mixed state $\rho_0={I}/{2^m}$.
The adviser will then help us update our guess, by providing a sequence of input values $x_1,x_2,\ldots$ for which our provisional advice state would fail to lead to the right answer (i.e., would cause us to output the wrong answer with probability at least $1/3$).

In more detail, let $Q(x,\ket{\psi})$ be the $\mathsf{BQP/qpoly}$ algorithm that we're trying to simulate.  Then the adviser will first send a classical value $x_1 \in \{ 0, 1\}^n$ such that
$$\Pr[Q(x_1, \rho_1) = L_n(x)]< 2/3,$$
if such a value exists.  We'll then run $Q(x_1,\rho_1)$, and postselect on obtaining the right answer.  After this happens, let the state left over in the advice register
be $\rho_2$.

Next the adviser again sends a classical value $x_2$ such that
$$\Pr[Q(x_2, \rho_2) = L_n(x)]< 2/3,$$
if such a value exists, and so on.

If at any point there {\em doesn't} exist a suitable $x_k$, then the adviser is done: by construction, we're now left with an advice state $\rho_k$ such that $\Pr[Q(x, \rho_k) = L_n(x)]\geq 2/3$ for every input $x\in \{0,1\}^n$.

Note also that we can simulate postselecting multiple times by using a single postselection step at the end.

Thus, the only remaining problem is to prove an upper bound on $k$, the number of iterations until the refinement procedure terminates.

To prove such an upper bound, the first step is to choose an orthonormal basis containing the ``true'' $m$-qubit quantum advice state $\ket{\psi}$, and use that basis to expand the initial maximally mixed state.  Here we use the fact that the maximally mixed state can be written as an equal mixture of the vectors in {\em any} orthonormal basis:
\begin{equation}
\label{decom}
  I = \frac{1}{2^m} \sum_{x \in \{ 0, 1 \}^n} \ketbra{x}{x} =
  \frac{1}{2^m} \sum_{x \in \{ 0, 1 \}^n} \ketbra{\varphi_i }{ \varphi_i },
\end{equation}
where, say, $\ket{\varphi_1} = \ket{\psi}$ is the true advice state.
By the Quantum Union Bound (Lemma \ref{lem:qunionbound}), we have
$$\Pr [ Q ( \ket{\psi}) \text{ succeeds on } x_1,
\ldots, x_k] > 0.9.$$
So
$$\left( \frac{2}{3} \right)^k \geq \Pr\left[ Q \left( \frac{I}{2^m}\right) \text{ succeeds on } x_1, \ldots, x_k \right] > \frac{0.9}{2^m},$$
where the first inequality follows from the way we defined the iterative procedure, and the second follows from equation \ref{decom} and linearity.  Solving for $k$, we find that the iterative process
must end with $k = O ( m)$.
\end{proof}

Aaronson and Drucker \cite{AD14} later gave the following improvement of Theorem \ref{bqpqpolyub}:
\begin{theorem}[\cite{AD14}]
\label{adruckerthm}
$\mathsf{BQP/qpoly} \subseteq \mathsf{QMA/poly}$ (i.e.,  a quantum computer with polynomial-size quantum advice can be
  simulated by $\mathsf{QMA}$ with polynomial-size {\em classical} advice).  Indeed, something stronger is true: polynomial-size
  trusted quantum advice can be simulated using polynomial-size trusted classical advice together with polynomial-size {\em untrusted} quantum advice.
\end{theorem}
In other words, we can assume that the quantum advice state is simply a $\mathsf{QMA}$ witness state.
This may seem surprising, since quantum advice states need to satisfy exponentially many constraints associated to different inputs.

Theorem \ref{adruckerthm} has an interesting implication: namely, if Question \ref{qmaques} has a positive answer---that is, if all $\mathsf{QMA}$ witness states can be prepared by polynomial-size quantum circuits---then that would imply not merely $\mathsf{QCMA} = \mathsf{QMA}$, but $\mathsf{BQP/poly} = \mathsf{BQP/qpoly}$ as well.  This explains, perhaps, why these two complexity class questions have seemed so closely related, with every result about one matched by a corresponding result about the other.

We'll end by mentioning one more result in this direction, by Aaronson \cite{aar:qmaqpoly} in 2006:

\begin{theorem}[\cite{aar:qmaqpoly}]
\label{qmaqpolythm}
$\mathsf{QMA/qpoly}\subseteq \mathsf{PSPACE/poly}$.
\end{theorem}

In other words, even if we {\em combine} quantum proofs and quantum advice in the same complexity class, we still don't get infinite computational power: everything we do can still be simulated using polynomial-size classical advice and polynomial space.  Surprisingly, this required a new proof, with the sticking point being to show that $\mathsf{QMA/qpoly}\subseteq \mathsf{BQPSPACE/qpoly}$---i.e., that there's a way to loop through all possible $\mathsf{QMA}$ witnesses, without destroying the $\mathsf{/qpoly}$ advice in the process.  After that, the same argument used to show that $\mathsf{BQP/qpoly}\subseteq \mathsf{PostBQP/poly}$ also easily implies that

$$\mathsf{BQPSPACE/qpoly} \subseteq \mathsf{PostBQPSPACE/poly} = \mathsf{PSPACE/poly}.$$

\lecture{Scott Aaronson}{Ryan O'Donnell and Toniann Pitassi}{Black Holes, Firewalls and Complexity}

We now take a break from complexity theory to discuss the black-hole information problem.  But we promise: by the end of the lecture, this will bring us right back to complexity!  (What doesn't?)

A black hole has two main regions of interest.
The first is the singularity in the middle.
If you ever jump into a black hole, be sure to visit the singularity;
fortunately you won't be able to miss it!
The other region of interest is the event horizon surrounding the singularity.
The event horizon of a black hole is the boundary
from inside of which no signals can escape. A black hole has so much mass that it distorts
the causal structure of spacetime, so signals can't get from the interior of the
black hole to the exterior.

It's a clear prediction of general relativity that a
sufficiently large mass in a sufficiently small volume
will collapse to a black hole.  Penrose and Hawking proved this in the 1960s.
Einstein didn't think that black holes existed, but that was because
he didn't understand his own theory well enough.
Today we know that they do exist.
Indeed, just before these lectures were given, LIGO announced
the most direct evidence yet for the reality of black holes: namely, a gravitational wave signal that
could only be caused by two spiralling black holes.

What are the puzzles about black holes that have made theoretical physicists so obsessed with them?

\section{Thermodynamics and Reversibility}

The first puzzle is that the existence of black holes {\em seems} inconsistent with
the Second Law of Thermodynamics.
The Second Law says that entropy never decreases,
and thus the whole universe is undergoing a mixing process
(even though the microscopic laws are reversible).
The problem is that black holes seem to provide a surefire way to decrease entropy. They're like
an entropy dumpster, since when you throw something into a black
hole it goes away forever.
Indeed, in classical general relativity, the {\it No-Hair Theorem}
states that a black hole in equilibrium is completely characterized by its
position, velocity,
mass, charge, and angular momentum.
The attractor equations are such that there
is no possibility of adding extra ``hair'': in other words, no possibility of changing
the shape of the event horizon by adding microscopic ripples and so forth to it.
This means that, when bits of information are
thrown into a black hole, the bits seem to disappear from the universe,
thus violating the Second Law.

This troubled the late Jacob Bekenstein, a graduate student in the early 1970s.
He thought that thermodynamics must prevail, and did not accept
that black holes were an exception.
(He was a ``conservative radical.'')
After considering various thought experiments, he proposed
that entropy is proportional to $A$, the surface area of
the event horizon, at a rate of about $10^{69}$ bits per meter squared.
In other words, he said that entropy scales like the {\em surface area} of the black
hold, rather than like the volume as one might have thought.

Bekenstein's later arguments implied that $10^{69}$ bits per meter$^2$ is actually the {\em maximum}
entropy that any physical system of a given surface area can store---since if a system had greater entropy,
then it would simply collapse to a black hole.  In other words: for storage density, nothing beats black holes.
They're the most compact hard drives allowed by the laws of physics, although they're terrible for retrieval!

So black holes went from having no hair to being hairiest things in the universe---i.e.,
having the maximum entropy of any system of the same surface area.

The story goes that Hawking didn't like this. So he set out to refute it.
One thing people were saying was that in
ordinary physics, anything that has an entropy also has
a nonzero temperature.  But if black holes had a nonzero temperature,
then they'd have to be giving off
thermal radiation.  But that seems like an obvious absurdity, since
black holes are defined as objects that nothing can escape from.

So then Hawking
did a calculation, which was called ``semi-classical''
because it treated the radiation around a black hole quantum-mechanically, but still
treated the curved spacetime as being classical.  The purpose of the
calculation was to predict what energy fluctuations an external observer would see in the vicinity of a black hole.
Say that an observer, Alice, is sitting outside the event horizon.
Hawking's calculation predicted that she would indeed see thermal radiation
slowly escaping from the black hole, and at exactly the
rate that Bekenstein's formula predicted.

(A popular, oversimplified view is that there are these
```virtual particles'' in the vacuum---positive- and negative-energy excitations
that are constantly forming and annihilating.  What can happen is that the negative-energy one falls in and the positive-energy one comes out, so the
black hole loses a tiny bit of mass/energy, which the external observer
sees coming out.)

This slowly-escaping radiation from a black hole is
called {\em Hawking radiation}.
Recall that the number of bits stored by a black hole scales like $A \sim r^2$, where $r$ is the radius of the event horizon.
By comparison, the evaporation time scales like $r^3 \sim A^{3/2}$.  For a black hole the mass of our sun, the evaporation time
would be approximately $10^{67}$ years.

In 1976, however, Hawking drew attention to a further puzzle.
The semi-classical calculation also predicts that the escaping radiation
is {\em thermal}---so in particular, uncorrelated with the details of the information that fell into the
black hole.
In other words, each photon would come out in a mixed state $\rho$,
independent of the infalling matter.

So let's come back to Alice. What does she see? Suppose she knows the complete
quantum state %
$\ket{\psi}$ (we'll assume for simplicity that it's pure)
of all the infalling matter.  Then, after collapse to a black hole and Hawking evaporation,
what's come out is thermal radiation in a mixed state ${\rho}$.
This is a problem.
We'd like to think of the laws of physics as just applying one huge unitary transformation to the quantum
state of the world. But there's no unitary $U$
that can be applied to a pure state $\ket{\psi}$ to get a mixed state ${\rho}$.
Hawking proposed that black holes were simply a case where unitarity broke down, and pure states
evolved into mixed states.  That is, he again thought that
black holes were exceptions to the laws that hold everywhere else.

People discussed this problem for forty years or so.
Most people (including even Hawking himself, eventually) evolved to the point of view that in the true quantum-gravitational description of nature, everything should be
completely unitary. So if semi-classical field theory says otherwise, then
the semi-classical calculations are wrong---they're just an approximation anyway
to the real, unitary calculations, and this is a case where
they break down.
(A familiar analogy is burning a book. The smoke and ashes appear to have
no correlation with what's written in the book.  But
in reality, we know that you could in principle recover the contents of
the book from the smoke, ashes, and other byproducts, together with complete knowledge of the relevant laws of physics.)
Reversibility has been a central concept in physics since Galileo and Newton.  Quantum mechanics says that the only exception to reversibility
is when you take a measurement, and the Many-Worlders say
not even {\em that} is an exception.

\section{The Xeroxing Problem}

Alas, even if we believe that the evolution is unitary and the information comes out (and that the physics of the future will explain how),
there's still a further problem.
This is what physicists in the 1980s called
the {\em Xeroxing  Problem}.
Let's consider another thought experiment.
This time, instead of classical information, Alice drops a qubit into the
black hole.
From the perspective of someone who's inside the black hole and sees the qubit come in, it never
leaves.  But from the perspective of Alice, who stays outside
(if we believe in reversibility), the qubit
eventually has to come out in some scrambled form in the Hawking radiation.
So this leads us to the conclusion that there are
two copies of the qubit, violating the No-Cloning Theorem.

To deal with this problem, Susskind and 't Hooft proposed a concept called {\it black hole complementarity}.
They argued that the qubit isn't actually cloned, because the same observer
would never actually see both copies of the qubit.  You might object: why couldn't Alice just stay outside, wait $10^{67}$ years
or whatever for the qubit to come out, and then once it comes out, immediately
jump in and look for the other copy?
Well, the calculations predict that by the time Alice jumps in, the ``inside'' manifestation of the qubit
would have long ago hit the singularity, where Alice could no longer access it.
On the basis of thought experiments such as those, Susskind and Hooft proposed that the inside and outside qubits were
{\em literally the same qubit} (rather than two copies of the qubit), but measured or viewed in two different ways.

\section{The Firewall Paradox}

In 2012, Almheiri, Marolt, Polchinski and Sully (AMPS) wrote a paper \cite{amps} where they proposed yet another thought experiment,
which leads to yet another apparent problem, even if we believe in black hole complementarity.  The new thought experiment is called
the {\em firewall paradox}.  In Scott's view, the firewall paradox helped to sharpen the discussion around black hole complementarity, because
at the end of the day, it's not about which mathematical formalism to use or how to describe a state space---it's simply about {\em what an observer would experience} in a certain situation.  And no one can claim to understand complementarity unless they have an answer to that question.

The firewall paradox is more technical than the previous paradoxes
and involves several ingredients.
The first ingredient is a fact from Quantum Field Theory (QFT), which we're going to ask you to take on faith.
The QFT fact is that if you look at the vacuum state of our universe,
it has an enormous amount of short-range entanglement.
So if you're jumping into a black hole, and you see a photon of Hawking radiation
just emerging from the event horizon,
QFT predicts that there must be a partner photon that is just inside the event horizon,
entangled with the other one.
We can visualize this as a whole bunch of Bell pairs straddling the horizon.
Taking the contrapositive, if you {\em didn't} see these Bell pairs
as you crossed the event horizon, then you wouldn't see a smooth spacetime.
You would instead encounter an ``end of spacetime,'' a Planck-energy wall of photons at which
you'd immediately disintegrate.  This is what's referred to as the firewall.

The second ingredient that we'll need is a modern, information-theoretic view of
black holes, as extremely efficient scramblers of information.
In particular, when we let matter collapse to form a black hole, we can imagine that we start with $n$ qubits
in a simple state (say the $\ket{0\cdots 0}$ state).
And when the qubits come back out of the black hole
in the form of Hawking radiation,
we can model that by a very complicated quantum circuit $C$---for some purposes, a {\em random} quantum circuit---having been applied to
the qubits.
So we can think of the output state as a pseudorandom pure state,
and for many purposes, we can model it as being Haar-random.
(Of course it can't {\em really} be Haar-random, because the black hole formation time is only polynomial,
whereas we know that a Haar-random pure state has $2^{\Omega(n)}$ circuit complexity
with overwhelming probability.)

Abstracting the situation, we want to say something
purely information-theoretic about Haar-random states.
Let's consider a Haar-random pure state $\ket{\psi}$ on $n$ qubits.
Look at the reduced (mixed) state, ${\rho}$, of the first $k$ of these qubits,
obtained by tracing out the remaining $n-k$ qubits.
(In other words, look at the ``marginal distribution'' on the
first $k$ qubits.)  This state will have the form

$$\rho = \sum_{i=1}^{2^{n-k}} p_i \ketbra{\psi_i}{\psi_i}.$$

What does this state look like?  There are two regimes of interest.

\begin{itemize}
\item[(1)] $k<n/2$. In this case, $\operatorname*{rank}(\rho) = 2^k$.
Since $2^{n-k}$ dominates $2^k$, there are enough terms in
the above sum to give a full rank matrix. Using
concentration of measure inequalities, one can show that
the resulting state $\rho$ will be
very close to the maximally mixed state, $I_k$.  (To build intuition, it might help
to consider the extreme case $k=1$ or $k=2$.)
\item[(2)] $k > n/2$.
In this case $\operatorname*{rank}(\rho) = 2^{n-k} < 2^k$.  So in this regime,
$\rho$ is
no longer the maximally mixed state.
\end{itemize}

In conclusion, something very interesting happens when exactly half of the
qubits have emerged from the black hole (in the form of
Hawking radiation).
When half come out, the state $\rho$ as seen by the
outside observer is no longer maximally mixed.
Alice, our outside observer, can in principle start to see correlations, between the Hawking photons themselves, and between them and the infalling matter.
This point in time when half of the qubits come out is
called the {\em Page Time}, after Don Page who studied
it in the 1980s.

There's a further thing that we can say.
Looking at this state $\rho$ in a bit more detail,
when $k>n/2$ qubits have come out, with overwhelming probability, any one of these $k$ qubits
that we pick is entangled with the remaining $k-1$ qubits that
have come out.  (The proof of this is left as an exercise.)

Now back to the thought experiment.
We imagine Alice, who's outside the black hole and is
an all-powerful observer.  Alice sets up a system of $n$ qubits,
in a known initial state, and lets them collapse to form a black hole---evolving,
in isolation from everything else in the universe (we'll assume for simplicity), according
to unitary laws of physics that are completely known to Alice.
Then, as the black hole slowly evaporates, there's a sphere of perfect photodetectors surrounding
the black hole, which catches all the photons of Hawking radiation as they escape,
and routes them into a quantum computer for
processing.

Alice waits roughly $10^{67}$ years, until the black hole has passed its Page Time (but hasn't yet evaporated completely).
Suppose $2n/3$ of the qubits
have come out in the form of Hawking radiation, while $n/3$ remain in the black hole.
Then we consider three subsystems:
\begin{itemize}
\item $R$, consisting of the $k=2n/3$ qubits that have come out,
\item $B$, the very next qubit coming out, and
\item $H$, the remaining qubits that are still inside the black hole.
\end{itemize}
Now, what we concluded from our previous discussion (about states
that are effectively like Haar-random states) is that we
expect $B$ to be entangled with $R$.
In more detail, the joint state of these $k+1$ qubits {\em cannot} be the maximally
mixed state, and we further expect that $B$ shares one bit of entanglement with $R$.
Thus, by applying
a suitable unitary transformation to $R$ alone, Alice should be able to put some designated qubit of $R$ (say,
the ``rightmost'' one, for some arbitrary ordering of the qubits) into a Bell pair with $B$.

So Alice does that, and then measures $B$ and the last qubit of $R$ to confirm that they really are in a Bell pair.  And then Alice jumps into the black hole.

We already said that there exists another qubit
inside the black hole (i.e., in $H$) that $B$ is maximally entangled with, since there's nothing special about the event horizon
in terms of spacetime deformity.
And this entanglement between $B$ and $H$ must also be observable as we cross the event horizon.
But this violates the {\em Principle of Monogamy of Entanglement}, from back in Lecture 2!  The same qubit $B$ can't be maximally entangled with two other qubits (or for that matter, even entangled with one qubit and correlated with another).  There's simply no way to do everything we've demanded,
if all of Alice's experiences are to be described using quantum mechanics.

So if we want to preserve quantum mechanics, or even the appearance of quantum mechanics, then either Alice must be unable to observe entanglement between $B$ and $R$, or
else she must be unable to observe entanglement between $B$ and $H$.  This is the firewall paradox.

There are a few possibilities for resolving the paradox:
\begin{itemize}
\item[(1)] We could throw up our hands and declare that ``black holes are like Las Vegas''---that is, that it's outside the scope of science
to understand what happens inside a black hole.  (I.e., that the goal is to describe that part of the universe that we can receive signals from---in this case, the regions $R$ and $B$ outside the event horizon, where there's no problem.)
\item[(2)] We could say that there really is a firewall, or end of spacetime, at the event horizon.
This would be a radical change to black hole physics, but note that it's compatible with unitarity.  Something weird happens to the
infalling observer, but everything is fine for an outside observer.
\item[(3)] We could give up on one of the basic assumptions of physics that
led to this point. For example, we could give up on unitarity.
(In that case, something weird {\em would} happen for the outside observer, who would
see a truly mixed state.)
\item[(4)] We could say that there exists a firewall {\em if} Alice does this crazy experiment,
but not under ``normal circumstances.''  I.e., that whether there is or isn't a firewall---what Alice perceives as the
nature of spacetime at the event horizon---depends on whether she switched on her quantum
computer and had it process the Hawking radiation in this bizarre way.
\end{itemize}

Scott, in common with many of the actual experts in this area (of which he's not one), tends to favor (4) as the ``least bad option.''
On this fourth view, the whole paradox arose from not taking complementarity seriously enough.
We were insisting on thinking of the qubits coming out as being
different from the ones going in, but complementarity has been telling us
all along that they're the same.  In other words, on the fourth view, $H$ {\em doesn't exist}: what we've been calling
$H$ is really just $R$ and $B$ measured in a different basis.  In normal circumstances one doesn't notice this,
but sufficiently extreme processing of $R$ and $B$ can make it apparent.

However, if we want option (4), then we face a further question: namely, what do we mean by ``normal circumstances''?
Or more pointedly: can we draw any sort of principled distinction between the types of unitary evolution of $R$ and $B$ that might create a firewall,
and the ``ordinary'' types that have no danger of creating one?

\section{The HH Decoding Task}

This brings us to the striking work of Harlow and Hayden (HH) \cite{harlowhayden}, who abstracted
the firewall paradox into a purely computational problem, as follows:

\begin{itemize}
\item \textbf{The Harlow-Hayden Decoding Task.}  We're given as input a description of a quantum circuit $C$, which
maps $n$ qubits (say, in the initial state $\ket{0}^{\tensor n}$) to a tripartite state $\ket{\psi}_{RBH}$, where $B$ is
a single qubit.  We're promised
that there exists a unitary transformation $U$, which acts only on the $R$ part of $\ket{\psi}_{RBH}$, and which has the effect
of putting $B$ and the rightmost qubit of $R$
into the joint state $\frac{1}{\sqrt{2}}(\ket{00}+\ket{11})$.  The challenge is to apply such a $U$ to the $R$ part of $\ket{\psi}_{RBH}$.
\footnote{Throughout these lectures, we'll assume for simplicity that the HH Decoding Task requires decoding the Bell pair $\frac{1}{\sqrt{2}}(\ket{00}+\ket{11})$ {\em perfectly}.  Relaxing this, to allow for approximate decoding, leads to technical complications but doesn't
change anything conceptually.}
\end{itemize}

The HH Decoding Task is {\em information-theoretically} possible---indeed, the very statement of the problem assures us that the desired unitary transformation $U$ exists---but it might be {\em computationally} intractable.  For suppose again that $R$ has $k$ qubits for $k>n/2$.  Then
we argued before that the reduced state $\rho_{RB}$ can't possibly have maximum rank.  If, moreover, $\ket{\psi}_{RBH}$ is the output of a random quantum circuit (as in the ``scrambling'' model of black holes discussed earlier), then a further argument shows that $\rho_{RB}$ will generically have entanglement between $R$ and $B$.  However, these arguments were abstract and non-effective, based on iterating through a basis for the whole Hilbert space.  If we tried to extract from the arguments an actual quantum circuit, acting on $R$, to distill the entanglement between $R$ and $B$, that circuit would have exponential size.\footnote{We do, of course, know the quantum circuit $C$ that maps $\ket{0}^{\tensor n}$) to $\ket{\psi}_{RBH}$, and {\em given access to all three registers}, it would be easy to apply $C^{-1}$.  The difficulty, if one likes, is that we're trying to uncover structure in $\ket{\psi}_{RBH}$ by acting only on $R$.}

But the above, as Harlow and Hayden noted, merely means that the decoding task {\em might} be computationally hard, not that it {\em is}.  We've said nothing against the possibility that there could be some clever, general way, which we haven't thought of yet, to convert the input circuit $C$ into a small circuit for the desired unitary $U$.  In complexity theory, to make the case for an algorithm's existence being unlikely, we normally at the least want a {\em reduction argument}, showing that the algorithm's existence would have many surprising consequences for seemingly-unrelated problems.  This is precisely what Harlow and Hayden accomplished with the following theorem.

\begin{theorem}[Harlow-Hayden \cite{harlowhayden}]
\label{hhthm}
If the HH Decoding Task can be done in polynomial time for arbitrary circuits $C$,
then $\mathsf{SZK} \subseteq \mathsf{BQP}$.
\end{theorem}

Indeed, if there exist $\mathsf{SZK}$ problems that are exponentially hard for quantum computers, then the HH Decoding Task also requires exponential time.

One can interpret this as saying that, at least in Harlow and Hayden's idealization, for a black hole the mass of our sun, {\em for Alice to do the preprocessing of the Hawking radiation necessary to create a firewall would be impossible in a mere $\sim 10^{67}$ years---it would instead take more like $\sim 2^{10^{67}}$ years!}.  And thus, we might argue, before Alice had even made a dent in the task, the black hole would have long ago evaporated anyway, so there would be nothing for Alice to jump into and hence no Firewall Paradox!

In the next section, we'll follow Harlow and Hayden's construction to prove a slightly weaker result.  We'll then see how to prove results that strengthen Theorem \ref{hhthm} in several respects.

\section{The Harlow-Hayden Argument}

Our goal now is to prove that, if the HH Decoding Task can be done in polynomial time for arbitrary circuits $C$, then a problem called ``Set Equality'' can also be solved in quantum polynomial time.

Set Equality is a special case of the Collision Problem, and is also a special case of Statistical Difference,
the canonical $\mathsf{SZK}$-complete problem (the Collision Problem and Statistical Difference were both defined in Lecture 4).

\begin{definition}
In Set Equality, we're given black-box access to two injective functions,
$f,g: \{1,...,N\} \rightarrow \{1,..,M\}$, where $M \geq 2N$.
We're promised that either
\begin{enumerate}
\item[(i)] $\operatorname*{Range}(f) = \operatorname*{Range}(g)$, or
\item[(ii)] $\operatorname*{Range}(f) \cap \operatorname*{Range}(g) = \emptyset$.
\end{enumerate}
The problem is to decide which.
\end{definition}

Now, in the same paper \cite{aar:col} where Aaronson proved the collision lower bound (Theorem \ref{colthm}), he also proved the following
lower bound on the quantum query complexity of Set Equality:

\begin{theorem}[\cite{aar:col}]
\label{seteqthm}
Any quantum algorithm for Set Equality must make
$\Omega(N^{1/7})$ queries to $f$ and $g$.
\end{theorem}

In 2013, Zhandry \cite{zhandry:col} strengthened Theorem \ref{seteqthm}, to show
that any quantum algorithm for Set Equality
must make $\Omega(N^{1/3})$ queries.  Zhandry's lower bound is tight.

If the functions $f$ and $g$ are described explicitly---for example, by circuits to compute them---then
of course we can no longer hope for a black-box lower bound.  Even in that case, however, we can say
that an efficient (in this case, $\log^{O(1)} N$-time) quantum algorithm for Set Equality would let us, for example,
solve Graph Isomorphism in polynomial time---since given two graphs $G$ and $H$, the sets $\{\sigma(G)\}_{\sigma\in S_n}$ and
$\{\sigma(H)\}_{\sigma\in S_n}$ are equal if $G\cong H$ and disjoint otherwise.  Indeed, the Statistical Difference
problem from Lecture 4 is just a slight generalization of Set Equality for explicit $f,g$, and it follows from Theorem \ref{svthm}
of Sahai and Vadhan \cite{DBLP:journals/jacm/SahaiV03} that an efficient quantum algorithm for Statistical Distance
would imply $\mathsf{SZK}\subseteq \mathsf{BQP}$.

Without further ado, let's show the reduction from Set Equality to the HH Decoding Task.  The reduction involves preparing
the following state $\ket{\psi}_{RBH}$:

$$\ket{\psi}_{RBH} = {1 \over \sqrt{2^{n+1}}}
\sum_{x\in\{0,1\}^n} \left( \ket{x,0}_R \ket{0}_B \ket{f(x)}_H + \ket{x,1}_R \ket{1}_B \ket{g(x)}_H \right).
$$

It's clear that $\ket{\psi}_{RBH}$ can be prepared by a polynomial-size circuit, given the ability to compute $f$ and $g$.  But why
does this state encode Set Equality?

Well, first suppose that $\operatorname*{Range}(f) \cap \operatorname*{Range}(g) = \emptyset$.
In this case, because the two ranges are disjoint, the $H$ register decoheres
any entanglement between $R$ and $B$, exactly as if $H$ had measured $B$ (or as happens in the GHZ state).
So the reduced state $\rho_{RB}$ is not entangled.  Thus, the HH Decoding Task is impossible because the promise is violated.

Second, suppose that $\operatorname*{Range}(f)=\operatorname*{Range}(g)$.
In that case, Alice (acting on $R$) simply needs to apply a permutation that maps each basis state $\ket{x,0}$ to itself,
and each basis state $\ket{x,1}$ to $\ket{f^{-1}(g(x)),1}$.  This yields the state

$${1 \over \sqrt{2^{n+1}}}
\sum_{x\in\{0,1\}^n} \left( \ket{x,0}_R \ket{0}_B + \ket{x,1}_R \ket{1}_B \right) \ket{f(x)}_H,
$$

\noindent in which $B$ and the rightmost qubit of $R$ are jointly in the entangled state $\frac{1}{\sqrt{2}}(\ket{00} + \ket{11})$, as desired.

Thus, if HH Decoding were easy whenever the promise were satisfied, then given functions $f,g$ for which we wanted to solve Set Equality,
we could first prepare the corresponding state $\ket{\psi}_{RBH}$, then try to perform HH Decoding, and then apply a measurement that projects
onto $\frac{1}{\sqrt{2}}(\ket{00} + \ket{11})$ to check whether we'd succeeded.  If $\operatorname*{Range}(f) = \operatorname*{Range}(g)$, we would find that we'd succeeded with probability $1$, while if $\operatorname*{Range}(f) \cap \operatorname*{Range}(g) = \emptyset$, we would find that we'd succeeded with probability at most $1/2$, which is the maximum possible squared fidelity of a separable state with the Bell pair (we leave the proof of that as an exercise).  Thus we can decide, with bounded error probability, whether $\operatorname*{Range}(f)$ and $\operatorname*{Range}(g)$ are equal or disjoint.

So, taking the contrapositive, if Set Equality is hard for a quantum computer, then so is the HH Decoding Task.

Before going further, let's address a few common questions about the HH argument.

\begin{enumerate}
\item[(1)] What role did black holes play in this complexity argument?  Firstly, the black hole is what prevents Alice from accessing $H$, thereby forcing her to apply complicated processing to $R$.  Secondly, the black hole is what ``scrambles'' the infalling qubits (admittedly, many more mundane physical systems would also scramble the qubits, though probably not as thoroughly as a black hole).  Thirdly, of course, the black hole is what led to the ``paradox'' in the first place!  With most physical systems, there would be no particular reason for $B$ and $H$ to be entangled, and therefore no difficulty if we observe $R$ and $B$ to be entangled.
\item[(2)] Are Harlow and Hayden saying that a contradiction in the laws of physics is okay, so long as it takes exponential time to reveal?   No, they aren't.   {\em Approximate} theories can and do have contradictions.  The HH view would say that general relativity and quantum field theory work well in ``ordinary'' circumstances, {\em but} if Alice were somehow able to solve, in polynomial time, a problem that we currently conjecture to be exponentially hard, then a full quantum theory of gravity would be needed to describe her experiences.   Note that physicists agree that GR and QFT fail in some regime, but that regime is usually taken to be that of extremely high energy and curvature (e.g., the Big Bang, or a black hole singularity).  The novelty of the HH argument is that it suggests a breakdown of field theories, {\em not}
    at the Planck energy (there's nothing high-energy in the HH Decoding Task), but instead in a regime of exponential computational complexity.
\item[(3)] What if we had a microscopic black hole?  Such a black hole would evaporate extremely quickly, thereby making asymptotic complexity apparently irrelevant.  The basic answer is that, if microscopic black holes exist, then they're fundamentally quantum-gravitational objects, so the physicists {\em don't much care} if they lead to a breakdown of smooth spacetime at their event horizons.  The firewall paradox only arises in the first place for ``astrophysical-mass'' black holes, because those are the ones for which the physicists expected GR and QFT to describe Alice's experiences (as long as she restricts herself to low-energy experiments).
\item[(4)] Doesn't the HH argument apply only to artificial, engineered quantum states?  Why couldn't the states occurring in real black holes have some special symmetry (for example)
        that would make the HH Decoding Task much easier than in the worst case?  Of course, this is possible in principle, but what we know about ``black holes as scramblers'' suggests the opposite: that the decoding task for generic black hole states should, if anything, be {\em harder} than for the special $\ket{\psi}_{RBH}$'s for which we're able to give
        formal hardness reductions.  For more about this, see Section \ref{ITSWORSE}.
\item[(5)] In a subsequent paper, Oppenheim and Unruh \cite{OU14} argued that, given exponential pre-computation time, it would be possible to engineer a special black hole to which the HH hardness argument didn't apply.  In their construction, the black hole qubits are maximally entangled with qubits that remain outside the black hole, in Alice's control---and before even letting the black hole qubits collapse, Alice spends exponential time processing their entangled partner qubits specifically to make her later decoding task easier.  Scott's inclination is to bite the bullet and say yes, in this extreme situation, the HH argument doesn't apply!  But if the argument applies to real, astrophysical black holes, or (better yet) to {\em any} black holes that we can produce via polynomial-time computations, then that's already an extremely interesting conclusion.
\end{enumerate}

\subsection{Improvements to Harlow-Hayden}

We'll now discuss some improvements to HH argument due to Aaronson (paper still in preparation).  The first improvement concerns the complexity assumption: despite the lower bounds for Collision and Set Equality, $\mathsf{SZK}\not\subset \mathsf{BQP}$ still strikes some people as a strong assumption.

In cryptography, the gold standard is to base whatever conclusion we want on {\em the existence of one-way functions}, which is considered the minimal assumption for complexity-based cryptography.  A one-way function (OWF) is a family of functions $f_n:\{0,1\}^n\longrightarrow \{0,1\}^{p(n)}$, for some polynomial $p$, such that
\begin{enumerate}
\item[(1)] $f_n$ is computable in $n^{O(1)}$ time, but
\item[(2)] for all polynomial-time algorithms $A$,
$$ \Pr_{x\in \{0,1\}^n}[f(A(f(x)))=f(x)]<\frac{1}{\operatorname*{poly}(n)}. $$
\end{enumerate}
(For simplicity, we often suppress the dependence on $n$.)

Or in words: an OWF is a function that's easy to compute but hard to invert---for which no polynomial-time algorithm can find preimages with any non-negligible success probability.

It's known (see Ostrovsky \cite{Ost91}) that if there are hard $\mathsf{SZK}$ problems, then there are also OWFs, but no one knows the converse.  So the existence of OWFs is a weaker assumption than $\mathsf{SZK}$ being hard.

Of course, for present purposes, we'll need OWFs that are hard to invert even by quantum computers.  But this is considered almost as safe an assumption as $\mathsf{NP}\not\subset \mathsf{BQP}$.  Certain {\em specific} OWFs useful in public-key cryptography, such as those based on factoring and discrete logarithms, can be inverted using Shor's algorithm.  But a general OWF $f$ could be based (for example) on simply iterating a cellular automaton or some other ``unstructured'' system, in which case, we wouldn't currently know how to invert $f$ any faster than by using Grover's algorithm.

We'll actually need {\em injective} OWFs, which are OWFs $f$ with the additional property that $f(x)\neq f(y)$ for all $x\neq y$.  But while no one knows how to convert an OWF into an injective one, the existence of injective OWFs is considered a safe assumption as well.  Indeed, if $f:\{0,1\}^n\longrightarrow \{0,1\}^{p(n)}$ ``behaves like a random function'' and $p(n) \gg 2n$, then probabilistic considerations immediately suggest that $f$ will be injective with overwhelming probability.

We can now state Aaronson's two improvements to Harlow and Hayden's Theorem \ref{hhthm}:

\begin{theorem}[Aaronson, not yet published]
\label{aarfirewall}
Suppose there exist injective OWFs that are hard to invert by quantum computers.  Then the HH Decoding Task is hard.
\end{theorem}
\begin{proof}
Let $f:\{0,1\}^n\longrightarrow \{0,1\}^{p(n)}$ be an injective OWF.  Then consider the state
$$
\ket{\psi}_{RBH} = {1 \over \sqrt{2^{n+1}}} \sum_{x\in\{0,1\}^n}  \left( \ket{x \, 0^{p(n)-n},0}_R \ket{0}_B + \ket{f(x),1}_R \ket{1} ) \ket{x}_H \right)
$$
where the $0^{p(n) -n}$ are ancilla bits.  It's not hard to see that the above state {\em does} have entanglement between $B$ and $R$.  Indeed, we could distill
the entanglement by applying a reversible transformation $U$ to $R$ that mapped each basis state of the form $\ket{f(x),1}$ to $\ket{x,0^{p(n)-n},1}$, while acting as the identity on basis states of the form $\ket{x,0^{p(n)-n},0}$.  (Such a $U$ exists because $f$ is injective.)  The only problem is that implementing such a $U$
would require inverting $f$---and intuitively, this seems inherent to {\em any} entanglement distillation procedure.

For a more formal proof: notice that, if $U$ succeeds at putting $B$ and the last qubit of $R$ into a Bell pair $\frac{1}{\sqrt{2}}{\ket{00}+\ket{11}}$, then there must exist states $\{ \ket{\phi_x} \}_x$ such that
$$ U \ket{x \, 0^{p(n) -n }, 0} = \ket{\phi_x} \ket{0} \ \ {\rm and } \ \ U\ket{f(x),1} = \ket{\phi_x} \ket{1}. $$
Now, let $V$ be the $p(n)$-qubit unitary induced by $U$ if we fix the last qubit of $R$ to 0, and let $W$ be the unitary induced by $U$ if we fix the last qubit of $R$ to $1$.  Then
\begin{eqnarray}
V \ket{x \, 0^{p(n) -n}} &=& \ket{\phi_x}, \\
W \ket{f(x)} &=& \ket{\phi_x}.
\end{eqnarray}
But this means that, to invert $f$, we simply need to apply
\begin{equation}
V^\dagger W \ket{f(x)} = V^\dagger \ket{\phi_x} = \ket{x \, 0^{p(n) - n}}.
\end{equation}
\end{proof}

Next, we give an improved construction based on injective OWFs, for which it's hard even to distill classical correlation between $R$ and $B$.

\begin{theorem}[Aaronson, not yet published]
Suppose there exist injective OWFs that are hard to invert by quantum computers.  Then not only is the HH Decoding Task hard, but it's hard even to distill
{\em classical correlation} (let alone entanglement) between the $R$ and $B$ regions.
\end{theorem}
\begin{proof}
For this, we'll need one more basic concept from cryptography, called
a {\em hardcore predicate} for an OWF $f$.  A hardcore predicate is a single bit $h(x)$ about the input $x$, which is ``as hard to learn, given only $f(x)$, as it is to invert $f$ outright''---that is, such that for all polynomial-time algorithms $A$,
$$ \Pr_{x\in \{0,1\}^n }[ A(f(x))=h(x) ] < \frac{1}{2} + \frac{1}{\operatorname*{poly}(n)}.$$
A celebrated result of Goldreich and Levin \cite{GL89} states that, given any injective OWF $f$, if we define the slightly-modified OWF
$$ g(x,s) := \left\langle f(x),s \right\rangle, $$
then $h(x,s) = x \cdot s (\mod 2)$ is a hardcore predicate for $g$.

So let $f$ be an injective OWF secure against quantum computers.  Then consider the following construction:
$$ \ket{\psi}_{RBH} = {1\over\sqrt{2^{2n+1}}} \sum_{x,s \in \{0,1\}^n; a\in \{0,1\}} \ \ket{f(x),s,a}_R \  \ket{(x\cdot s) \oplus a}_B \ \ket{x,s}_H. $$
Again, clearly this state can be prepared efficiently.  To see that it has entanglement between $R$ and $B$: observe that, if we applied a unitary $U$ to $R$ that inverted $f$ (possible because $f$ is injective), and then XORed $x\cdot s$ into the last qubit of $R$, we would have
$$ \frac{1}{\sqrt{2^{2n+1}}} \sum_{x,s \in \{0,1\}^n; a\in \{0,1\}} \ \ket{x,s,(x\cdot s) \oplus a}_R \  \ket{(x\cdot s) \oplus a}_B \ \ket{x,s}_H $$
or equivalently
$$ \frac{1}{\sqrt{2^{2n+1}}} \sum_{x,s \in \{0,1\}^n; a\in \{0,1\}} \ \ket{x,s,a}_R \  \ket{a}_B \ \ket{x,s}_H. $$
On the other hand, if you could detect any classical correlation between $R$ and $B$---not even entanglement---that would imply that given $\left\langle f(x),s \right\rangle$, you could guess the hardcore bit $x\cdot s$ better than chance.  But by Goldreich-Levin, that would imply that you could invert $f$, contrary to assumption.
\end{proof}

\subsection{The ``Generic Case''}
\label{ITSWORSE}

One might complain that the above hardness arguments concerned very special states $\ket{\psi}_{RBH}$.  What about the ``generic case'' of interest for physics?  Well, suppose one applied a random, polynomial-size quantum circuit $C$, which acted across all three registers $R$, $B$, and $H$, like so:
\[
\Qcircuit @C=1em @R=0.7em
{
    \lstick{R} & \ctrl{1}   & \qw & \qw &\rstick{}          \qw  &     \\
           \lstick{} & \targ    & \qw& \qw & \rstick{ }  \qw  \\
          \lstick{} &  \qw & \ctrl{1}& \ctrl{6} & \rstick{ } \qw  \\
    \lstick{B} &  \ctrl{1} & \targ & \qw & \rstick{ } \qw \\
            \lstick{} &  \targ & \qw & \qw & \rstick{ } \qw  \\
            \lstick{} &  \qw & \ctrl{1} & \qw & \rstick{ } \qw  \\
     \lstick{H} &  \ctrl{1} & \targ & \qw & \rstick{ } \qw \\
          \lstick{} &  \targ & \qw & \qw & \rstick{ } \qw  \\
          \lstick{} &  \qw & \qw & \targ& \rstick{ } \qw
}
\]
As before, we want to apply a unitary $U_R$ to the $R$ register in order to decode one qubit of entanglement between $R$ and $B$.  Now, while we lack a rigorous result, intuitively, if anything this seems even {\em harder} than decoding the entanglement for the special states $\ket{\psi}_{RBH}$ that we constructed for the proofs of Theorems \ref{hhthm} and \ref{aarfirewall}!  For the special states, at least there was a classical computational problem---namely, computing $f^{-1}(g(x))$ in the HH construction, or inverting the injective OWF $f$ in Aaronson's constructions---such that, {\em if} we could solve that problem, then the HH Decoding Task would be easy.  Indeed, in both cases, the problem is in $\mathsf{NP}$, meaning that the HH Decoding Task would be easy for these specific states if $\mathsf{NP}\subseteq \mathsf{BQP}$.

For the generic case, by contrast, it's not obvious that there's {\em any} classical computational problem, such that an efficient solution to that problem would make the HH Decoding Task easy.  In other words, we don't have any general converse to Theorems \ref{hhthm} and \ref{aarfirewall}.  We know that, if quantum-resistant injective OWFs exist, then the HH Decoding Task is hard, but we don't even know whether if $\mathsf{P}=\mathsf{PSPACE}$ (for example), then the generic HH Decoding Task is easy.

This directly relates to (and helps to motivate) the Unitary Synthesis Problem, and the other general questions about the complexity of unitary transformations from Lecture 3.

\subsection{Preskill's Problem}

We'll end our discussion of the Harlow-Hayden argument with a beautiful problem, which was communicated to Scott by John Preskill.  Preskill wondered: could one evade the HH argument by feeding a black hole a daily diet of particles, in order to keep it alive for exponential time---enough for Alice to perform the HH Decoding Task?  This then led to a followup question: if one tried to do this, then would one necessarily generate quantum {\em states} at intermediate times with $2^{\Omega(n)}$ quantum circuit complexity?  Or, to put the question differently: is it possible to perform the HH Decoding Task by applying a quantum circuit to $R$ that has exponentially many gates, but that keeps the {\em state} complexity upper-bounded by a polynomial throughout, so that at every time $t$, we have step $ {\cal C}_{{\eps}}(\ket{\psi_t}) \leq n^{0(1)}$?

For the special constructions of $\ket{\psi}_{RBH}$ above, based on Set Equality and injective OWFs, the answer turns out to be yes. For in those cases, it's not hard to see that one can solve the requisite $\mathsf{NP}$ problems---that is, compute $f^{-1}(g(x))$ or $f^{-1}(x)$---in a way that takes exponential time but never makes the current state much more complicated than $\ket{\psi}_{RBH}$ itself.  On the other hand, for the generic case, Preskill's problem remains open: the entanglement-distilling circuits that one could derive from the abstract dimension-counting arguments of Lecture 6 would not only have exponentially many gates, but would also produce intermediate states that plausibly have $2^{\Omega(n)}$ circuit complexity.

\lecture{Scott Aaronson}{Vijay Balasubramanian}{Complexity, AdS/CFT, and Wormholes}

In this lecture, we'll move on from firewalls to describe {\em another} recent connection between quantum circuit complexity and quantum gravity.

The AdS/CFT correspondence is a conjectured (and well-supported, but not proven) duality between quantum gravity in anti de-Sitter (AdS) spacetime  and Conformal Field Theories (CFTs). So what is AdS spacetime, and what is  a CFT?

   Anti de-Sitter (AdS)  is the name given to universes with a negative cosmological constant.  Let's consider such universes in $d+1$ dimensions where the additional $1$ is for the time dimension.   The spatial slices of such a universe are hyperbolic spaces.   To visualize, we can conformally map  hyperbolic space to a finite disc in the plane (think of the pictures of M.C.\ Escher with repeating figures in a disc shrinking as they approach the boundary which has been mapped from infinitely far away to a finite distance).      You can visualize an AdS spacetime (after this conformal mapping) as a solid, infinitely long cylinder.  Time runs along the infinite length.   An equal time slice is a horizontal section of the cylinder, and each such section is a hyperbolic space.  The interior of the cylinder is often referred to as the ``bulk.''  The boundary of the cylinder is the boundary of AdS spacetime and is  infinitely far way in the hyperbolic metric.  In the AdS/CFT correspondence we consider a quantum gravitational theory on AdS spacetime, generally a string theory, and study its physics.  Each such theory comes equipped with many kinds of fields and particles (generally an infinite tower of objects with increasing masses, and a small number of light or massless degrees of freedom which are most interesting to us since they are easy to excite).

   Meanwhile, a Conformal Field Theory (CFT) is a quantum field theory that has a particularly powerful symmetry, called conformal invariance.   By saying that a quantum field theory ``has a symmetry'' we mean that the dynamical equations and energy functional of the theory are invariant under transformation by the associated symmetry group.   This usually means that the vacuum state of the theory is invariant under the symmetry group, but that may not be the case (the vacuum can ``break'' the symmetry).  Regardless of whether the vacuum itself is invariant under the symmetry group action, all states and configurations of the quantum field can be usefully organized in representations of the symmetry group.      In two dimensions ($1+1$) the conformal group is infinite-dimensional and conformal field theories in 2D are highly constrained because of this.  The conformal group in more than $2$ dimensions is finite dimensional, but still produces many constraints on theories.     In the case of the AdS/CFT correspondence, the relevant CFT (which we'll say more about below) is formulated on the conformal boundary of the the $d+1$-dimensional AdS spacetime.  As described above, this boundary is $d$-dimensional, a $(d-1)$-dimensional sphere plus time.

   We can discretize the space and time on which the CFT is defined and idealize it as a quantum circuit acting on a finite number of qubits:
 \[
\Qcircuit @C=1em @R=0.7em
{
    \lstick{A} & \ctrl{1}   & \qw & \qw &\rstick{}          \qw       \\
           \lstick{Schematic} & \ctrl{0}    & \qw& \ctrl{1}& \rstick{ }  \qw  \\
          \lstick{Quantum} &  \qw & \ctrl{1}& \ctrl{0s} & \rstick{ } \qw  \\
    \lstick{Field} &  \ctrl{1} & \ctrl{0} & \qw & \rstick{ } \qw \\
            \lstick{Theory} &  \ctrl{0}s & \qw & \ctrl{1} & \rstick{ } \qw  \\
            \lstick{} &  \qw & \ctrl{1} & \ctrl{0} & \rstick{ } \qw  \\
     \lstick{} &  \ctrl{1} & \ctrl{0} & \qw & \rstick{ } \qw \\
          \lstick{} &  \ctrl{0} & \qw & \ctrl{1} & \rstick{ } \qw  \\
          \lstick{} &  \qw & \qw & \ctrl{0}& \rstick{ } \qw
}
\]
Here the circuit is shown as acting locally between neighboring qubits because we're interested in {\it local} CFTs, which have spatially local interactions.  The CFT interactions are occurring ``at all times'' and ``between every pair of neighboring points,'' so the above is just a very crude schematic to help the imagination.   A CFT generally has a finite number of interacting quantum fields, but, importantly, does not contain dynamical gravity.

The {\em AdS/CFT correspondence} is a conjectured equivalence or duality between certain quantum theories of gravity (string theories) acting in AdS spacetime and corresponding CFTs associated to the AdS boundary.  On the AdS side the theory is $(d+1)$-dimensional.  On the CFT side the theory is $d$-dimensional.  So if the equivalence is true, one dimension of the AdS space can be regarded as ``emergent'' in the CFT, i.e.\ arising as an effective description of lower-dimensional dynamics.  Likewise on the AdS side the theory contains gravity.  On the CFT side there is no gravity.  So from the CFT perspective, the gravity is emergent also.

All very cool indeed---but what does it mean for two such theories to be ``equivalent''?  First, the symmetries must match.  Thus, the conformal group of the CFT is realized on the gravity side as an isometry group.   The CFT typically contains additional symmetries, which also get realized as symmetries of the theory on the AdS side.  In addition, the objects in the two theories must be in correspondence---that is, for each operator that creates objects and states in the CFT there must be a corresponding quantum field in the AdS theory.  In the proposed examples of the AdS/CFT correspondence this is the case.   In particular, in all such examples, the graviton field in AdS space is associated to the stress-energy tensor of the dual CFT.   It should be said that it's rather non-trivial to find examples of fully defined and self-consistent theories in different dimensions that can match up in this way.    Finally, once we have a dictionary matching up the objects and states on the AdS and CFT sides, the remaining thing is to specify how to compute observable quantities---the AdS/CFT duality states that you get the same answer for the observable by computing on either side of the correspondence.  There's an algorithm for systematically computing observables on the CFT side (essentially correlation functions) and on the AdS side.  In the cases where computations have been possible on both sides they have always agreed, but a general proof of equivalence---and for that matter, even a general, independent {\em definition} of the AdS theory---haven't yet been given.

\subsection{Susskind's Puzzle}

The map between bulk and boundary doesn't seem to be ``complexity-preserving.''   That is, simple states on one side can map onto complex states on the other.

 To be concrete, let's consider two separate regions of space connected by a non-traversible wormhole.  The non-traversibility occurs because the wormhole stretches rapidly.   In general relativity, such wormholes are described by geometries called Einstein-Rosen (ER) bridges.

In the AdS/CFT correspondence, the CFT dual of the wormhole is a so-called ``Thermofield Double State'' (TDS).   This means that the state is maximally entangled between two identical Hilbert spaces:
\begin{equation}
TDS = {1 \over \sqrt{N}} \sum_{i=1}^N \ket{i}_A \ket{i}_B.
\end{equation}
Recalling that the weird properties of maximally entangled states were originally emphasized in a famous paper of Einstein, Podolsky and Rosen (EPR), Susskind and Maldacena \cite{MS13} summarized this relation between ER bridges and entanglement by the slogan ``ER = EPR.''

Now the wormhole gets longer as time passes.  But what happens to the dual CFT state?  The state changes through some unitary evolution:
$$ TDS(t) = \frac{1}{\sqrt{N}}  \sum_i V^t \ket{i}_A \otimes  (V^T)^t \ket{i}_B. $$
Recall, the maximally entangled state has the property that
$$ \frac{1}{\sqrt{N}} \sum_{i=1}^N V \ket{i} \otimes \ket{i} = {1 \over \sqrt{N}} \sum_{i=1}^N \ket{i} \otimes V^T \ket{i}. $$
Thus, we actually have
$$ TDS(t) = {1 \over \sqrt{N}} \sum_{i=1}^N \ket{i} \otimes U^t \ket{i} $$
for some unitary $U = (V^T)^2$.

Assume that $U$ is implemented by a polynomial-sized quantum circuit.  Then we can state Susskind's puzzle as follows: if we apply $U U \cdots U \ket{0}^{\otimes n}$, we'll get something that very quickly ``thermalizes.''  To see this, we can consider what physicists call {\em $k$-point correlation functions}---or equivalently, the local density matrices of $k=O(1)$ qubits at a time.   For example, the density matrix of a single qubit will rapidly converge to the maximally mixed state, and similarly for the state of any small number of qubits.  Furthermore, after the local density matrices have thermalized in this way---namely, after about $n \log n$ gates have been applied---one can show that, with overwhelming probability, the local density matrices will remain in that thermalized condition for {\em doubly-}exponential time.

But on the AdS side, Susskind pointed out, rapid thermalization is {\it not} what's happening---rather, the wormhole just keeps getting longer and longer. Thus, let $\ket{\psi_t}$ be the CFT state after a time $t$.  Then Susskind raised the question: {\em what quantity $f(\ket{\psi_t})$ is the CFT dual to the wormhole length?}

He made an astonishing proposal: namely, the quantum circuit complexity of the state, i.e., $f(\ket{\psi_t}) = {\cal C}_{\eps} (\ket{\psi_t})$.  In the following sections, we'll discuss the heuristic justifications for Susskind's proposal, as well as some recent work aiming to prove rigorously that (under reasonable complexity hypotheses) ${\cal C}_{\eps} (\ket{\psi_t})$ does indeed grow like $f(\ket{\psi_t})$ should.

\section{The Proposal}

To recap, the CFT dual of the thermofield double state at time $t$ is given by
\[
\ket{\psi_t} = \frac{1}{2^{n/2}} \sum_{x \in \{0,1\}^n} \ket{x} \otimes U^t \ket{x}.
\]
Our question was: what property of this state grows linearly with $t$, and as such can explain the wormhole's linear growth in volume over time (which follows from general relativity)? No standard information-theoretic measure seems to suffice. Instead, Susskind's proposal is that the \emph{quantum circuit complexity} ${\cal C}_{\eps}(\ket{\psi_t})$ of the state $\ket{\psi_t}$ is what corresponds to the volume.  Clearly,
\[
{\cal C}_\eps(\psi_t) \le t \cdot \poly(n),
\]
since we can generate $\psi_t$ from $n$ Bell pairs by simply applying $U$ repeatedly $t$ times. (Here we're assuming that $U$ is a ``simple'' unitary that can be implemented by a quantum circuit of size $\poly(n)$, as is the case for
other unitaries in nature.)  One might expect that the upper bound above is tight and the circuit complexity does grow linearly, at least until time $t\approx 2^n$. Beyond that time, the circuit complexity stops growing, since it's easy to see that for \emph{any} $n$-qubit state $\ket{\psi}$,
\[
{\cal C}_\eps(\psi) \le 2^{n} \cdot \poly(n),
\]
which follows by preparing the state one amplitude at a time.
To summarize, one can expect the state complexity to behave as shown in Figure \ref{complexitygraph}.

\begin{figure}[ht]
\label{complexitygraph}
\begin{center}
\includegraphics[width= 0.6\textwidth]{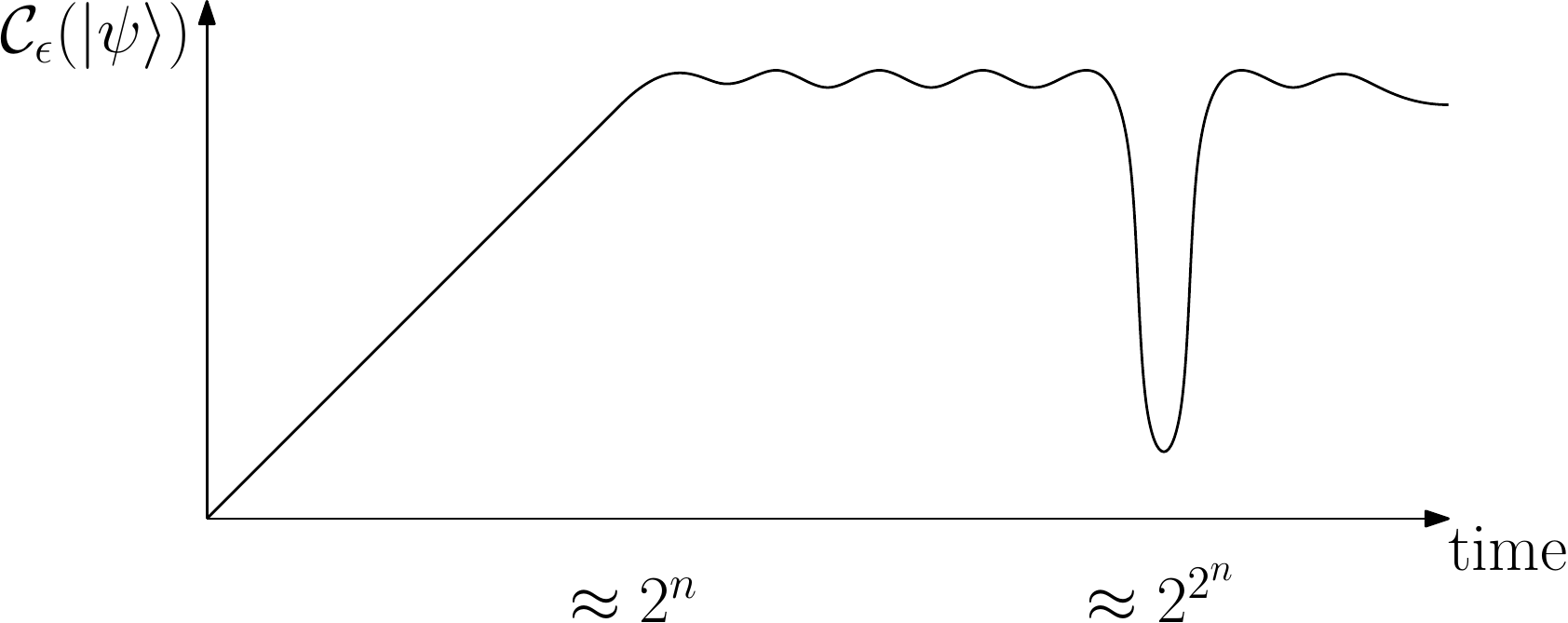}
\end{center}
\end{figure}

The reason for the dip is that once in roughly $\sim 2^{2^n}$ steps, we can expect the ``random-looking'' $\ket{\psi_t}$ to happen to land near a low-complexity state.  (This can also be seen by considering the eigenvalues of $U$; by the pigeonhole principle, one can show that after at most $\sim 2^{2^n}$ steps, we'll reach a $t$ for which $U^t$ is very close to the identity, in which case the state has low complexity.)  This speculative behavior of the state complexity agrees with the known predictions for the volume of a wormhole (based on quantum gravity). Other predictions on the behavior of the volume also agree with this measure:
\begin{enumerate}
    \item Imagine that partway through the process, instead of multiplying by $U$, we start multiplying by $U^{-1}$. We'd expect the wormhole to start shrinking at the same speed it previously grew. And indeed, this is exactly what happens to the state complexity.
    \item More interestingly, imagine that before switching from $U$ to $U^{-1}$, we apply once a unitary $V$ not related to $U$. In that case, the prediction from quantum gravity turns out to be that the wormhole would \emph{continue growing} instead of shrinking as before. And indeed, that's also what we expect from the state complexity (see for example \cite{SS14}):
\[
{\cal C}_\eps\Big(\frac{1}{2^{n/2}} \sum_{x \in \{0,1\}^n} \ket{x} \otimes U^{-t} V U^t \ket{x} \Big) \approx 2t \cdot \poly(n).
\]
\end{enumerate}

The idea that circuit complexity is dual to wormhole volume remains speculative: even if we agree that the two quantities are correlated, we know that in science correlations need not be causal, but could be explained by some yet-unknown third factor.

We could also ask: if circuit complexity is dual to wormhole volume, then {\em which} variant of circuit complexity are we talking about?  For example, exact or approximate circuit complexity?  Are ancilla qubits allowed?  Must garbage be uncomputed (recall the discussion in Lecture 3)?  Also, surely we can't expect the laws of physics to know about nontrivial efficient algorithms like (say) Edmonds' algorithm for maximum matching, or the ellipsoid algorithm for convex programming.  Yet those algorithms could in principle be the reason why a state had polynomial rather than exponential quantum circuit complexity.

Nevertheless, right now circuit complexity is essentially the {\em only} proposal that produces the right answers in the cases of interest.  Other quantities, like entropy, are non-starters, because they saturate too quickly.  There {\em is} the tautological proposal of defining the CFT dual of wormhole volume, for a given state $\ket{\psi}$, to be the minimum value of the time $t$ that satisfies the equation $\ket{\psi}=\ket{\psi_t}$.  But this proposal would require detailed knowledge of the specific Hamiltonian or unitary $U$.

In summary, circuit complexity {\em might} ultimately turn out to be just one way among others to define an ``intrinsic clock'' for a quantum state, to obtain $t$ from $\ket{\psi_t}$.  But at present, we have no other principled way to do that.  So for now, circuit complexity has established itself as a useful tool in the study of the AdS/CFT correspondence.

\subsection{Lower Bounds}

The preceding discussion left at least one question glaringly open: namely, can we \emph{prove} any lower bounds on ${\cal C}_\eps(\ket{\psi_t})$, for natural examples of unitary transformations $U$, showing that it really does grow linearly with $t$, as Susskind's proposal would predict and require?

Alas, this is too much to hope for, in the present state of complexity theory.  Here's the reason:

\begin{proposition}
\label{nope}
Suppose ${\cal C}(\ket{\psi_{b^n}}) > n^{\omega(1)}$, for some $U$ that admits a polynomial-size quantum circuit.  Then $\mathsf{PSPACE} \not\subset \mathsf{BQP/poly}$.
\end{proposition}
\begin{proof}
Consider the contrapositive: if $\mathsf{PSPACE} \subset \mathsf{BQP/poly}$, then we could prepare $\ket{\psi_{t}}$ for any $t\leq b^n$ in polynomial time, by computing the ``aggregated amplitudes'' $\beta_x$ and $\gamma_x$ in $\mathsf{PSPACE}$, as in the method of Proposition \ref{canmake} from Lecture 3.
\end{proof}

(Incidentally, by a simple counting/dimension argument, analogous to that in Lecture 3, one {\em can} prove without much difficulty that ${\cal C}(\ket{\psi_t})$ will become exponentially large after {\em doubly-}exponential time, $t\sim 2^{2^n}$.  But this isn't enough for Susskind's proposal.)

As a consequence of Proposition \ref{nope}, one can't possibly prove that circuit complexity grows linearly with time under unitary evolution, without also proving $\mathsf{PSPACE} \not\subset \mathsf{BQP/poly}$.

Still, despite Proposition \ref{nope}, perhaps we {\em can} prove a lower bound on ${\cal C}_\eps(\ket{\psi_t})$ assuming some reasonable complexity hypothesis? Indeed we can, as in the following result.

\begin{theorem}[Aaronson-Susskind (paper still in preparation)]
\label{thm:asusskind}
There's an $n$-qubit unitary transformation $U$, implemented by a polynomial-size quantum circuit, for which the following holds.  Let
\[
\ket{\psi_t} = \frac{1}{2^{n/2}} \sum_{y \in \{0,1\}^n} \ket{y} \otimes U^t \ket{y}.
\]
Then ${\cal C}_\eps(\ket{\psi_t}) > n^{\omega(1)}$ for some $t > c^n$, unless $\mathsf{PSPACE} \subset \mathsf{PP/poly}$.  Here $1<c<2$ is an absolute constant.
\end{theorem}

We remark that one can get a stronger conclusion at the cost of a stronger complexity assumption.  For example, if $\mathsf{SPACE}(n)$ requires $\mathsf{PP}$-circuits of size $2^{\Omega(n)}$, then ${\cal C}_\eps(\ket{\psi_t}) = 2^{\Omega(n)}$ for time $t=2^{O(n)}$.  On the other hand, because of blowup in the reduction, we don't currently know how to prove, on the basis of any ``standard'' complexity assumption, that ${\cal C}_\eps(\ket{\psi_t})$ grows {\em linearly} with $t$; that remains a question for future work.

\begin{proof}[Proof of Theorem \ref{thm:asusskind}]
Take $U$ to be the step function of some reversible, computationally-universal cellular automaton.  For this construction, $U$ can even be a {\em classical} reversible CA, but we do need it to be able to solve $\mathsf{PSPACE}$-complete problems using polynomial memory and exponential time.  Fortunately, it's well-known that there are many reversible CAs with this behavior.  A result due to Lange, McKenzie, and Tapp \cite{lmt} even assures us that, if we don't care about exponential overhead in {\em time} (as we don't, in this case), then making a Turing machine reversible incurs only a constant-factor overhead in the required memory.

So suppose toward a contradiction that
${\cal C}_\eps(\ket{\psi_t}) \le n^{O(1)}$ for all $t \le 2^n$. Let $L$ be a $\mathsf{PSPACE}$-complete language.  Then we'll see how to decide $L$ in the complexity class $\mathsf{PostBQP/poly} = \mathsf{PP/poly}$.  The advice to the $\mathsf{PostBQP}$ machine will simply be a description of a small quantum circuit that prepares the state $\ket{\psi_t}$, for some $t=c^n$ that provides sufficient time for the $\mathsf{PSPACE}$ computation to have returned its answer.

So the $\mathsf{PostBQP/poly}$ machine prepares this state $\ket{\psi_t}$, and then postselects on the first register of $\ket{\psi_t}$ being in the state $y = S_x$, where $S_x$ is the encoding of the input $x$ used by the cellular automaton.  Conditioned on the first register being observed in the state $\ket{S_x}$, the second register will be in the state $U^t \ket{S_x}$, which encodes the information about whether $x\in L$.  Therefore $\mathsf{PSPACE} \subset \mathsf{PP/poly}$ as desired.

(One fine point: we need $\eps$ exponentially small to make sure that the error in the preparation doesn't affect the post-selected state.)
\end{proof}

Note that, in the proof of Theorem \ref{thm:asusskind}, we needed {\em some} condition on the unitary $U$.  If, for example, $U=I$ were the identity, then clearly we'd never increase complexity no matter how often $U$ was applied.

Previously, some quantum gravity theorists had speculated that the relevant condition on $U$ would have something to do with its chaotic mixing behavior, or sensitive dependence on initial conditions.  But the condition that we were actually able to use, for a hardness reduction, was basically {\em computational universality}.  Universality might be related to chaos, but the exact nature of the relationship is unclear.

The good news is that ``being able to implement universal reversible computation'' is an extremely mild condition on a unitary or Hamiltonian---one that we can confidently expect the real laws of quantum gravity (whatever they are) to satisfy.  Thus, unlike with (say) the proofs of Theorems \ref{hhthm} and \ref{aarfirewall}, Theorem \ref{thm:asusskind} is not a case of a hardness reduction producing instances of the target problem unlike those that would ever arise in real life.  Rather, this argument should apply even to the physically-relevant choices of $U$.

If we care about the circuit complexity of {\em approximating} $\ket{\psi_t}$---that is, about ${\cal C}_\eps(\ket{\psi_t})$ for large values of $\eps$---we can use a similar reduction, with the difference that now we need to start from a $\mathsf{PSPACE}$-complete language $L$ with the property of {\em worst-case/average-case equivalence}.  In other words, we need a problem that's $\mathsf{PSPACE}$-complete to solve even for (say) $90\%$ of inputs $x\in \{0,1\}^n$, so that even advice that let a $\mathsf{PostBQP/poly}$ machine prepare $\ket{\psi_{b^n}}$ to within some constant error would still let that machine decide $\mathsf{PSPACE}$.  Fortunately, such problems are well-known to exist; see for example \cite{AFK89}.  (As an example, one can start with any $\mathsf{PSPACE}$-complete Boolean function $f:\{0,1\}^n \longrightarrow \{0,1\}$, then take the unique extension of $f$ as a multilinear polynomial over a large finite field.)

As we alluded to before, there are several reasonable ways to define the ``quantum circuit complexity'' of the state $\ket{\psi_t}$, other than the way we chose.  Some of those ways lead to interesting changes to the picture of CFT complexity that we've outlined so far.

To illustrate, let's consider a measure that we call {\em separable complexity}, and
denote ${\cal C}_{\mathrm{sep},\eps}(\ket{\psi})$.  This is defined as the size of the smallest \emph{separable} quantum circuit $C$ that generates $\ket{\psi}$ to within error $\eps$, starting from the initial state
$$ \frac{1}{\sqrt{2^n}} \sum_{x\in \{0,1\}^n} \ket{x}\ket{x}. $$
Here, separable means that $C$ must have no gates that cut across the boundary between the two subsystems, i.e., it must act independently on both halves of the state.  (Restricting to such circuits seems reasonable, since the two halves of the state are located in two regions of space that are only connected by a non-traversable wormhole!)  We also don't allow the use ancilla qubits in this definition.

Because we assumed the initial state to be maximally entangled, the above is equivalent to requiring the circuit to act only on the second register.  This is again because for any unitary $U$, we have
$$\sum_x U \ket{x} \otimes \ket{x} = \sum_x \ket{x} \otimes U^T \ket{x}.$$

We now observe that ${\cal C}_{\mathrm{sep},\eps}(\ket{\psi_t}) =
{\cal C}_\eps (U^t)$: in other words, {\em the separable circuit complexity of $\ket{\psi_t}$, as a state, equals the circuit unitary complexity of $U^t$ as a unitary}.
To see this:
\begin{itemize}
\item The $\le$ direction is obvious---if we have a circuit to compute $U^t$, then we can generate $\ket{\psi_t}$ from the maximally entangled state by acting only on the second register.
\item The $\ge$ direction follows from the property of the maximally entangled state that for all unitaries $U \neq V$,
$$\sum_x \ket{x} \otimes U \ket{x} \neq \sum_x \ket{x} \otimes V \ket{x}.$$
I.e., the only way to generate $\ket{\psi_t}$ from the maximally entangled state, by acting on the second register alone, is to apply $U^t$ itself. Furthermore, while in principle one could apply circuits $U$ and $V$ to the two registers separately, it's easy to see that one doesn't gain anything this way: one could instead just apply $VU^{T}$ to the second register, using a circuit with at most ${\cal C}(U)+{\cal C}(V)$ gates.  (Here we assume, for simplicity, that our gate set is closed under transposition.  By the Solovay-Kitaev Theorem, dropping this assumption increases the circuit complexity by at most a $\log^{O(1)}(1/\eps)$ factor.)
\end{itemize}
As a result, we see that if ${\cal C}_{\mathrm{sep},\eps}(\ket{\psi_t}) \le n^{O(1)}$, then $\mathsf{PSPACE} \subset \mathsf{BQP/poly}$---a stronger conclusion than the $\mathsf{PSPACE} \subset \mathsf{PP/poly}$ that we got from Theorem~\ref{thm:asusskind} for the standard (non-separable) complexity.

To summarize, we have the following chain of implications:
\[
{\cal C}_{\mathrm{sep},\eps}(\ket{\psi_t}) \le n^{O(1)} \Leftrightarrow {\cal C}_{\mathrm{sep},\eps}(U^{b^n}) \le n^{O(1)}
\]
\[
\Downarrow
\]
\[
\mathsf{PSPACE} \subset \mathsf{BQP/poly} \Leftrightarrow {\cal C}_\eps(U^{b^n}) \le n^{O(1)}
\]
\[
\Downarrow
\]
\[
{\cal C}_{\eps}(\ket{\psi_t}) \le n^{O(1)}
\]
\[
\Downarrow
\]
\[
\mathsf{PSPACE} \subset \mathsf{PP/poly}
\]
where $U$ implements the step function of a reversible classical universal cellular automaton.
It's open whether one can strengthen these results.  Note also that a solution to the Unitary Synthesis Problem, from Lecture 3, might let us reverse the first ``down'' implication.

\lecture{Scott Aaronson}{Oded Regev}{Private-Key Quantum Money}

We now start on our final application of the complexity of states and unitaries---one that, on its face, seems about as remote as possible from quantum gravity, though it will turn out not to be completely unrelated.

In the late 1960s, long before anyone even dreamed of quantum computation, Stephen Wiesner, then a graduate student, wrote a paper proposing to use quantum physics to construct unforgeable money. The paper was repeatedly rejected until it appeared 14 years later in {\em SIGACT News} \cite{wiesner}.

Because of the copyability of classical information, all existing approaches to electronic cash have to rely either on a trusted third party like a bank or credit card agency, or sometimes on the entire Internet (in the case of Bitcoin's block chain).  To get around this, Wiesner's basic idea was to exploit what we now call the No-Cloning Theorem (from Lecture 2), stating that there's no way to duplicate an unknown quantum state.  Why couldn't we use this phenomenon to create ``quantum banknotes,'' which could be held and traded, yet which would be physically impossible to copy?

In implementing quantum money, the obvious engineering difficulty is that the users need to transport quantum states around (ideally, at room temperature and in their wallets), while keeping them coherent for more than a tiny fraction of a second!  It's mainly this issue that's so far prevented the practical realization of quantum money---in contrast to {\em quantum key distribution}, a closely-related idea (also originating in Wiesner's 1969 paper) that's already seen some modest commercial deployment.  In any case, from now on we'll ignore implementation issues and concentrate on the theory of quantum money.

For simplicity, we'll assume throughout that quantum money comes in only one denomination.

At a high level, quantum money needs to satisfy two requirements: it needs to be
\begin{enumerate}
\item[(1)] unclonable (obviously), but also
\item[(2)] verifiable---that is, users should be able to verify that a quantum money state presented to them is a valid state and not some junk.
\end{enumerate}

\section{The Wiesner and BBBW Schemes}

Wiesner's proposal is now as follows.  Each quantum banknote consists of two parts. The first part is an $n$-bit classical string $s$, called the {\em serial number}, chosen independently and uniformly for each note.  Just like with today's classical money, the serial number serves to identify each banknote uniquely.

The second part is an $n$-qubit quantum state of the form
\[ \ket{0} \ket{-} \ket{1} \ket{1} \ket{+} \cdots \ket{0} \]
which is chosen uniformly from among all $4^n$ possible $n$-qubit tensor products of $\ket{0}$, $\ket{1}$, $\ket{+}$, and $\ket{-}$. The bank that generates the notes stores
a database of pairs, $(s,f(s))$ where $f(s)$ is a classical description of the quantum state $\ket{\psi_s}$ generated for note $s$ (note that $f(s)$ is $2n$ bits long).

Clearly, if a user takes a note $(s,\ket{\psi})$ back to the bank, the bank can verify the note's veracity.  Indeed, the bank simply needs to check in its database that $s$ is a valid serial number, and then if it is, measure each qubit of $\ket{\psi_s}$ in the appropriate basis ($\{ \ket{0},\ket{1}\}$ or $\{ \ket{+},\ket{-}\}$) to make sure
that it's in the state corresponding to $f(s)$.  A user who wants to verify a note needs to go to the bank (or send the note to the bank over a quantum channel).

What about counterfeiting?  Assume the counterfeiter is given a legitimate banknote, $(s,\ket{\psi_s})$.  Can he generate from it two notes that both pass the bank's verification with high probability?

The na\"{i}ve strategy to do so would simply be to measure each qubit in (say) the standard basis and then copy the result.  This strategy succeeds with probability $(5/8)^n$. (Why? Qubits that happen to be in the standard basis are copied successfully; this happens with probability $1/2$.  The remaining qubits are damaged by the measurement, so that the result has probability $1/4$ of passing both measurements.)

A less trivial strategy generates two entangled notes such that both pass verification with probability $(3/4)^n$.  This turns out to be tight!  Answering a question posed by Aaronson, this was proven in 2012 by Molina, Vidick, and Watrous \cite{mvw}, and independently by Pastawski et al.\ \cite{pastawski}, using semidefinite formulations of the problem faced by the counterfeiter.  (For some reason, from 1969 until 2012, people had been satisfied with handwavy, qualitative security analyses of Wiesner's scheme.)

Strictly speaking, the results of \cite{mvw,pastawski} don't constitute a full security proof, since (for example) they don't consider a counterfeiter who starts with {\em multiple} legitimate banknotes, $(s_1,\ket{\psi_{s_1}}),\ldots,(s_m,\ket{\psi_{s_m}})$.  However, in recent unpublished work, Aaronson gives a general security reduction that implies, as a special case, the full security of Wiesner's scheme.

The obvious disadvantage of Wiesner's scheme is the need to bring a banknote back to the bank in order to verify it.  However, even if we set that aside, a second disadvantage is that in order to verify, the bank (and all its branches) need to maintain a database of all banknotes ever produced.
In 1982, Bennett, Brassard, Breidbart, and Wiesner (BBBW) \cite{bbbw} suggested using a standard cryptographic trick to get rid of the giant database. Namely, they proposed replacing the random function $f(s)$ by a cryptographic \emph{pseudorandom function}---or more precisely, by a function $f_k:\{0,1\}^n\longrightarrow \{0,1\}^{2n}$ chosen from a pseudorandom function family (PRF) $\{f_k\}_k$.  That way, the bank need only store a single secret key $k$.  Given a banknote $(s,\ket{f_k(s)})$, the bank then verifies the note by first computing $f_k(s)$ and then measuring each qubit of $\ket{f_k(s)}$ in the appropriate basis, as before.

Intuitively, by doing this we shouldn't be compromising the security of the scheme, since a pseudorandom function can't be efficiently distinguished from a truly random function anyway.  More formally, we argue as follows: suppose a counterfeiter could successfully attack the BBBW scheme.  Then we could distinguish the pseudorandom function $f_k$ from a truly random function, by running the attack with the given function and then checking whether the attack succeeds in counterfeiting.  Given a pseudorandom function, the attack should succeed by assumption, whereas given a truly random function, the attack {\em can't} succeed because of the security of Wiesner's original scheme.

\section{Formal Underpinnings}

To make the above discussion rigorous, we must define what exactly we mean by quantum money, and what properties we need it to satisfy.  This was done by Aaronson \cite{Aar09} in 2009.

\begin{definition}[\cite{Aar09}] \label{def:private-key-quantum-money}
A \emph{private-key quantum money scheme} consists of two polynomial-time quantum algorithms:
\begin{itemize}
\item $\mathrm{Bank}(k)$ generates a banknote $\$_k$ corresponding to a key $k$. ($\$_k$ can be a mixed state corresponding, e.g., to a mixture over serial numbers.)
\item $\mathrm{Ver}(k,\rho)$ verifies that $\rho$ is a valid note for the key $k$.
\end{itemize}
We say that the scheme has \emph{completeness error} $\eps$ if valid banknotes $\$_k$ generated by $\mathrm{Bank}(k)$ are accepted by $\mathrm{Ver}(k,\$_k)$ with probability at least $1-\eps$.
We say that the scheme has {\em soundness error} $\delta$ if for all polynomial-time counterfeiters $C$ outputting some number $r$ of registers, and all polynomials $q$,
\[
\Pr[\mathrm{Count}(k,C(\$_k^q)) > q] < \delta \;,
\]
where $\mathrm{Count}$ is the procedure that runs $\mathrm{Ver}(k,\cdot)$ on each of the $r$ registers output by $C$ and counts the number of times it accepts. (So informally, we're asking that given $q$ banknotes, $C$ cannot generate more than $q$ banknotes, except with probability $\delta$.)
\end{definition}

Let's also define a ``mini-scheme'' to be a scheme as above but with soundness restricted to $q=1$ and $r=2$, i.e., there's only one bill in circulation and the counterfeiter is asked to produce two valid bills.

\begin{theorem}[Aaronson 2013, unpublished]
We can build a private-key quantum money scheme with negligible completeness and soundness errors given
\begin{itemize}
\item a secure mini-scheme, and
\item a pseudorandom function family (PRF) $\{f_k\}_k$ secure against quantum attacks.
\end{itemize}
\end{theorem}
\begin{proof}[Proof Sketch]
The construction is as follows. Given a mini-scheme $M$ that generates a quantum state $\$_k$, we define a money scheme $M'$ whose bills are of the form
$$\$'_{k,s} := \ket{s}\bra{s} \otimes \$_{f_k(s)},$$
where $s$ is chosen uniformly at random. For verification, the bank applies $\mathrm{Ver}(f_k(s))$.

We remark that when we apply this construction to the ``mini-Wiesner'' scheme (i.e., Wiesner's scheme without the serial number), we get the BBBW scheme.  To get Wiesner's original scheme, we should replace $f_k$ by a truly random function $f$.

The proof of completeness is easy, while soundness (or security) uses a hybrid argument.
Intriguingly, the proof of soundness uses the PRF assumption twice: first one argues that breaking $M'$ implies either an attack on the PRF (i.e., that it can be distinguished from a uniform function) or an attack on $M'$ as a mini-scheme; second, one argues that an attack on $M'$ as a mini-scheme implies either an attack on the PRF, or an attack on $M$ (as a mini-scheme).
\end{proof}

What should we use as the PRF family? There are many constructions of PRF families in cryptography based on various cryptographic assumptions.
There's also a general result, which says that the existence of a one-way function (the most basic of all cryptographic primitives, defined in Lecture 6, without which there's basically no crypto) is already enough to imply the existence of PRFs. This implication is proved in two steps:
\[
\mathrm{OWF} \stackrel{HILL}{\Longrightarrow} \mathrm{PRG} \stackrel{GGM}{\Longrightarrow} \mathrm{PRF} \; ,
\]
The first step---that OWFs imply pseudorandom generators (PRGs), stretching $n$ random bits into $n^{O(1)}$ pseudorandom bits---is due to H{\aa}stad et al.\ \cite{hill}.  The second step---that PRGs imply PRFs, stretching $n$ random bits into $2^n$ pseudorandom bits---is due to Goldreich, Goldwasser, and Micali \cite{ggm}.  Both steps are nontrivial, the first extremely so.

In our case, of course, we need a PRF family secure against \emph{quantum} attacks.  Fortunately, as discussed in Lecture 6, there are many known candidate OWFs that are quantum-secure as far as anyone knows.  Also, one can show, by adapting the HILL and GGM arguments, that quantum-secure PRFs follow from quantum-secure OWFs.  The first step, that of HILL, applies as-is to the quantum setting.  The GGM step, however, breaks in the quantum setting, because it makes crucial use of ``the set of all inputs queried by the attacker,'' and the fact that this set is of at most polynomial size---something that's no longer true when the attacker can query all $2^n$ input positions in superposition.  Luckily, though, in 2012, Zhandry \cite{Zha12} found a different proof of the GGM security reduction that does apply to quantum attackers.\footnote{On the other hand, for applications to quantum money, one can assume {\em classical} queries to the pseudorandom function $f_k$---since the set of banknotes that the attacker possesses is a determinate, classical set---and thus one doesn't strictly speaking need the power of Zhandry's result.}

\section{Security/Storage Tradeoff}

To summarize, we've seen that Wiesner's scheme is unconditionally secure but requires a huge database, and that it can be transformed as above to the BBBW scheme, which is only computationally secure but which doesn't require a huge database.  We also remarked that quantum-secure PRFs, as needed in the BBBW scheme, are known to follow from quantum-secure OWFs.

Given the discussion above, one might wonder whether there are quantum money schemes that
are (1) information-theoretically secure, and (2) don't require lots of storage.  It turns out that such schemes are impossible.

\begin{theorem}[Aaronson 2013, unpublished]
\label{thm:sront_counterfeiter}
Any scheme in which the bank stores only $n$ bits can be broken by an exponential-time
quantum counterfeiter, using only $\poly(n)$ legitimate money states, and $O(n)$ verification queries to the bank. (Alternatively, we can avoid the verification queries, and instead get an output that passes verification with probability $\Omega(1/n)$.)
\end{theorem}

\begin{proof}
We fix the notation $k^\ast\in\{0,1\}^n$ for the actual secret stored by the bank.  ``We'' (that is, the counterfeiter) are given the state $\$_{k^\ast}^{\otimes m}$ for $m=n^{O(1)}$, and we need to achieve a non-negligible probability of producing $m+1$ or more quantum money states that all pass verification.

A na\"{i}ve strategy would be brute-force search for the key $k^\ast$, but that doesn't work for obvious reasons.  Namely, we have only $\poly(n)$ legitimate bills at our disposal (and the bank is not {\em that} naive to issue us a brand-new bill if ours doesn't pass verification!). And we can make only $\poly(n)$ queries to the bank.

So instead, we'll recurse on the set of ``potential candidate secrets $S$,'' and the crucial statement will say that at any given iteration, $S$ shrinks by at least a constant factor.

A word of warning: we do {\em not} promise to find the actual value $k^\ast$ (e.g. if $\$_{k^\ast}=\$_{k'}$ for some other $k'$, we can never distinguish between the two). We only claim that we can come up with a ``linotype'' $S\ni k^\ast$ that's good enough for forging purposes.

Initially, we don't possess any knowledge of $k^\ast$, so we simply set $S:=\{0,1\}^n$.

Now let $S$ be our current guess.  We generate (not surprisingly) the uniform mixture over the corresponding bills:
$$
\sigma_S = \frac 1{|S|} \sum_{k\in S} \proj{\$_k}
$$
and submit it to the bank.

If $\sigma_S$ is accepted, then $S$ is our linotype; we can use it to print new bills.

The more interesting case is when $\mathrm{Ver}(k^\ast,\sigma_S)$ rejects with high probability. Then, having unlimited computational power, we can construct for ourselves the set $U\ni k^\ast$ of all those $k$ for which $\mathrm{Ver}(k,\sigma_S)$ rejects with high probability. Next, we use $\poly(n)$ copies of the legitimate state $\$_{k^\ast}$ to find (we'll later explain how to do it) {\em some} key $k'\in U$ for which $\mathrm{Ver}(k',\$_{k^\ast})$ accepts with high probability.  And lastly, we use this key $k'$ to go over all individual states $\$_k$ in the mixture $\sigma_S$, and weed out all those $k$ for which $\mathrm{Ver}(k',\$_k)$ rejects with high probability. Note that $k^\ast$ survives (which is the only property we need to maintain).  Meanwhile, the fact that at least a constant fraction of entries in $S$ is removed follows from going through the definitions and from a simple application of Markov's inequality.

So all that's left is to explain how to find $k'$ using only $\poly(n)$ legitimate bills.

For that, we use two more ideas: error amplification and the ``Almost As Good As New Lemma'' (Lemma \ref{lem:aagan} from Lecture 1).  Using error amplification and $\poly(n)$ copies of the legitimate bill, we can assume the following without loss of generality:
\begin{itemize}
\item $\mathrm{Ver}(k^\ast, \$_{k^\ast})$ accepts with the overwhelming probability $1-\exp(-Cn)$, whereas
\item for any ``bad'' key $k$ (defined as a key for which the original acceptance probability is less than $0.8$), $\mathrm{Ver}(k, \$_{k^\ast})$ accepts with exponentially small probability $\exp(-Cn)$.
\end{itemize}
Now consider doing a brute-force search through $U$, say in lexicographic order, repairing our state $\$_{k^\ast}$ using the Almost As Good As New Lemma as we go along.  Then the hope is that our state will still be in a usable shape by the time we reach $k^\ast$---since any ``bad'' keys $k$ that we encounter along the way will damage the state by at most an exponentially small amount.

Unfortunately, the above doesn't {\em quite} work; there's still a technical problem that we need to deal with.  Namely: what if, while searching $U$, we encounter keys that are ``merely pretty good, but not excellent''?  In that case, the danger is that applying $\mathrm{Ver}(k,\cdot)$ will damage the state substantially, but still without causing us to accept $k$.

One might hope to solve this by adjusting the definitions of ``good'' and ``bad'' keys: after all, if $k$ was good enough that it caused even an amplified verifier to have a non-negligible probability of accepting, then presumably the original, unamplified acceptance probability was at least (say) $0.7$, if not $0.8$.  The trouble is that, no matter how we set the threshold, there could be keys that fall right on the border, with a high enough acceptance probability that they damage our state, but not high enough that we reliably accept them.

To solve this problem, we use the following lemma, which builds on work by Aaronson \cite{aar:qmaqpoly}, with an important correction by Harrow and Montanaro \cite{hm16}.
\end{proof}

\begin{lemma}[``Secret Acceptor Lemma'']
\label{secretacc}
Let $M_1,\ldots,M_N$ be known $2$-outcome POVMs, and let $\rho$ be an unknown state.  Suppose we're promised that there exists an $i^{*} \in [N]$ such that
$$\Pr[M_{i^{*}}(\rho)\text{ accepts}] \ge p.$$
Then given $\rho^{\otimes r}$, where $r=O\left( \frac{\log^4 N}{\eps^2} \right)$, there's a measurement strategy to find an $i \in [N]$ such that
$$\Pr[M_{i}(\rho)\text{ accepts}] \ge p - \eps,$$
with success probability at least $1-\frac{1}{N}$.
\end{lemma}
\begin{proof}[Proof Sketch]
The key ingredient is a result called the ``Quantum OR Bound,'' which states the following:
\begin{itemize}
\item Let $M_1,\ldots,M_N$ be known $2$-outcome POVMs, and let $\rho$ be an unknown state.  Then given $\rho^{\otimes r}$, where $r=O(\log N)$, there exists a measurement strategy that accepts with overwhelming probability if there exists an $i$ such that $\Pr[M_{i}(\rho)\text{ accepts}] \ge 2/3$, and that rejects with overwhelming probability if $\Pr[M_{i}(\rho)\text{ accepts}] \le 1/3$ for all $i$.
\end{itemize}
To prove the Quantum OR Bound, the obvious approach would be to apply an amplified version of $M_i$ to $\rho^{\otimes r}$, for each $i$ in succession.  The trouble is that there might be $i$'s for which the amplified measurement accepts with probability that's neither very close to $0$ nor very close to $1$---in which case, the measurement could damage the state (compromising it for future measurements), yet without leading to an overwhelming probability of acceptance.  Intuitively, the solution should involve arguing that, if this happens, then at any rate we have $\Pr[M_{i}(\rho)\text{ accepts}] > 1/3$, and therefore it's safe to accept.  Or in other words: if we accept on the $i^{th}$ measurement, then either $M_i(\rho)$ accepts with high probability, or else our state was damaged by previous measurements---but in the latter case, some of those previous measurements (we don't know which ones) must have accepted $\rho$ with a suitably high probability.

The trouble is that, if $M_1,\ldots,M_N$ were chosen in a diabolical way, then applying them in order could slowly ``drag'' the state far away from the initial state $\rho^{\otimes r}$, without any of the $M_i$'s accepting with high probability (even though the later measurements would have accepted $\rho^{\otimes r}$ itself).  In 2006, Aaronson \cite{aar:qmaqpoly} encountered this problem in the context of proving the containment $\mathsf{QMA/qpoly} \subseteq \mathsf{PSPACE/poly}$, which we mentioned as Theorem \ref{qmaqpolythm} in Lecture 5.  In \cite{aar:qmaqpoly}, Aaronson claimed that one could get around the problem by simply choosing the $M_i$'s in {\em random} order, rather than in the predetermined (and possibly adversarial) order $M_1,\ldots,M_N$.

Unfortunately, very recently Harrow and Montanaro \cite{hm16} found a bug in Aaronson's proof of this claim.  As they point out, it remains plausible that Aaronson's random measurement strategy is sound---but if so, then a new proof is needed (something Harrow and Montanaro leave as a fascinating open problem).

In the meantime, though, Harrow and Montanaro recover the Quantum OR Bound itself, by giving a different measurement strategy for which they {\em are} able to prove soundness.  We refer the reader to their paper for details, but basically, their strategy works by first placing a control qubit in the $\ket{+}$ state, and then applying the amplified measurements to $\rho^{\otimes r}$ conditioned on the control qubit being $\ket{1}$, and periodically measuring the control qubit in the $\{ \ket{+}, \ket{-} \}$ basis to check whether the measurements have damaged the state by much.  We accept if {\em either} some measurement accepts, or the state has been found to be damaged.

Now, once we have the Quantum OR Bound, the proof of Lemma \ref{secretacc} follows by a simple application of binary search (though this increases the number of copies of $\rho$ needed by a further $\operatorname*{polylog}(N)$ factor).  More precisely: assume for simplicity that $N$ is a power of $2$.  Then we first apply the Quantum OR Bound, with suitably-chosen probability cutoffs, to the set of measurements $M_1,\ldots,M_{N/2}$.  We then apply the OR Bound separately (i.e., with fresh copies of $\rho$) to the set $M_{N/2+1},\ldots,M_N$.  With overwhelming probability, at least one of these applications will accept, in which case we learn a subset of half the $M_i$'s containing a measurement that accepts $\rho$ with high probability---without loss of generality, suppose it's $M_1,\ldots,M_{N/2}$. We then recurse, dividing $M_1,\ldots,M_{N/2}$ into two subsets of size $N/4$, and so on for $\log N$ iterations until we've isolated an $i^{*}$ such that $M_{i^{*}}$ accepts with high probability.  At each iteration we use fresh copies of $\rho$.

There's just one remaining technical issue: namely, at each iteration, we start with a promise that the region we're searching contains a measurement $M_i$ that accepts $\rho$ with probability at least $p$.  We then reduce to a smaller region, which contains a measurement that accepts $\rho$ with probability at least $p-\delta$, for some fudge factor $\delta > 0$.  Thus, if $\delta$ was constant, then we could only continue for a constant number of iterations.  The solution is simply to set $\delta := \frac{\eps}{\log N}$, where $\eps$ is the final error bound in the lemma statement.  This blows up the number of required copies of $\rho$, but only (one can check) to

$$O\left( \frac{\log^4 N}{\eps^2} \right).$$
\end{proof}

\subsection{Open Problems}

As usual, we can now ask: is there any collapse of complexity classes, such as $\mathsf{P}=\mathsf{PSPACE}$, that would make the counterfeiting algorithm of Theorem \ref{thm:sront_counterfeiter} {\em computationally} efficient, rather than merely efficient in the number of legitimate bills and queries to the bank?  We've seen that, if quantum-secure OWFs exist, then the algorithm can't be made efficient---but just like in the case of firewalls and the HH Decoding Task, we don't have a converse to that result, showing that if counterfeiting is hard then some ``standard'' cryptographic assumption must fail.

Meanwhile, the proof of Lemma \ref{secretacc} raises a fascinating open problem about quantum state tomography, and we can't resist a brief digression about it.

Let $\rho$ be an unknown state, and let $E_1,\ldots,E_N$ be some list of known two-outcome POVMs.  Suppose we wanted, not merely to find an $i^{*}$ for which $\Tr(E_{i^{*}} \rho)$ is large, but to estimate {\em all} $\Tr(E_i\rho)$'s to within some additive error $\eps > 0$, with success probability at least (say) $0.99$.  What resources are needed for this?

The task is trivial, of course, if we have $O(N \log N)$ fresh copies of $\rho$: in that case, we just apply every $E_i$ to its own $\log N$ copies.

But suppose instead that we're given only $\operatorname*{polylog} N$ copies of $\rho$.  Even in that case, it's not hard to see that the task is achievable when the $\Tr(E_i\rho)$'s are promised to be bounded away from $1/2$ (e.g., either $\Tr(E_i\rho)>2/3$ or $\Tr(E_i\rho)<1/3$ for all $i$).  For then we can simply apply the amplified $E_i$'s in succession, and use the Almost As Good As New Lemma to upper-bound the damage to $\rho^{\operatorname*{polylog} N}$.

We now raise the following open question:

\begin{question}[``The Batch Tomography Problem'']
Is the above estimation task achievable with only $\operatorname*{polylog} N$ copies of $\rho$, and {\em without} a promise on the probabilities?
\end{question}

Note that this problem doesn't specify any dependence on the Hilbert space dimension $d$, since it's conceivable that a measurement procedure exists with no dimension dependence whatsoever (as happened, for example, for the Quantum OR Bound and Lemma \ref{secretacc}).  But a nontrivial dimension-dependent result (i.e., one that didn't simply depend on full tomography of $\rho$) would also be of interest.

\section{Interactive Attacks}

To return to quantum money, we've seen that there's an inherent tradeoff between Wiesner-style and BBBW-style private-key quantum money schemes.  Whichever we choose, though, there's an immediate further problem.  Namely, a counterfeiter might be able to use repeated queries to the bank---a so-called {\em interactive attack}---to learn the quantum state of a bill.

Indeed, this could be seen as a security flaw in our original definition of private-key quantum money.  Namely, we never specified what happens {\em after} the customer (or, we might as well assume, a counterfeiter) submits a bill for verification.  Does the counterfeiter get back the damaged (that is, post-measured) bill if it's accepted, or does the bank issue a brand-new bill?  And what happens if the bill doesn't pass verification? The possibilities here vary from the bank being exceedingly na\"{i}ve, and giving the bill back to the customer even in that case, to being exceedingly strict and calling the police immediately.  As we'll see in Lecture 9, we can get schemes that are provably secure even with a maximally na\"{i}ve, bill-returning bank, but proving this requires work.

As for the original Wiesner and BBBW schemes, it turns out that both can be fully broken by an interactive attack.  To see this, let's first consider the ``na\"{i}ve bank'' scenario: the bank returns the bill even if it didn't pass the verification.  In that scenario, a simple attack on Wiesner's scheme was independently observed by Aaronson \cite{Aar09} and by Lutomirski \cite{lutomirski:attack}.

The attack works as follows: the counterfeiter starts with a single legitimate bill,

$$ \ket{\$} = \ket{\theta_1} \ket{\theta_2} \cdots \ket{\theta_n}. $$

The counterfeiter then repeatedly submits this bill to the bank for verification, swapping out the first qubit $\ket{\theta_1}$ for $\ket{0}$, $\ket{1}$, $\ket{+}$, and $\ket{-}$ in sequence.  By observing which choice of $\ket{\theta_1}$ causes the bank to accept, after $O(\log n)$ repetitions, the counterfeiter knows the correct value of $\ket{\theta_1}$ with only (say) $1/n^2$ probability of error.  Furthermore, since the bank's measurements of $\ket{\theta_2}, \ldots, \ket{\theta_n}$ are in the correct bases, none of these qubits are damaged at all by the verification process.  Next the counterfeiter repeatedly submits $\ket{\$}$ to the bank, with the now-known correct value of $\ket{\theta_1}$, but substituting $\ket{0}$, $\ket{1}$, $\ket{+}$, and $\ket{-}$ in sequence for $\ket{\theta_2}$, and so on until all $n$ qubits have been learned.

\subsection{The Elitzur-Vaidman Bomb Attack}

Of course, the above counterfeiting strategy fails if the bank adopts the simple countermeasure of calling the police whenever a bill fails verification (or perhaps, whenever too many bad bills get submitted with the same serial number)!

In 2014, however, Nagaj and Sattath \cite{NS14} cleverly adapted the attack on Wiesner's scheme, so that it works even if the bank calls the police after a single failed verification.  Their construction is based on the {\em Elitzur-Vaidman bomb} \cite{EV93}, a quantum effect that's related to Grover's algorithm but extremely counterintuitive and interesting in its own right---so let's now have a digression about the Elitzur-Vaidman bomb.

Suppose someone has a package $P$ that might or might not be a bomb. You can send to the person a bit $b\in\{0,1\}$. You always get the bit back. But if you send $b=1$ (``I believe you have a bomb in your package'') and $P$ is a bomb, then it explodes, killing everyone. The task is to learn whether or not $P$ is a bomb without actually setting it off.

This is clearly impossible classically.  For the only way to gain {\em any} information about $P$ is eventually to dare the question $b=1$, with all its consequences.

Not so quantumly!  We can model the package $P$ as an unknown $1$-qubit operator, on a qubit $\ket{b}$ that you send.  If there's no bomb, then $P$ simply acts as the identity on $\ket{b}$.  If there {\em is} a bomb, then $P$ measures $\ket{b}$ in the $\{\ket{0},\ket{1} \}$ basis.  If $P$ observes $\ket{0}$, then it returns the qubit to you, while if $P$ observes $\ket{1}$, then it sets off the bomb.

Now let
$$
R_\eps = \begin{pmatrix}\cos\eps & -\sin\eps\\ \sin\eps & \cos\eps\end{pmatrix}
$$
be a unitary that rotates $\ket{b}$ by a small angle $\eps > 0$.

Then you simply need to use the following algorithm:

\begin{itemize}
\item \textbf{Initialize} $\ket{b} := \ket{0}$
\item \textbf{Repeat $\frac{\pi}{2 \eps}$ times:}
\begin{itemize}
\item Apply $R_\eps$ to $\ket{b}$
\item Send $\ket{b}$ to $P$
\end{itemize}
\item \textbf{Measure $\ket{b}$ in the $\{\ket{0},\ket{1} \}$ basis}
\item \textbf{If the result is $\ket{0}$ then output ``bomb''; if $\ket{1}$ then output ``no bomb''}
\end{itemize}

Let's analyze this algorithm.  If there's no bomb, then the invocations of $P$ do nothing, so $\ket{b}$ simply gets rotated by an $\eps$ angle $\frac{\pi}{2 \eps}$ times, evolving from $\ket{0}$ to $\ket{1}$.  If, on the other hand, there {\em is} a bomb, then you repeatedly submit $\ket{b} = \cos\eps \ket{0} + \sin\eps \ket{1}$ to $P$, whereupon $P$ measures $\ket{b}$ in the standard basis, observing $\ket{1}$ (and hence setting off the bomb) with probability $\sin^2 \eps \approx \eps^2$.  Thus, across all $\frac{\pi}{2 \eps}$ invocations, the bomb gets set off with total probability of order $\eps^2 / \eps = \eps$, which of course can be made arbitrarily small by choosing $\eps$ small enough.  Furthermore, assuming the bomb is {\em not} set off, the qubit $\ket{b}$ ends up in the state $\ket{0}$, and hence measuring $\ket{b}$ reveals (non-destructively!) that the bomb was present.  Of course, what made the algorithm work was the ability of quantum measurements to convert an $\eps$ amplitude into an $\eps^2$ probability.

Nagaj and Sattath \cite{NS14} realized that one can use a similar effect to attack Wiesner's scheme even in the case of a suspicious bank.  Recall the attack of Aaronson \cite{Aar09} and Lutomirski \cite{lutomirski:attack}, which learned a Wiesner banknote one qubit $\ket{\theta_i}$ at a time.  We want to adapt this attack so that the counterfeiter still learns $\ket{\theta_i}$, with high probability, without ever submitting a banknote that fails verification.  In particular, that means: without ever submitting a banknote whose $i^{th}$ qubit is more than some small $\eps$ away from $\ket{\theta_i}$ itself.  (Just like in the earlier attack, we can assume that the qubits $\ket{\theta_j}$ for $j\neq i$ are all in the correct states, as we vary $\ket{\theta_i}$ in an attempt to learn a classical description of that qubit.)

The procedure by Nagaj et al., like the previous procedure, works qubit by qubit: this is okay since we get our bill back if it has been validated, and if we fail then we have more pressing issues to worry about than learning the remaining qubits.  So let's assume $n=1$: that is, we want to learn a single qubit $\ket\theta \in \{\ket 0, \ket 1, \ket +,\ket -\}$.

The algorithm is as follows: we (the counterfeiter) prepare a control qubit in the state $\ket{0}$.  We then repeat the following, $\frac{\pi}{2 \eps}$ times:

\begin{itemize}
\item Apply a $1$-qubit unitary that rotates the control qubit by $\eps$ counterclockwise.
\item Apply a CNOT gate from the control qubit to the money qubit $\ket{\theta}$.
\item Send the money qubit to the bank to be measured.
\end{itemize}

To see why this algorithm works, there are four cases to consider:

\begin{itemize}
\item If $\ket{\theta}=\ket{+}$, then applying a CNOT to $\ket{\theta}$ does nothing.  So the control qubit just rotates by an $\eps$ angle $\frac{\pi}{2 \eps}$ times, going from $\ket{0}$ to $\ket{1}$. In no case does the bank's measurement yield the outcome $\ket{-}$.
\item If $\ket{\theta}=\ket{0}$, then applying a CNOT to $\ket{\theta}$ produces the state
$$ \cos(\eps) \ket{00} + \sin(\eps) \ket{11}. $$
The bank's measurement then yields the outcome $\ket{0}$ with probability $\cos^2 (\eps) \approx 1-\eps^2$.  Assuming $\ket{0}$ is observed, the state collapses back down to $\ket{00}$, and the cycle repeats.  By the union bound, the probability that the bank {\em ever} observes the outcome $\ket{1}$ is at most
$$ \frac{\pi}{2 \eps} \sin^2(\eps) = O(\eps). $$
\item If $\ket{\theta}=\ket{1}$, then the situation is completely analogous to the case $\ket{\theta}=\ket{0}$.
\item If $\ket{\theta}=\ket{-}$, then applying a CNOT to $\ket{\alpha}$ produces the state
$$ (\cos(\eps) \ket{0} - \sin(\eps) \ket{1}) \otimes \ket{-}. $$
So at the next iteration, the control qubit gets rotated back to $\ket{0}$ (and nothing happens); then it gets rotated back to $\cos(\eps) \ket{0} - \sin(\eps) \ket{1}$, and so on, cycling through those two states.  In no case does the bank's measurement yield the outcome $\ket{+}$.
\end{itemize}

To summarize, if $\ket{\theta}=\ket{+}$ then the control qubit gets rotated to $\ket{1}$, while if $\ket{\theta}$ is any of the other three possibilities, then the control qubit remains in the state $\ket{0}$ (or close to $\ket{0}$) with overwhelming probability.  So at the end, measuring the control qubit in the $\{ \ket{0},\ket{1} \}$ basis lets us distinguish $\ket{\theta}=\ket{+}$ from the other three cases.  By symmetry, we can repeat a similar procedure for the other three states, and thereby learn $\ket{\theta}$, using $O(1/ \eps)$ trial verifications, with at most $O(\eps )$ probability of getting caught.  So by repeating the algorithm for each $\theta_i$, we learn the entire bill with $O(n / \eps)$ trial verifications and at most $O(n \eps)$ probability of getting caught.

While we explained these attacks for Wiesner's scheme, the same attacks work with only minor modifications for the BBBW scheme---or indeed, for {\em any} scheme where the banknotes consist of unentangled qubits, and a banknote is verified by projectively measuring each qubit separately.

However, having explained these attacks, we should also mention an extremely simple fix for them.  The fix is just that, rather than {\em ever} returning a bill to a customer after verification, the bank destroys the bill, and hands back a new bill (with a new serial number) of the same denomination!  If this is done, then the original security definition that we gave for private-key quantum money really does make sense.

\section{Quantum Money and Black Holes}

We end this lecture with the following observation.  As we saw above, quantum-secure OWFs are enough to obtain a private-key quantum money scheme with small secret.  It can be shown that if, moreover, the OWFs are injective, they give us something additional: namely, an ``injective'' private-key quantum money scheme, one where distinct serial numbers $k\ne k'$ map to nearly-orthogonal banknotes, $\left| \left\langle \$_k | \$_{k'} \right\rangle \right| < \eps$.  Now, we saw in Lecture 6 that injective, quantum-secure OWFs are also enough to imply that the HH Decoding Task is hard.

But is there any {\em direct} relationship between private-key quantum money and the HH Decoding Task, not mediated by OWFs?  It turns out that there is.  To see this, in Aaronson's first hardness result for the HH Decoding Task (i.e., the first construction in Theorem \ref{aarfirewall}), we simply need to replace the injective OWF output $\ket{f(x)}$ by the output $\ket{\$_k}$ of an injective private-key quantum money scheme, yielding the state
$$
\ket{\psi}_{RBH} = {1 \over \sqrt{2^{n+1}}} \sum_{k\in\{0,1\}^n}  \left( \ket{k \, 0^{p(n)-n},0}_R \ket{0}_B + \ket{\$_k,1}_R \ket{1} ) \ket{k}_H \right).
$$
Then the same argument as in the proof of Theorem \ref{aarfirewall} shows that the ability to decode entanglement between $R$ and $B$ in this state would imply the ability to pass from $\ket{\$_k}$ to $\ket{k}$---and {\em therefore}, the ability to break the private-key quantum money scheme.  So we get the following consequence, which perhaps exemplifies the strange connections in this course better than any other single statement:

\begin{theorem}
\label{whoa}
Suppose there exists a secure, private-key, injective quantum money scheme with small keys.  Then the HH Decoding Task is hard.
\end{theorem}

(Or in words: ``the ability to decode the Hawking radiation from a black hole implies the ability to counterfeit quantum money.'')

Since injective OWFs imply injective private-key quantum money with small keys, but the reverse isn't known, Theorem \ref{whoa} is stronger than Theorem \ref{aarfirewall}, in the sense that it bases the hardness of HH decoding on a weaker assumption.

\lecture{Scott Aaronson}{Alexander Razborov}{Public-Key Quantum Money}

As the last lecture made clear, one interesting question is {\em whether there's any private-key quantum money scheme where the banknote can be returned to the customer after verification, yet that's still secure against interactive attacks}.  We'll return to that question later in this lecture.

But as long as we're exploring, we might as well be even more ambitious, and ask for what Aaronson \cite{Aar09} called a {\em public-key} quantum money scheme.  By analogy to public-key cryptography, this is a quantum money scheme where absolutely anyone can verify a banknote, not just the bank that printed it---i.e., where verifying a banknote no longer requires taking it to the bank at all.  And yet, despite openly publishing a verification procedure, the bank still somehow ensures that no one can use that procedure to {\em copy} a bill efficiently.

Following Aaronson \cite{Aar09}, and Aaronson and Christiano \cite{achristiano}, we now give a more precise definition of public-key quantum money, in the style of Definition~\ref{def:private-key-quantum-money}.

\begin{definition}
A {\em public-key quantum money scheme} consists of three polynomial-time algorithms:
\begin{itemize}
\item $\text{KeyGen}(0^n)$, which outputs a pair $k = (k_{\text{pub}}, k_{\text{private}})$ (this is probabilistic, of course);

\item $\text{Bank}(k)=\$_k$ (exactly as in Definition~\ref{def:private-key-quantum-money});

\item $\text{Ver}(k_{\text{pub}}, \rho)$ (same as in Definition~\ref{def:private-key-quantum-money}, except that now the verifier has access only to the public part of the key).
\end{itemize}
\end{definition}

Completeness and soundness errors are defined exactly as before, replacing $k$ with $k_{\text{pub}}$ in appropriate places.  Completeness is defined in the worst case (with respect to the choice of $k$), whereas while defining the soundness, we average over the internal randomness of KeyGen.

\section{Public-Key Mini-Schemes}

Just like we had a notion of private-key mini-schemes, we have a corresponding notion of public-key mini-schemes.

\begin{definition} \label{def:public_mini}
A {\em public-key mini-scheme} consists of two polynomial-time algorithms:
\begin{itemize}
\item $\text{Bank}(0^n)$, which probabilistically outputs a pair $\$ = (s, \rho_s)$ (where $s$ is a classical serial number);
\item $\text{Ver}(\$)$, which accepts or rejects a claimed banknote.
\end{itemize}
\end{definition}
The scheme has completeness error $\epsilon$ if
$$
\Pr[\text{Ver}(\$)\ \text{accepts}] \geq 1-\epsilon.
$$
It has soundness error $\delta$ if for all polynomial-time counterfeiters mapping $\$ = (s,\rho_s)$ into two possibly entangled states $\sigma_1$ and $\sigma_2$, we have

$$\Pr[\text{Ver}(s,\sigma_1),\ \text{Ver}(s,\sigma_2)\text{ both accept}] < \delta.$$

Just like in the private-key case, to go from a mini-scheme to a full money scheme, one needs an additional ingredient.  In this case, the ingredient turns out to be a conventional {\em digital signature scheme} secure against quantum attacks.  We won't give a rigorous definition, but roughly: a digital signature scheme is a way of signing messages, using a private key $k_{\text{private}}$, in such a way that
\begin{itemize}
\item the signatures can be verified efficiently using a public key $k_{\text{public}}$, and yet
\item agents who only know $k_{\text{public}}$ and not $k_{\text{private}}$ can't efficiently generate {\em any} yet-unseen message together with a valid signature for that message, even after having seen valid signatures for any polynomial number of messages of their choice.
\end{itemize}

Now, given a public-key mini-scheme plus a quantum-secure signature scheme, the construction that produces a full public-key money scheme was first proposed by Lutomirski et al.\ in 2010 \cite{breaking}, with the security analysis given by Aaronson and Christiano in 2012 \cite{achristiano}.

The idea is simply to turn a mini-scheme banknote, $(s,\rho_s)$, into a full money scheme banknote by digitally signing the serial number $s$, and then tacking on the signature as an additional part of the serial number:

$$\$_k = (s, \text{sign}(k,s), \rho_s).$$

To check the resulting banknote, one first checks the digital signature $\text{sign}(k,s)$ to make sure the note $\$_k$ really came from the bank rather than an impostor.  One then checks $(s, \rho_s)$ using the verification procedure for the underlying mini-scheme.  Intuitively, then, to print more money, a counterfeiter {\em either} needs to produce new bills with new serial numbers and new signatures, thereby breaking the signature scheme, {\em or else} produce new bills with already-seen serial numbers, thereby breaking the mini-scheme.  This intuition can be made rigorous:

\begin{theorem}[Aaronson-Christiano \cite{achristiano}]
\label{standardcons}
If there exists a counterfeiter against the public-key quantum money scheme above, then there also exists a counterfeiter against either the mini-scheme or the signature scheme.
\end{theorem}

The proof of Theorem \ref{standardcons} is a relatively standard cryptographic hybrid argument, and is therefore omitted.

Now, we have many good candidates for quantum-secure signature schemes (indeed, it follows from a result of Rompel \cite{rompel} that such schemes can be constructed from any quantum-secure one-way function).  Thus, the problem of public-key quantum money can be reduced to the problem of constructing a secure public-key mini-scheme.

\section{Candidates}

How can we do that?  Many concrete mini-schemes were proposed over the past decade, but alas, the majority of them have since been broken, and those that remain seem extremely hard to analyze.  To give three examples:

\begin{itemize}
\item Aaronson \cite{Aar09}, in 2009, proposed a scheme based on stabilizer states.  This scheme was subsequently broken by Lutomirski et al.\ \cite{breaking}, by using a nontrivial algorithm for finding planted cliques in random graphs.
\item There were several schemes based on random instances of the $\mathsf{QMA}$-complete Local Hamiltonians problem (see Lecture 4).  All these schemes were then broken by Farhi et al.\ \cite{farhi:restore} in 2010, using a new technique that they called {\em single-copy tomography}.  We'll have more to say about single-copy tomography later.
\item In 2012, Farhi et al.\ \cite{knots} proposed another approach based on knot theory.  In this approach, a quantum banknote is a superposition over oriented link diagrams sharing a given value $v$ of the Alexander polynomial (a certain efficiently-computable knot invariant), with the serial number encoding $v$.  As of 2016, this scheme remains unbroken, but it remains unclear how to say anything else about its security (for example, by giving a reduction).
\end{itemize}

\section{Public-Key Quantum Money Relative to Oracles}

Given the apparent difficulty of constructing a secure public-key mini-scheme, perhaps we should start with an easier question: {\em is there even an oracle $A$ such that, if the bank, customers, and counterfeiters all have access to $A$, then public-key quantum money is possible?}  If (as we'd hope) the answer turns out to be yes, {\em then} we can work on replacing $A$ by an explicit function.

In 2009, Aaronson \cite{Aar09} showed that there's at least a {\em quantum} oracle (as defined in Lecture 5) relative to which public-key quantum money is possible.

\begin{theorem}[Aaronson 2009 \cite{Aar09}] There exists a quantum oracle $U$ relative to which a public-key mini-scheme exists.
\label{qomoney}
\end{theorem}
\begin{proof}[Proof Sketch] We start by fixing a mapping from $n$-bit serial numbers $s$ to $n$-qubit pure states $\ket{\psi_s}$, where each $\ket{\psi_s}$ is chosen independently from the Haar measure.  Then every valid banknote will have the form $(s,\ket{\psi_s})$, and verifying a claimed a banknote $(s,\rho)$ will consist of projecting $\rho$ onto the subspace spanned by $\ket{\psi_s}$.

Now, the quantum oracle $U$ is simply a reflection about the subspace spanned by valid banknotes: that is,

$$ U = I - 2\sum_{s \in \{0,1\}^n} \ket{s}\bra{s} \otimes \ket{\psi_s}\bra{\psi_s}. $$

Given this $U$, verifying a banknote is trivial: we just delegate the verification to the oracle.  (Technically, $U$ also includes a hidden component, accessible only to someone who knows the bank's secret key $k$, which maps $\ket{s}\ket{0\cdots 0}$ to $\ket{s}\ket{\psi_s}$, and which the bank uses to print new bills.)

How do we prove this scheme secure?  On the one hand, it follows immediately from the No-Cloning Theorem that a counterfeiter who had only $(s,\ket{\psi_s})$, and who lacked access to $U$, would not be able to print additional bills.  On the other hand, if a counterfeiter {\em only} had access to $U$, and lacked any legitimate banknote, it's not hard to see that the counterfeiter could produce a valid banknote $(s,\ket{\psi_s})$ using $O(2^{n/2})$ queries to $U$---but, on the other hand, that $\Omega(2^{n/2})$ queries are also necessary, by the optimality of Grover's algorithm (see Lecture 5).

The interesting part is to show that, even if the counterfeiter starts with a valid banknote $(s,\ket{\psi_s})$, the counterfeiter still needs $\Omega(2^{n/2})$ queries to $U$ to produce a {\em second} valid banknote---i.e., to produce any state that has nontrivial overlap with $\ket{\psi_s}\otimes \ket{\psi_s}$.  This is a result that Aaronson \cite{Aar09} calls the {\em Complexity-Theoretic No-Cloning Theorem} (since it combines the No-Cloning Theorem with query complexity):

\begin{itemize}
\item Given a Haar-random $n$-qubit pure state $\ket{\psi}$, together with a quantum oracle $U_{\psi}=I-2\ket{\psi}\bra{\psi}$ that reflects about $\ket{\psi}$, one still needs $\Omega(2^{n/2})$ queries to $U_{\psi}$ to produce any state that has $\Omega(1)$ fidelity with $\ket{\psi}\otimes \ket{\psi}$.
\end{itemize}

We won't prove the Complexity-Theoretic No-Cloning Theorem here, but will just sketch the main idea.  One considers all pairs of $n$-qubit pure states $\ket{\psi},\ket{\phi}$ such that (say) $\left| \left\langle \psi | \phi \right\rangle \right| = 1/2$.  One then notices that a successful cloning procedure would map these pairs to $\ket{\psi}^{\otimes 2},\ket{\phi}^{\otimes 2}$ that satisfy

$$ \left| \bra{\psi}^{\otimes 2} \ket{\phi}^{\otimes 2} \right| = \left| \left\langle \psi | \phi \right\rangle \right|^2 = \frac{1}{4}. $$

In other words, for all these pairs, the procedure needs to {\em decrease the inner product} by $1/4 = \Omega(1)$.  Given any {\em specific} pair $\ket{\psi},\ket{\phi}$, it's not hard to design a single query, to the oracles $U_{\psi},U_{\phi}$ respectively, that would decrease the inner product by $\Omega(1)$.  However, by using Ambainis's adversary method \cite{ambainis}, one can show that no query can do this for {\em most} $\ket{\psi},\ket{\phi}$ pairs, if the pairs are chosen Haar-randomly subject to $\left| \left\langle \psi | \phi \right\rangle \right| = 1/2$.  Rather, any query decreases the {\em expected} squared fidelity,
$$ \mathbb{E}_{\ket{\psi},\ket{\phi}}\left[ \left| \psi | \phi \right|^2 \right], $$
by at most $O(2^{-n/2})$, where the expectation is taken over all such pairs.  This then implies that, to decrease the squared fidelity by $\Omega(1)$ for all or even most pairs, we need to make $\Omega(2^{n/2})$ queries to $U$---i.e., just as many queries as if we were doing pure Grover search for $\ket{\psi}$, rather than starting with one copy of $\ket{\psi}$ and then merely needing to make a second one.
\end{proof}

\section{The Hidden Subspace Scheme}

Because it delegates all the work to the oracle $U$, Theorem \ref{qomoney} gives us little insight about how to construct secure public-key mini-schemes in the ``real,'' unrelativized world.
As the next step toward that goal, one naturally wonders whether public-key quantum money can be shown possible relative to a {\em classical} oracle $A$ (that is, a Boolean function accessed in quantum superposition).  This question was answered by the work of Aaronson and Christiano \cite{achristiano} in 2012.

\begin{theorem}[Aaronson-Christiano 2012 \cite{achristiano}] There exists a classical oracle relative to which public-key quantum money is possible.
\label{acthm}
\end{theorem}
\begin{proof}[Proof Sketch] The construction is based on ``hidden subspaces.''  In particular, let $S\leq \mathbb F_2^n$ be a subspace of the vector space $\mathbb F_2^n$, which is chosen uniformly at random subject to $\text{dim}(S) = n/2$.  (Thus, $S$ has exactly $2^{n/2}$ elements.)  Let
$$S^\perp = \{ x \;|\; x\cdot s\equiv 0 \left(\mathrm{mod }2\right) \forall s\in S\}$$
be $S$'s dual subspace, which also has $\text{dim}(S^\perp) = n/2$ and $\left|S\right| = 2^{n/2}$.

Then in our public-key mini-scheme, each banknote will have the form $(d_S, \ket S)$, where
\begin{itemize}
\item $d_S$ is an obfuscated classical description of the subspace $S$ and its dual subspace (which we take to be the serial number), and
\item $$\ket S =\frac {1}{\sqrt{|S|}}\sum_{x\in S} \ket{x}$$ is a uniform superposition over $S$.
\end{itemize}
When we feed it the serial number $d_S$ as a ``password,'' the classical oracle $A$ gives us access to the characteristic functions $\chi_S, \chi_{S^\perp}$ of $S$ and $S^\perp$ respectively.  Using these, we can easily realize a projection onto $\ket{S}$.  To do so, we just do the following to our money state (which is supposed to be $\ket{S}$):
\begin{itemize}
\item Apply $\chi_S$, to check that the state has support only on $S$ elements.
\item Hadamard all $n$ qubits, to map $\ket S$ to $\ket{S^\perp}$.
\item Apply $\chi_{S^\perp}$, to check that the transformed state has support only on $S^\perp$ elements.
\item Hadamard all $n$ qubits again, to return the state to $\ket S$.
\end{itemize}

As a technical point, the oracle also contains a one-way mapping $(s_1,\ldots,s_{n/2})\rightarrow d_S$, which lets the bank find the obfuscated serial number $d_S$ for a given bill $\ket S$ that it prepares, given a basis for $S$, but {\em doesn't} let a user learn the basis for $S$ given $d_S$, which would enable counterfeiting.

The key claim is this: {\em any algorithm that breaks this mini-scheme---that is, maps $\ket{d_S}\ket S$ to $\ket{d_S}{\ket S}^{\otimes 2}$---must make $\Omega(2^{n/4})$ queries to the oracles $\chi_S$ and $\chi_{S^\perp}$.}

This claim, whose proof we omit, is a somewhat stronger version of the Complexity-Theoretic No-Cloning Theorem from the proof of Theorem \ref{qomoney}.  It's stronger because, in its quest to prepare $\ket{S}^{\otimes 2}$, the counterfeiting algorithm now has access not only to a copy of $\ket{S}$ and a unitary transformation that implements $I-2\ket{S}\bra{S}$, but also the oracles $\chi_S$ and $\chi_{S^\perp}$.  Nevertheless, again by using Ambainis's quantum adversary method, it's possible to show that there's no faster way to prepare a new copy of $\ket{S}$, than simply by doing Grover search on $\chi_S$ or $\chi_{S^\perp}$ to search for a basis for $S$ or $S^\perp$---i.e., the same approach one would use if one were preparing $\ket{S}$ ``de novo,'' with no copy of $\ket{S}$ to begin with.  And we know that Grover search for a nonzero element of $S$ or $S^\perp$ requires $\Omega(2^{n/4})$ queries to $S$ or $S^\perp$ respectively.

There are two further ideas needed to complete the security proof.  First, the argument based on the Complexity-Theoretic No-Cloning Theorem implies that $\Omega(2^{n/4})$ oracle queries are needed to prepare a state that has {\em very} high fidelity with $\ket{S}^{\otimes 2}$.  But what about preparing a state that merely has non-negligible ($1/\text{poly}$) fidelity with $\ket{S}^{\otimes 2}$, which of course would already be enough to break the mini-scheme, according to our security definition?  To rule that out, Aaronson and Christiano give a way to ``amplify'' weak counterfeiters into strong ones, with only a polynomial overhead in query complexity, in any money scheme where the verification consists of a projective measurement onto a $1$-dimensional subspace.  This, in turn, relies on recent variants of amplitude amplification that converge monotonically toward the target state (in this case, $\ket{S}^{\otimes 2}$), rather than ``overshooting'' it if too many queries are made (see, for example, Tulsi, Grover, and Patel \cite{tgp}).

The second issue is that the Complexity-Theoretic No-Cloning Theorem only rules out a counterfeiter that maps $\ket{S}$ to $\ket{S}^{\otimes 2}$ for {\em all} subspaces $S$ with $\text{dim}(S) = n/2$, or at least the vast majority of them.  But again, what about a counterfeiter that only succeeds on a $1/\text{poly}$ fraction of subspaces?  To deal with this, Aaronson and Christiano show that the hidden subspace scheme has a ``random self-reducibility'' property: any counterfeiter $C$ that works on a non-negligible fraction of $S$'s can be converted into a counterfeiter that works on {\em any} $S$.  This is proven by giving an explicit random self-reduction: starting with a state $\ket{S}$, one can apply a linear transformation that maps $\ket{S}$ to some random {\em other} subspace state $\ket{T}$, while also replacing the oracles $\chi_S$ or $\chi_{S^\perp}$ by ``simulated oracles'' $\chi_T$ or $\chi_{T^\perp}$, which use $\chi_S$ and $\chi_{S^\perp}$ to recognize elements of $T$ and $T^{\perp}$ respectively.
\end{proof}

Now, once that we have a public-key quantum money scheme that's secure relative to a classical oracle, Aaronson and Christiano \cite{achristiano} observed that we get something else essentially for free: namely, a {\em private}-key quantum money scheme where bills are returned after verification, yet that's secure against interactive attack.  This solves one of the main problems left open by Lecture 8.  Let's now explain the connection.

\begin{theorem}[Aaronson-Christiano 2012 \cite{achristiano}]
There exists a private-key quantum money scheme secure against interactive attack (with no computational assumptions needed, though with the bank maintaining a huge database).
\end{theorem}
\begin{proof}
Each banknote has the form $\ket{z}\ket{S_z}$, where $\ket{z}$ is a classical serial number, and $\ket{S_z}$ is an equal superposition over $S_z$, a subspace of $\mathbb{F}_2^n$ of dimension $n/2$.  The bank, in a giant database, stores the serial number $z$ of every bill in circulation, as well as a basis for $S_z$.  When the user submits $\ket{z}\ket{S_z}$ to the bank for verification, the bank can use its knowledge of a basis for $S_z$ to decide membership in both $S_z$ and $S_z^\perp$, and thereby implement a projection onto $\ket{S_z}$.  On the other hand, suppose it were possible to map $\ket{z}\ket{S_z}$ to $\ket{z}\ket{S_z}^{\otimes 2}$, using $\poly(n)$ queries to the bank (with the submitted bills returned to the user on successful or even failed verifications).  Then by using the oracle $A$ to simulate the queries to the bank, we could also map $\ket{d_S}\ket{S}$ to $\ket{d_S}\ket{S}^{\otimes 2}$ in the scheme of Theorem \ref{acthm}, and thereby break that scheme.  But this would contradict Theorem \ref{acthm}.  It follows that mapping $\ket{z}\ket{S_z}$ to $\ket{z}\ket{S_z}^{\otimes 2}$ requires $\Omega(2^{n/2})$ queries to the bank.
\end{proof}

\section{Instantiation}

In light of Theorem \ref{acthm}, an obvious question is whether it's possible to ``instantiate'' the hidden subspace construction, thus getting rid of the oracle $A$.  This would require providing an ``obfuscated'' description of the subspaces $S$ and $S^\perp$---one that let the user efficiently apply $\chi_S$ and $\chi_{S^\perp}$ in order to verify a banknote, yet that didn't reveal a basis for $S$.

In their paper, Aaronson and Christiano \cite{achristiano} offered one suggestion for this, based on low-degree polynomials.  More concretely: assume that the serial number of a banknote $\ket{S}$ encodes two sets of (say) degree-$3$ polynomials,
$$p_1,\ldots,p_m: \mathbb F_2^n\to \mathbb F_2\text{  and  }q_1,\ldots,q_m: \mathbb F_2^n\to \mathbb F_2,$$
for some $m=O(n)$, such that $p_1,\ldots,p_m$ simultaneously vanish on $S$ (and only on $S$), while $q_1,\ldots,q_m$ simultaneously vanish on $S^\perp$ (and only on $S^\perp$).  Given $S$, it's easy to generate such polynomials efficiently---for example, by choosing random degree-$3$ polynomials that vanish on some canonical subspace $S_0$ (say, that spanned by the first $n/2$ basis vectors), and then acting on the polynomials with a linear transformation that maps $S_0$ to $S$.  For additional security, we could also add noise: an $\epsilon$ fraction of the polynomials can be chosen completely at random.

Clearly, given these polynomials, we can use them to compute the characteristic functions $\chi_S$ and $\chi_{S^\perp}$, with no oracle needed.  But given only the polynomials, there's at least no {\em obvious} way to recover a basis for $S$, so one might hope that copying $\ket{S}$ would be intractable as well.

To give evidence for that, Aaronson and Christiano \cite{achristiano} gave the following security reduction.

\begin{theorem}
\label{secreduction}
Suppose the hidden subspace mini-scheme can be broken, with the functions $\chi_S$ and $\chi_{S^\perp}$ implemented by sets of degree-$3$ polynomials $p_1,\ldots,p_m: \mathbb F_2^n\to \mathbb F_2$ and $q_1,\ldots,q_m: \mathbb F_2^n\to \mathbb F_2$ respectively.  Then there exists a polynomial-time quantum algorithm that, given the $p_i$'s and $q_i$'s, finds a basis for $S$ with success probability $\Omega(2^{-n/2})$.
\end{theorem}
\begin{proof}
The argument is surprisingly simple.  Suppose a counterfeiting algorithm $C$ exists as in the hypothesis, which maps $\ket{S}$ to $\ket{S}^{\otimes 2}$ using $p_1,\ldots,p_m$ and $q_1,\ldots,q_m$.  Then consider the following algorithm to find a basis for $S$ with $\sim 2^{-n/2}$ success probability:
\begin{enumerate}
\item[(1)] Prepare a uniform superposition over all $n$-bit strings.
\item[(2)] Compute $p_1,\ldots,p_m$ on the elements of the superposition, and condition on getting the all-$0$ outcome---thereby collapsing the superposition down to $\ket{S}$ with success probability $2^{-n/2}$.
\item[(3)] If the measurement succeeds, then run $C$ repeatedly, to map $\ket{S}$ to $\ket{S}^{\otimes m}$ for some $m=\Theta(n)$.
\item[(4)] Measure each of the copies of $\ket{S}$ in the computational basis to learn a basis for $S$.
\end{enumerate}
\end{proof}

An algorithm that finds a basis for $S$ with exponentially-small success probability, $\Omega(2^{-n/2})$, might not sound impressive or unlikely.  However, notice that if one tried to find a basis for $S$ by random guessing, one would succeed with probability only $2^{-\Omega(n^2)}$.

\section{Attacks}

Unfortunately, some recent developments have changed the picture substantially.

First, in 2014, Pena et al.\ \cite{PFP15} proved that the {\em noiseless} version of the low-degree polynomial scheme can be broken, if the scheme is defined not over $\mathbb{F}_2$ but over a finite field of odd characteristic.  Their proof used Gr\"{o}bner basis methods, and it actually yielded something stronger than a break of the money scheme: namely, a polynomial-time classical algorithm to recover a basis for $S$ given $p_1,\ldots,p_m$ (the dual polynomials, $q_1,\ldots,q_m$, aren't even needed).  In the $\mathbb{F}_2$ case, Pena et al.\ conjectured, based on numerical evidence, that their Gr\"{o}bner basis techniques would at least yield a quasipolynomial-time algorithm.

In any case, the noisy version of the low-degree polynomial scheme wasn't affected by these results, since the Gr\"{o}bner basis techniques break down in the presence of noise.

However, in a recent dramatic development, Christiano and others have fully broken the noisy low-degree polynomial scheme.  They did this by giving a quantum reduction from the problem of breaking the noisy scheme to the problem of breaking the noiseless scheme, which Pena et al.\ already solved.  (Thus, unlike Pena et al.'s result, this part of the attack is specific to quantum money.)

The reduction is extremely simple in retrospect.  Given the polynomials $p_1,\ldots,p_m$ and $q_1,\ldots,q_m$, and given a money state $\ket{S}$, our goal is to use measurements on $\ket{S}$ to decide which $p_i$'s and $q_i$'s are ``genuine'' (that is, vanish everywhere on $S$), and which ones are ``noisy'' (that is, vanish on only a $\sim 1/2$ fraction of points in $S$).

Given a polynomial (say) $p_j$, we can of course evaluate $p_j$ on $\ket{S}$ and then measure the result.  If $p_j$ is genuine, then we'll observe the outcome $0$ with certainty, while if $p_j$ is noisy, we'll observe the outcomes $0$ and $1$ with roughly equal probability.  The problem with this approach is that, if $p_j$ is noisy, then the measurement will destroy the state $\ket{S}$, making it useless for future measurements.

Here, however, we can use the remarkable idea of {\em single-copy tomography}, which Farhi et al.\ \cite{farhi:restore} introduced in 2010, in the context of breaking other public-key quantum money schemes.  The idea is this: after we make a measurement that corrupts $\ket{S}$ to some other state $\ket{\psi}$, we then use amplitude amplification (see Lecture 5) to restore $\ket{\psi}$ back to $\ket{S}$.  Now, amplitude amplification requires the ability to reflect about both the initial state $\ket{\psi}$ and the target state $\ket{S}$.  But we can do this by using the polynomials $p_1,\ldots,p_m$ and $q_1,\ldots,q_m$!  In particular, to reflect about $\ket{S}$, we negate the amplitude of each basis state iff a sufficient fraction of the $p_i$'s and $q_i$'s evaluate to $0$, while to reflect about $\ket{\psi}$, we negate the amplitude of each basis state iff a sufficient fraction of the $p_i$'s and $q_i$'s evaluate to $0$, {\em and} the polynomial $p_j$ that we just measured evaluates to whichever measurement outcome we observed ($0$ or $1$).  This approach lets us proceed through the $p_i$'s and $q_i$'s, identifying which ones are noisy, without permanent damage to the state $\ket{S}$.

Taking the contrapositive of Theorem \ref{secreduction}, a remarkable corollary of the new attack is that, given a noisy set of low-degree polynomials $p_1,\ldots,p_m$ and $q_1,\ldots,q_m$, there's a polynomial-time quantum algorithm that recovers a basis for $S$ with success probability $\Omega(2^{-n/2})$.

Again, a na\"{i}ve algorithm would succeed with probability only $2^{-\Omega(n^2)}$ (with some cleverness, one can improve this to $2^{-cn}$ for some large constant $c$).  We don't know any quantum algorithm for this problem that succeeds with probability $\Omega(1)$, nor do we know a classical algorithm that succeeds with probability $\Omega(2^{-n/2})$.  Thus, the algorithm that emerges from this attack seems to be a genuinely new quantum algorithm, which is specialized to the regime of extremely small success probabilities.  The most surprising part is that this algorithm came from an attack on a quantum money scheme; we wouldn't know how to motivate it otherwise.

\section{Future Directions}

Returning to quantum money, though, what are the prospects for evading this attack on the low-degree polynomial scheme?  All that's needed is some new, more secure way to instantiate the black-box functions $\chi_S$ and $\chi_{S^\perp}$ in Aaronson and Christiano's hidden subspace scheme.  As a sort of stopgap measure, one could instantiate $\chi_S$ and $\chi_{S^\perp}$ using a recent idea in cryptography known as {\em indistinguishability obfuscation}, or i.o.  I.o.\ has the peculiar property that, when we use it to obfuscate a given function $f$, we can generally prove that, if there's {\em any} secure way to obfuscate $f$, and if moreover i.o.\ itself is secure, then i.o.\ yields a secure obfuscated implementation of $f$.  In this case, the upshot would be that, {\em if} the characteristic functions of the subspaces $S$ and $S^\perp$ can be obfuscated at all, then i.o.\ is one way to obfuscate them!  Ironically, this gives us a strong argument in favor of using i.o.\, even while giving us little evidence that {\em any} secure implementation is possible.

Other ideas include using implementing $\chi_S$ and $\chi_{S^\perp}$, or something like them, using lattice-based cryptography, or returning to other ideas for quantum money, such as the Farhi et al.\ \cite{knots} knot-based scheme, and trying to give them more rigorous foundations.  In the meantime, there are many other wonderful open problems about public-key quantum money.  Here's a particular favorite: is public-key quantum money possible relative to a {\em random} oracle?

Besides quantum money, another fascinating challenge is {\em copy-protected quantum software}: that is, quantum states $\ket{\psi_f}$ that would be useful for evaluating a Boolean function $f$, but not for preparing additional states that could also be used to evaluate $f$.  In forthcoming work, Aaronson and Christiano show that copy-protected quantum software is ``generically'' possible relative to a classical oracle.  But it remains to give a secure ``real-world'' implementation of quantum copy-protection for general functions $f$, or even a good candidate for one.

\lecture{Adam Bouland and Luke Schaeffer}{Adam Bouland and Luke Schaeffer}{Classification of Gate Sets}

\section{Physical Universality}

For much of this course, we considered questions of the form: given a universal quantum gate set, which unitaries can we implement, and which states can we prepare, in polynomial time?  In this final lecture, we consider a different but related question: given a (possibly \emph{non}-universal) gate set, which unitaries can we implement at all?  We'll see some recent progress on this latter question, as well as numerous open problems.

First, though, we need the concept of {\em physical universality}.

\begin{definition}
A set of $k$-qubit gates $G = \{ g_1, \ldots, g_\ell \}$ is called \emph{physically universal} if there exists a positive integer $n_0 \in \N$ such that for all $n \geq n_0$, gates from $G$ densely generate $\mathrm{SU}(2^n)$, the special unitary group of dimension $2^n$ on $n$ qubits.
\end{definition}

Let's comment on the definition.  First, we focus on $\mathrm{SU}$ since global phase is irrelevant, so we might as well assume determinant $1$.  Second, note that physical universality is about implementing unitary operations (specifically, all of $\mathrm{SU}(2^n)$ for sufficiently large $n$) rather than performing computation. That is, it's possible to be computationally universal (for quantum computation) but not physically universal, for instance if we only generate $\mathrm{SO}(2^n)$, the special orthogonal group of dimension $2^n$.

Also, the $n_0$ in the above definition is so that we only have to generate all unitaries on any sufficiently large (i.e., larger than $n_0$) number of qubits. It turns out there exist, for instance, $2$-qubit gates $G$ that approximately generate all unitaries on $3$ qubits, but don't approximately generate some $2$-qubit gates, if we apply $G$ to $2$ qubits only. That is, such $G$ can approximate $I \otimes U$ but not $U$ for some $2$-qubit unitary $U$. Surely it's good enough to generate $I \otimes U$ for most practical purposes, so our definition of physical universality allows it.

Recall the Solovay-Kitaev theorem from Lecture 2 (also see \cite{solovaykitaev}).

\begin{theorem}[Solovay-Kitaev]
If a set of gates $G$ (closed under inverses) approximately generates all $n$-qubit unitaries, then we can $\eps$-approximate any unitary on $n$ qubits using $O( \exp(n) \polylog(1/\eps))$ gates from $G$.
\end{theorem}

This theorem has, for the most part, let us ignore the gate set in this course. We can usually assume our gate set satisfies the conditions of the Solovay-Kitaev theorem, and therefore generates any other $O(1)$-qubit gate we might prefer to use. Furthermore, using only $O(\polylog(1/\eps)$) gates for an approximation of accuracy $\eps$ is an excellent trade-off.

However, there are some drawbacks to the Solovay-Kitaev theorem. First, it assumes that the gate set is closed under inverses.  It's true that by the Poincar\'{e} recurrence theorem, we can find an $O(1/\eps)$ approximation in polynomial time, but this is insufficient for Solovay-Kitaev.  It's an open problem whether we can remove the requirement of exact inverses from the theorem.

Second, the Solovay-Kitaev theorem requires a universal gate set. It doesn't apply when the gate set doesn't densely generate the entire space. Indeed, there exist non-universal gate sets for which the analogue of Solovay-Kitaev is false. For example, any irrational rotation $R_\theta$ in $\mathrm{SO}(2) \subseteq \mathrm{SU}(2)$ will densely generate $\mathrm{SO}(2)$. But since the only thing we can do with $R_\theta$ is apply it repeatedly, a counting argument shows that there aren't enough different circuits to achieve $O(\polylog(1/\eps))$ accuracy approximations.

Even if we accept these drawbacks, there's another issue: namely, is it even {\em decidable} whether a given gate set is universal? To formalize the problem, let's assume that each gate is given by its unitary matrix.  The matrix entries should be computable numbers at the very least, but let's be on the safe side and assume the entries are algebraic.  It's still not obvious how to decide whether a gate set is physically universal, mainly because we lack an effective bound on $n_0$, the number of qubits needed to ``see'' physical universality.  That is, it could be that a gate set is extremely limited until we have $10^{100}$ qubits of workspace and then, for some reason, it becomes able to approximate any unitary quite easily.  Admittedly, such a situation seems unlikely, but we must have some bound on $n_0$ to prove decidability.

So, we turn to a theorem of Ivanyos \cite{Ivanyos}.

\begin{theorem}[Ivanyos '06]
If a $k$-qudit (i.e., a $d$-dimensonal generalization of a qubit) gate set doesn't densely generate all unitaries on $O(kd^{8})$ qubits, then it will never densely generate all unitaries.
\end{theorem}

For any given gate set, this theorem gives an upper bound on $n_0$. Decidability is still not trivial; we need to compute the closure of a gate set with algebraic entries on a finite number of qubits. Fortunately, this can be done, so physical universality of any gate set is decidable.

Another question: suppose we show a gate set $G$ is universal for some particular $n' \geq 2$.  Can we argue that it's universal for all $n \geq n'$?  Yes, as it turns out.  We use the fact that there are universal two-qubit gates $U$ (see Lecture 2 for examples).  Since $G$ generates $U \otimes I_{n'-2}$, and $U$ generates all unitaries, it follows that $G$ generates any unitary on $n \geq n'$ qubits.

Lastly, let's mention a result that shows that, when we discuss non-universal gate sets, we're actually dealing with a measure-0 subset of all gate sets.

\begin{theorem}
A Haar-random $2$-qubit gate is physically universal with probability $1$.
\end{theorem}

This result was shown by Deutsch, Barenco and Ekert \cite{deutsch} in 1995, and separately by Lloyd \cite{lloyd} around the same time.

\section{Difficulty of Classification}

We're interested in better criteria for deciding whether a gate set is universal, and if it's not universal, characterizing the set of unitaries it can generate. Let's take perhaps the simplest possible open case; suppose $G$ is a $2$-qubit gate set that's closed under inverses.  Again, since global phase is irrelevant, we can assume without loss of generality that every element of $G$ has determinant $1$.  Then clearly $G$ generates a subgroup $S$ of $\mathrm{SU}(4)$, the group of $4 \times 4$ complex matrices with determinant $1$, or the $4\times 4$ \emph{special unitary group}. Since our notion of generating a gate is approximate, $S$ must be closed in the topological sense (i.e., under taking limits). That is, we say we can generate $U$ if we can construct a sequence of unitaries that approaches $U$ in the limit.

Why not approach the problem of determining $S$ group-theoretically?  Alas, despite the importance of $\mathrm{SU}(n)$ to physics and mathematics, much remains unknown about its subgroups, even for small $n$. Let's survey what's known.

First, $\mathrm{SU}(2)$ is a double cover of $\mathrm{SO}(3)$, the group of rotations in three dimensions, and the latter has been understood since antiquity. Specifically, there are discrete groups corresponding to the platonic solids (the tetrahedral group, the cubic/octahedral group, and the icosahedral/dodecahedral group), the cyclic and dihedral groups of every order, and the trivial group. There's also the group of all rotations in a plane, the group of symmetries of a disk, and the entire group. This is a complete list of subgroups of $\mathrm{SU}(2)$.

For $\mathrm{SU}(3)$, the infinite subgroups were supposedly classified in 1913, but the classification wasn't completed until 2013, when Ludl fixed an error in a classification result from the '70s. Finally, most of the subgroups of $\mathrm{SU}(4)$ are known. More precisely, all of the exceptional subgroups are understood, but there are also known to exist infinite families of subgroups, and these subgroups are poorly characterized.

\section{Classification of Quantum Gates and Hamiltonians}

We'll look at three ways to make the classification problem more approachable. This section covers the first two approaches: two-level unitaries and Hamiltonians.  Then, Section \ref{REV} will discuss classical reversible gates, and the recent solution to the ``classical part'' of the quantum gate classification problem.

\subsection{Two-Level Unitaries}

Two-level unitaries ignore the usual tensor product structure of quantum gates, in favor of a direct sum structure. Ordinarily, when we apply a $1$-qubit gate $G$ to an $8$-dimensional system (i.e., the equivalent of three qubits), that means transforming the state vector by a tensor product matrix
$$
G \otimes I_4 =
\begin{pmatrix}
G & 0 & 0 & 0 \\
0 & G & 0 & 0 \\
0 & 0 & G & 0 \\
0 & 0 & 0 & G \\
\end{pmatrix}
$$
In the two-level unitary model, we instead apply $G$ to two \emph{modes} (i.e., basis states) of the system, instead of to a qubit. The result is that the matrix we apply has the form
$$
G \oplus I_6 =
\begin{pmatrix}
G & 0 \\
0 & I_6
\end{pmatrix}
$$
instead of a tensor product. Two-level unitaries are motivated by beamsplitters, and the kind of quantum computation feasible with optics.

Reck et al.\ \cite{reck} proved the following important physical universality result:

\begin{theorem}[Reck et al.\ '94]
The set of all two-level unitaries generates all of $\mathrm{SU}(n)$, for any $n\ge 2$.
\end{theorem}

This was later improved to a classification result by Bouland and Aaronson \cite{boulandaaronson14}:

\begin{theorem}[Bouland and Aaronson '14]
Let $U$ be any $2 \times 2$ matrix with determinant $-1$. Then either
\begin{enumerate}
\item $U$ is of the form $\left(\begin{smallmatrix} e^{i\theta} & 0 \\ 0 & -e^{-i\theta} \end{smallmatrix}\right)$ or $\left(\begin{smallmatrix} 0 & e^{i\theta} \\ e^{i\theta} & 0 \end{smallmatrix}\right)$ and generates only trivial unitaries,
\item $U$ is real and densely generates $\mathrm{SO}(n)$ for all $n \geq 3$, or
\item $U$ densely generates $\mathrm{SU}(n)$ for all $n \geq 3$.
\end{enumerate}
\end{theorem}

The proof relies heavily on ``brute force'' enumeration of the representations of subgroups of $\mathrm{SU}(3)$. This underlines the importance for the classification problem of understanding the special unitary group.

\subsection{Hamiltonians}

One view of quantum mechanics centers on the Schr\"{o}dinger equation,
$$
\frac{d \Psi}{dt} = -i H \Psi,
$$
where $\Psi$ is a wavefunction for the quantum system (what we've been calling the ``state''), and $H$ is a Hermitian operator called the \emph{Hamiltonian} of the system. The Hamiltonian gives the total energy of the system.

For a finite-dimensional system, with a time-independent Hamiltonian $H$, the wave function evolves over time $t$ according to the unitary transformation
$$
\Psi(t) = e^{i H t} \Psi(0)
$$
where $e^{i H t}$ is a matrix exponential. Designing a quantum computer is about controlling the Hamiltonian of a physical system to perform a sequence of gates. Many proposed designs assume the Hamiltonian can be ``switched'' between several time-independent Hamiltonians. Then, by varying how long we apply a Hamiltonian $H$, we can get any unitary gate in the set
$$
G_H = \{ e^{iHt} : t \in \R \}.
$$
As before, we assume we can apply any gate in $G_H$ to any pair of qubits, and in any order.

Now the problem is to decide whether $G_H$ is universal, for a given Hamiltonian $H$.

The main result on this problem is due to Childs et al.\ \cite{childs}:

\begin{theorem}[Childs et al., '11]
A $2$-qubit Hamiltonian is universal for $2$ qubits, with no ancillas allowed, unless (and only unless)
\begin{itemize}
\item $e^{iH}$ does not generate entanglement, or
\item $e^{iH}$ shares an eigenvector with SWAP.
\end{itemize}
\end{theorem}

The proof makes heavy use of Lie algebras. Specifically, we consider the closure of $\{ H_{ij} : i \neq j \in n \}$, where $H_{ij}$ denotes the Hamiltonian $H$ applied to qubits $i$ and $j$, under linear combinations over $\R$ and commutators $[A,B] = AB - BA$. If the closure is all Hermitian matrices, then $G_H$ is physically universal.

On an intuitive level, it's easy to see why gates which fail to generate entanglement aren't universal. Likewise, if $e^{iH}$ shares an eigenvector $\ket{\phi}$ with SWAP, then every gate in $G_H$ has $\ket{\phi}$ as an eigenvector. Since we can only apply gates from $G_H$ and SWAP, $\ket{\phi}$ will then be an eigenvector of everything we generate, which is therefore not all $2$-qubit unitary operations. Note that on $3$ qubits, this eigenvector argument no longer holds. In that case, it's open which of these operations are universal, and on how many qubits.

There's been some partial progress on the problem of classifying commuting Hamiltonians. A \emph{commuting Hamiltonian} is a Hamiltonian $H$ such that any two gates in $G_H$, applied to any pairs of qubits, commute. That is,
$$
e^{i H_{ab} t_1} e^{i H_{cd} t_2} = e^{i H_{cd} t_2} e^{i H_{ab} t_1}
$$
for all $a,b,c,d$ and $t_1, t_2 \in \R$. Alternatively, for $2 \times 2$ unitaries, this is equivalent to all gates in $G_H$ being diagonal after a change of basis (i.e., a conjugate by a single-qubit $U$ on every qubit). \footnote{Clearly diagonal matrices commute, but the converse is more involved.}

Note that if $H$ is a commuting Hamiltonian, then $G_H$ can't be universal, because any unitary we construct is diagonal in the same basis. However, a 2010 result of Bremner, Josza and Shepherd \cite{bremner} shows that a particular set of commuting Hamiltonians of this form can't be classically simulated unless the polynomial hierarchy collapses.

\begin{theorem}[Bremner, Josza, Shepherd \cite{bremner}]
Suppose you can apply any Hamiltonian which is diagonal in the $X$ basis (i.e., after conjugation by the Hadamard gate on all qubits). Using such circuits, one can sample from a probability distribution which cannot be exactly sampled in classical polynomial time unless the polynomial hierarchy collapses.
\end{theorem}

The proof uses ``postselection gadgets'': that is, gadgets built from commuting Hamiltonian gates which use postselection to simulate all of $\mathsf{PostBQP}$. One argues that, if every probability distribution generated by the commuting Hamiltonians can be sampled in $\mathsf{BPP}$ then $\mathsf{PostBQP}$ is contained in $\mathsf{PostBPP}$, and hence
$$
\mathsf{PH} \subseteq \mathsf{P}^{\mathsf{PP}} = \mathsf{P}^{\mathsf{PostBQP}} \subseteq \mathsf{P}^{\mathsf{PostBQP}} \subseteq \mathsf{\Delta}_3^{\mathsf{P}},
$$
so the polynomial hierarchy collapses.

This was recently generalized by Bouland, Man\u{c}inski and Zhang \cite{bmz}.
\begin{theorem}[Bouland, Man\u{c}inski and Zhang '16]
Any $2$-qubit commuting Hamiltonian is either classically simulable, or can perform a sampling task which can't be done in classical polynomial time unless the polynomial hierarchy collapses.
\end{theorem}

As before, the proof uses postselection gadgets built from the commuting Hamiltonian gates. There's some difficulty constructing inverses of these gadgets due to the non-unitarity of postselection.

\section{Classification of Reversible Classical Gates}
\label{sec:classical}

Suppose we wish to classify all quantum gate sets in terms of what subset of unitaries they generate. Since classical reversible gates are a subset of quantum gates, we'd first need to classify classical reversible gates. This classification was recently completed by Aaronson, Grier and Schaeffer \cite{ReversibleGates}. Before explaining their result, we'll first describe what was previously known about this problem in the simpler setting of \emph{irreversible} gates, which were classified in the early twentieth century.

\subsection{Background}

Classical circuits are usually described in terms of irreversible gates, such as AND, OR, NOT, and NAND. In introductory electrical engineering courses, one learns that the NAND gate is universal---i.e.\ NAND gates can be used to construct any function, assuming we have access to ancilla bits as input. A natural problem is to classify irreversible gate sets in terms of which sets of Boolean functions that they can generate, assuming we have access to ancilla bits that don't depend on the input.

We can easily see that certain gate sets won't be universal. For instance, the AND and OR gates are monotone, i.e.\ if two input strings $x$ and $y$ satisfy $x\leq y$ (where $\leq$ denotes bitwise less-than-or-equal), then $\mathrm{AND}(x)\leq \mathrm{AND}(y)$.  So AND and OR will only ever be able to generate monotone functions. Likewise, the XOR gate is linear over $\mathbb{F}_2$, and therefore can't generate nonlinear functions such as AND.

One interesting point is that it's possible for non-universal gate sets to be universal in an encoded sense. For instance, suppose that we represent $0$ by the string $01$ and $1$ by the string $10$. This is known as the ``dual rail'' encoding of a bit. (This same encoding is also used in linear optics). Then one can show that $\{\mathrm{AND},\mathrm{OR}\}$ is universal on this encoded space. To perform a logical NOT, one merely needs to swap the bits of the encoded bit. (In the following, we'll assume SWAP gates are included in the gate set.) To perform a logical AND, one takes the AND of the first bits of each encoded bit.  Therefore, in this construction one can perform logical NOT and AND gates, and hence achieve encoded universality.

The classification of ordinary (irreversible) Boolean gates was completed by Post \cite{PostsLattice} in the 1920s, although he only published his work in 1941. Suppose we have a set of $k$-bit irreversible gates, and we're allowed ancillas which don't depend on the input. Then there are only 7 possible subsets of functions that we can generate:
\begin{itemize}
\item All possible functions
\item All monotone functions
\item All functions computable with AND only
\item All functions computable with OR only
\item All functions computable with XOR only
\item All functions computable with NOT only
\item The empty set---i.e.\ only constant functions
\end{itemize}
The relationships among these classes are shown in Figure \ref{post}.
\begin{figure}[h]
\begin{center}
\begin{tikzpicture}[>=latex]
\tikzstyle{class}=[circle, thick, minimum size=1.2cm, text width=0.8cm, align=center, draw, font=\tiny]
\tikzstyle{all}=[class,fill=blue!20]
\tikzstyle{affine}=[class,fill=green!20]
\tikzstyle{monotone}=[class,fill=yellow!20]
\tikzstyle{none}=[class,fill=red!20]
\tikzstyle{invis}=[class,draw=none]
\matrix[row sep=0.7cm,column sep=1.1cm,ampersand replacement=\&] {
\& \node (ALL) [all] {$\top$}; \& \\
\node (XOR) [affine] {$\mathsf{XOR}$}; \& \node (MONO1) {}; \& \node (MONO2) {}; \\
\node (NOT) [affine] {$\mathsf{NOT}$}; \& \node (AND) [monotone] {$\mathsf{AND}$}; \& \node (OR) [monotone] {$\mathsf{OR}$}; \\
\& \node (NONE1) [invis] {}; \& \node (NONE2) [invis] {}; \\
};
\node (MONO) [monotone] at ($(MONO1)!0.5!(MONO2)$) {$\mathsf{MONO}$};
\node (NONE) [none] at ($(NONE1)!0.5!(NONE2)$) {$\bot$};
\path[draw,->] (ALL) edge (XOR)
(ALL) edge (MONO)
(MONO) edge (AND)
(MONO) edge (OR)
(AND) edge (NONE)
(OR) edge (NONE)
(XOR) edge (NOT)
(NOT) edge (NONE);
\end{tikzpicture}
\end{center}
\caption{Simplified Post's Lattice}
\label{post}
\end{figure}
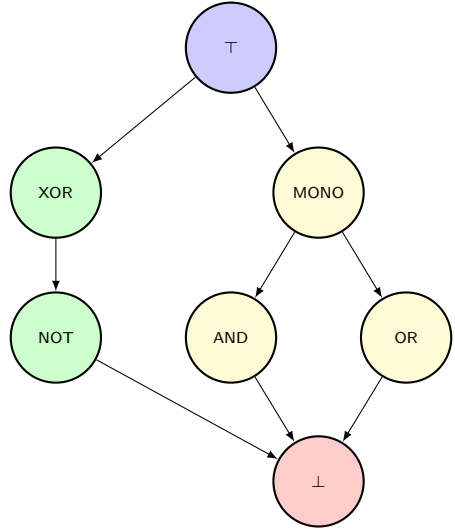

Here an arrow is drawn from gate class A to gate class B if adding any element to B (which isn't already in B) boosts the set of functions generated to at least those generated by A.  The classification is striking in its simplicity, and takes a page or so to prove (we leave it as an exercise).

We should mention, however, that in his original work, Post did {\em not} make the simplifying assumption that $0$ and $1$ ancillas are always available for free.  And when one drops that assumption, the classification (now known as {\em Post's lattice}) becomes horrendously more complicated.

\subsection{Classification of Reversible Gates} \label{REV}

Recently, Aaronson, Grier and Schaeffer \cite{ReversibleGates} completed the classification of reversible Boolean gates, analogous to the classification that we stated above for {\em irreversible} gates.  One can motivate this in two ways.  First, classical reversible gates are a subset of unitary gates, so classifying the former is a necessary step towards classifying all quantum gate sets. Second, even if we're only interested in classical computing, we might want to consider reversible gates because (in contrast to irreversible gates) they don't generate heat by dumping entropy into the environment.

In the following, we assume that the SWAP gate (i.e., relabeling of the bits) is always available for free.  We also assume that we're allowed classical ancillas (in the state 0 or 1), but that at the end of the computation, the classical ancillas must be returned to their initial states. So in particular, the final state of the ancillas can't depend on the input to the function. One reason to make this assumption is that it lets the classification extend naturally to quantum computing. In that context, requiring the ancillas to return to their initial values ensures that if the classical circuit is queried in superposition, the ancillas don't become entangled with the input to the circuit. Therefore, this classification captures which classical functions can be computed (possibly in superposition) using subsets of classical reversible gates.

Additionally, adding ancillas avoids some technical difficulties which seem artificial. For instance, if we're not allowed ancillas, then using 3-bit permutations only (i.e., any 3 bit reversible gates), it's not possible to generate all possible 4-bit permutations---because 3-bit operations generate even permutations when acting on 4 bits, so they can never generate an odd permutation of 4 bits. Allowing ancillas removes this sign issue.

Aaronson, Grier, and Schaeffer's classification is shown in Figure \ref{ReversibleFigure}. To explain the diagram, we need to define some terminology. The Toffoli gate is the Controlled-Controlled-NOT gate. The Fredkin gate is the Controlled-SWAP gate. The NOTNOT gate applies NOT to two bits simultaneously (this is weaker than the NOT gate because applying NOTNOT always preserves the parity of the input string; it doesn't give you the ability to apply a single NOT). The set $C_k$ consists of all gates that preserve Hamming weight mod $k$.
The sets $F_4$, $T_4$, and $T_6$ are exceptional classes of gates which will be defined shortly.

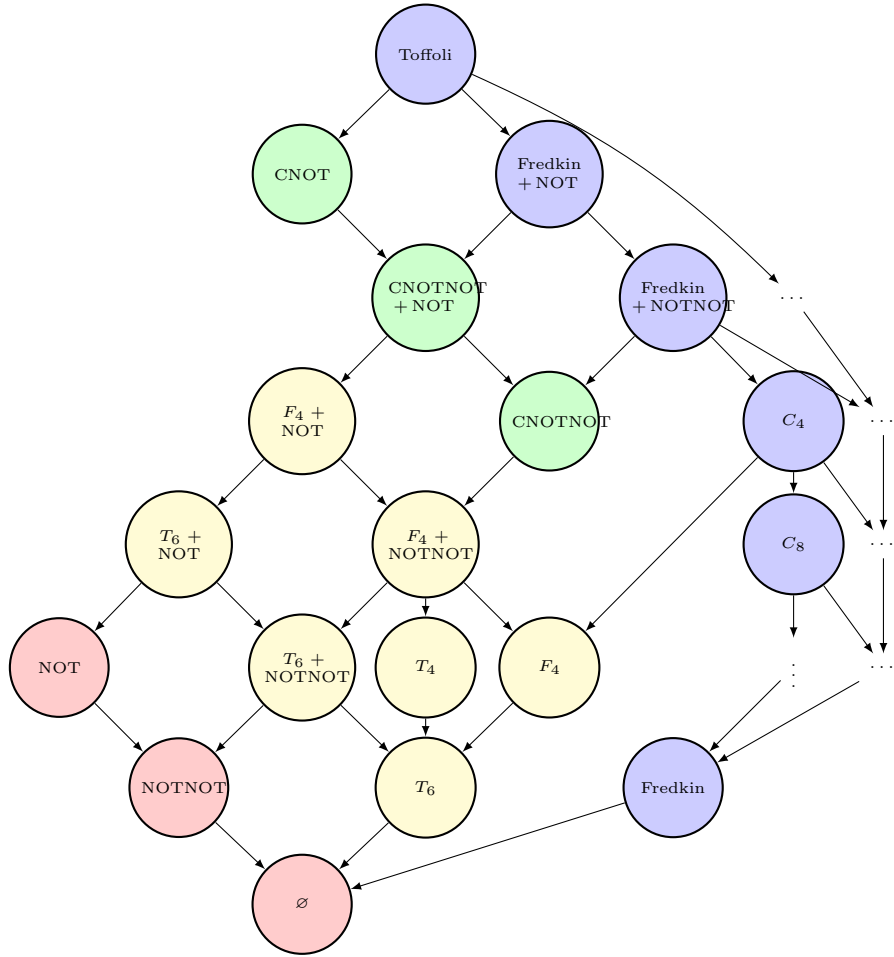
\begin{figure}[ptb]
\begin{center}
\begin{tikzpicture}[>=latex,scale=0.5]
\tikzstyle{class}=[circle, thick, minimum size=1.2cm, text width=1.0cm, align=center, draw, font=\tiny,scale=1.0]
\tikzstyle{nonaffine}=[class, fill=blue!20]
\tikzstyle{affine}=[class, fill=green!20]
\tikzstyle{orthogonal}=[class, fill=yellow!20]
\tikzstyle{inf4}=[class, fill=yellow!20] 
\tikzstyle{automorphism}=[class,fill=red!20]
\matrix[row sep=0.2cm,column sep=0.2cm,ampersand replacement=\&] {
\& \& \& \node (ALL) [nonaffine]{$\Toffoli$}; \& \& \& \& \\
\& \& \node (CNOT) [affine]{$\CNOT$}; \& \& \node (FRNOT) [nonaffine]{$\Fredkin$ \\ $+\NOT$}; \& \& \& \\
\& \& \& \node (CNOTNOTN) [affine]{$\CNOTNOT$ \\ $+\NOT$}; \& \& \node (MOD2) [nonaffine]{$\Fredkin$ \\ $+\NOTNOT$}; \& \node (MOD3) [font=\tiny]{$\cdots$}; \& \\
\& \& \node (F4NOT) [orthogonal]{$F_4+\NOT$}; \& \& \node (CNOTNOT) [affine]{$\CNOTNOT$}; \& \& \node (MOD4) [nonaffine]{$C_4$}; \& \node (MOD6) [font=\tiny]{$\cdots$}; \\
\& \node (T6NOT) [inf4]{$T_6+\NOT$}; \& \& \node (F4NOTNOT) [orthogonal]{$F_4+\NOTNOT$}; \& \& \& \node (MOD8) [nonaffine]{$C_8$}; \& \node (MOD12) [font=\tiny]{$\cdots$}; \\
\node (NOT) [automorphism]{$\NOT$}; \& \& \node (T6NOTNOT) [inf4]{$T_6+\NOTNOT$}; \& \node (T4) [orthogonal]{$T_4$}; \& \node (F4) [orthogonal]{$F_4$}; \& \& \node (MODELLIPSIS) [font=\tiny]{$\vdots$}; \& \node (MODELLIPSIS2) [font=\tiny]{$\cdots$}; \\
\& \node (NOTNOT) [automorphism]{$\NOTNOT$}; \& \& \node (T6) [inf4]{$T_6$}; \&  \& \node (FREDKIN) [nonaffine]{$\Fredkin$}; \& \& \\
\& \& \node (NONE) [automorphism]{$\varnothing$}; \& \& \& \& \\
};
\path[draw,->] (ALL) edge (FRNOT)
(FRNOT) edge (MOD2)
(ALL) edge[bend left=10] (MOD3)
(MOD2) edge (MOD4)
(MOD2) edge (MOD6)
(MOD3) edge (MOD6)
(MOD4) edge (MOD8)
(MOD4) edge (MOD12)
(MOD6) edge (MOD12)
(MOD8) edge (MODELLIPSIS)
(MOD8) edge (MODELLIPSIS2)
(MOD12) edge (MODELLIPSIS2)
(MODELLIPSIS) edge (FREDKIN)
(MODELLIPSIS2) edge (FREDKIN)
(FREDKIN) edge (NONE)
(ALL) edge (CNOT)
(CNOT) edge (CNOTNOTN)
(CNOTNOTN) edge (CNOTNOT)
(F4NOT) edge (F4NOTNOT)
(F4NOTNOT) edge (F4)
(T6NOT) edge (T6NOTNOT)
(T6NOTNOT) edge (T6)
(NOT) edge (NOTNOT)
(NOTNOT) edge (NONE)
(FRNOT) edge (CNOTNOTN)
(CNOTNOTN) edge (F4NOT)
(F4NOT) edge (T6NOT)
(T6NOT) edge (NOT)
(MOD2) edge (CNOTNOT)
(CNOTNOT) edge (F4NOTNOT)
(F4NOTNOT) edge (T6NOTNOT)
(T6NOTNOT) edge (NOTNOT)
(MOD4) edge (F4)
(F4) edge (T6)
(T6) edge (NONE)
(F4NOTNOT) edge (T4)
(T4) edge (T6);
\end{tikzpicture}
\end{center}
\caption{The inclusion lattice of reversible gate classes}%
\label{ReversibleFigure}
\end{figure}

The classification is quite involved, so let's discuss it from the top down. At the top is the Toffoli class---this set contains all possible permutations of the bit strings. The classes along the upper right of the diagram are non-linear gate sets, all of which are encoded universal, either by the dual rail encoding ($x\rightarrow x\bar{x}$), or the doubled encoding ($x\rightarrow xx$), or the doubled dual rail encoding ($x\rightarrow x\bar{x}x\bar{x}$). All of these gate sets have special properties with regard to parity. The $C_k$ gates preserve parity mod $k$. Fredkin+NOTNOT preserves parity mod 2. And Fredkin+NOT either conserves or anti-conserves parity mod 2. (I.e.\ for any element of the Fredkin+NOT class, either parity is preserved or flipped.) For the $C_k$ classes, where $k$ is prime, adding any gate outside the class bumps the power all the way to that of Toffoli (i.e.\ all permutations)---hence these classes are connected directly to the top of the diagram.

Next are the gate classes that are linear, but not orthogonal transformations of their input. These include the classes generated by CNOT, CNOTNOT (the controlled-NOTNOT gate), and CNOTNOT+NOT.  These classes are not universal. In fact, computation with these gates is contained in the complexity class $\oplus$\textsf{L} (``parity-L''), which is the power of counting the parity of the number of accepting paths of a non-deterministic logspace Turing machine.

Next are the gate classes that are orthogonal but nontrivial. These include three exceptional classes: those generated by gates that we call $F_4$, $T_6$, and $T_4$.  The $T_k$ gates take as input a $k$-bit $x$, and replace $x$ with $\bar{x}$ if the parity of the bits of $x$ is odd, and otherwise leave $x$ fixed. The $F_k$ gates take as input a $k$-bit $x$, and replace $x$ with $\bar{x}$ if the parity of bits of $x$ is even, and otherwise leave $x$ fixed.

The remaining gate classes, such as those generated by NOT or NOTNOT, are trivial.  This completes the classification.

One can prove without great difficulty that the classes in the diagram all distinct, have the containment relations shown, and are generated by the gate sets that we said.  The hard part is to prove that there are no {\em additional} classes.  Briefly, one does this by repeatedly arguing that any gate $G$ that {\em fails} to preserve the defining invariant of some class $C_1$, must therefore be usable to generate a generating set for some larger class $C_2$.  Depending on the $C_1$ and $C_2$ in question, this can require arguments involving, e.g., lattices, Diophantine equations, and Chinese remaindering.  We won't say more here, but see \cite{ReversibleGates} for details.

If we allow {\em quantum} ancillas (e.g. the state $\frac{\ket{00}+\ket{11}}{\sqrt{2}}$) in this model, then many of the classes in the diagram collapse. In particular, only six of the sets remain: Toffoli, Fredkin, CNOT, $T_4$, NOT, and the empty set. An inclusion diagram for these is provided in Figure \ref{ReversibleFigureWithQuantumAncillas}. As a result, if one were trying to classify all quantum gate sets, one would probably want to start with this simplified lattice.

\begin{figure}[h]
\begin{center}
\begin{tikzpicture}[>=latex]
\tikzstyle{class}=[circle, thick, minimum size=1.2cm, text width=0.8cm, align=center, draw, font=\tiny]
\tikzstyle{all}=[class,fill=blue!20]
\tikzstyle{affine}=[class,fill=green!20]
\tikzstyle{monotone}=[class,fill=yellow!20]
\tikzstyle{none}=[class,fill=red!20]
\tikzstyle{invis}=[class,draw=none]
\matrix[row sep=0.7cm,column sep=1.1cm,ampersand replacement=\&] {
\& \node (ALL) [all] {$\mathsf{TOFFOLI}$}; \& \\
\node (CNOT) [affine] {$\mathsf{CNOT}$}; \& \&   \\
\node (T4) [monotone] {$T_4$}; \& \& \node (FREDKIN)[all]{$\mathsf{FREDKIN}$}; \\
\node (NOT)[none]{$\mathsf{NOT}$}; \& \& \\
\& \node (NONE) [none]{$\bot$}; \&  \\
};
\path[draw,->] (ALL) edge (CNOT)
(ALL) edge (FREDKIN)
(CNOT) edge (T4)
(T4) edge (NOT)
(FREDKIN) edge (NONE)
(NOT) edge (NONE);
\end{tikzpicture}
\end{center}
\caption{Classification of Reversible Gates with Quantum Ancillas}
\label{ReversibleFigureWithQuantumAncillas}
\end{figure}
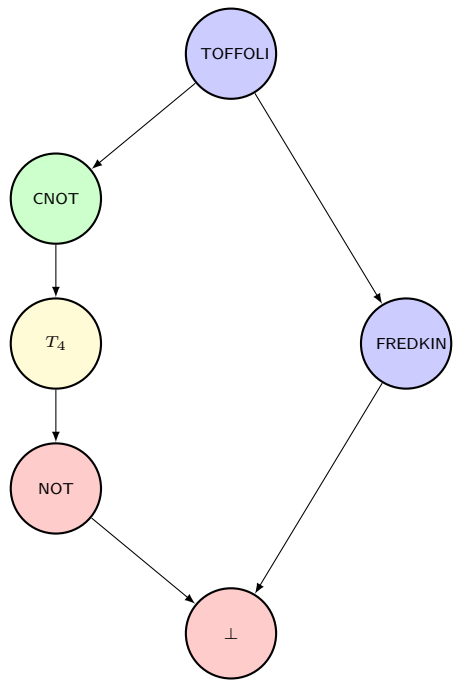

Several open problems remain. First, can we extend this classification to trits (that is, a $3$-symbol alphabet) rather than bits? Already for the case of irreversible gates, trits are much more complicated: the trit lattice is known to have uncountable cardinality, in contrast to Post's lattice, which has only countable cardinality.

More interestingly, can we classify all quantum gate sets acting on qubits, or even just (for example) all the classes generated by $2$-qubit gates?  Can we prove a dichotomy theorem, saying that every quantum gate set is either physically universal, or else fails to be universal in one of a list of prescribed ways?  Besides the classical gates, $1$-qubit gates, stabilizer gates and their conjugates, etc., does that list include anything not yet known?  Finally, what can we say about the complexity of states and unitary transformations for the non-universal classes in the list?

\phantomsection\addcontentsline{toc}{section}{References}

\bibliography{thesis}

\end{document}